\documentclass[a4paper, 11pt]{article}
\usepackage[utf8]{inputenc}
\usepackage[ngermanb, english]{babel}
\usepackage{amsmath}
\usepackage{amssymb}
\usepackage{amsthm}
\usepackage{tikz-cd}
\usepackage{bm}
\usepackage{commath}
\usepackage{mathrsfs}
\usepackage{mathtools}
\usepackage{slashed}
\usepackage{tensor}
\usepackage{parskip}
\usepackage{epstopdf}
\usepackage{graphicx}
\usepackage[nottoc]{tocbibind}
\usepackage{babelbib}
\usepackage{hyperref}
\usepackage[top=2.5cm, left=2.5cm, right=2.5cm, bottom=2.5cm]{geometry}

\theoremstyle{plain}
\newtheorem{thm}{Theorem}[section]
\newtheorem{lem}[thm]{Lemma}
\newtheorem{prop}[thm]{Proposition}
\newtheorem{col}[thm]{Corollary}

\theoremstyle{definition}
\newtheorem{defn}[thm]{Definition}
\newtheorem{exmp}[thm]{Example}

\newtheorem{con}[thm]{Convention}
\newtheorem{ass}[thm]{Assumption}
\theoremstyle{remark}
\newtheorem{rem}[thm]{Remark}

\providecommand{\sectionref}[1]{Section~\ref{#1}}

\providecommand{\ssecref}[1]{Subsection~\ref{#1}}
\providecommand{\sssecref}[1]{Subsubsection~\ref{#1}}
\providecommand{\eqnref}[1]{Equation~\eqref{#1}}

\providecommand{\eqnsaref}[2]{Equations~\eqref{#1}~and~\eqref{#2}}

\providecommand{\conref}[1]{Convention~\ref{#1}}
\providecommand{\asssnref}[3]{Assumptions~\ref{#1},~\ref{#2}~and~\ref{#3}}
\providecommand{\defnssnref}[3]{Definitions~\ref{#1},~\ref{#2}~and \ref{#3}}

\newcommand{\bbM}{\mathbbit{M}}
\newcommand{\sbbM}{\mathbbsit{M}}

\newcommand{\sctn}[1]{\Gamma \left ( #1 \right )}

\newcommand{\vectc}[1]{\mathfrak{X}_\text{c} \left ( #1 \right )}
\newcommand{\vectg}[1]{\mathfrak{X}_{[1]} \left ( #1 \right )}
\newcommand{\order}[1]{O \left ( #1 \right )}
\newcommand{\dt}[1]{\operatorname{Det} \left ( #1 \right )}
\newcommand{\tr}[1]{\operatorname{Tr} \left ( #1 \right )}
\newcommand{\deDonder}{{d \negmedspace D}}
\newcommand{\deDonderFR}{{\mathfrak{d} \negmedspace \mathfrak{D}}}
\newcommand{\deDonderpreFR}{{\mathfrak{d} \negmedspace \mathfrak{d}}}
\newcommand{\gravitonghost}{C}
\newcommand{\textfrac}[2]{#1 / #2}
\newcommand{\imaginary}{\mathrm{i}}
\newcommand{\gcoupling}{\varkappa}
\newcommand{\diff}{\operatorname{Diff}_0 \left ( M \right )}
\newcommand{\diffbbM}{\operatorname{Diff}_0 \left ( \bbM \right )}

\newcommand{\id}{\operatorname{Id}}
\newcommand{\Lie}{\pounds}
\newcommand{\met}{\gamma}
\newcommand{\trivmap}{\tau}
\newcommand{\conext}{\widetilde{M}}
\newcommand{\conmet}{\tilde{\met}}
\newcommand{\scri}{\mathscr{I}}
\newcommand{\particlefield}{\varphi}

\newcommand{\enter}{\vspace{0.5\baselineskip}}
\newcommand{\first}{\mathbf{a}}
\newcommand{\second}{\mathbf{b}}
\newcommand{\third}{\mathbf{c}}
\newcommand{\fourth}{\mathbf{d}}
\newcommand{\mtx}{\mathbf{M}}
\newcommand{\gravfr}{\mathfrak{G}}
\newcommand{\pregravfr}{\mathfrak{g}}

\newcommand{\ghostfr}{\mathfrak{C}}
\newcommand{\preghostfr}{\mathfrak{c}}

\newcommand{\gravprop}{\mathfrak{P}}
\newcommand{\ghostprop}{\mathfrak{p}}

\newcommand{\matterfrk}{\tensor[_k]{\mathfrak{M}}{}}
\newcommand{\prematterfrk}{\tensor[_k]{\mathfrak{m}}{}}
\newcommand{\prematterfri}{\tensor[_1]{\mathfrak{m}}{}}
\newcommand{\prematterfrii}{\tensor[_2]{\mathfrak{m}}{}}
\newcommand{\prematterfriii}{\tensor[_3]{\mathfrak{m}}{}}
\newcommand{\prematterfriv}{\tensor[_4]{\mathfrak{m}}{}}
\newcommand{\prematterfrv}{\tensor[_5]{\mathfrak{m}}{}}
\newcommand{\prematterfrvi}{\tensor[_6]{\mathfrak{m}}{}}
\newcommand{\prematterfrvii}{\tensor[_7]{\mathfrak{m}}{}}
\newcommand{\prematterfrviii}{\tensor[_8]{\mathfrak{m}}{}}
\newcommand{\prematterfrix}{\tensor[_9]{\mathfrak{m}}{}}
\newcommand{\prematterfrx}{\tensor[_{10}]{\mathfrak{m}}{}}
\newcommand{\triplevert}{\vert\kern-0.25ex\vert\kern-0.25ex\vert}

\newsavebox{\foobox}
\newcommand{\italicbox}[2][.25]
{%
	\mbox
	{%
		\sbox{\foobox}{#2}%
		\hskip\wd\foobox
		\pdfsave
		\pdfsetmatrix{1 0 #1 1}%
		\llap{\usebox{\foobox}}%
		\pdfrestore
	}%
}

\newcommand{\mathbbit}[1]{{\mspace{-1mu} \italicbox{$\mathbb{#1}$} \mspace{2mu}}}
\newcommand{\mathbbsit}[1]{{\mspace{-1mu} \italicbox{$\scriptstyle{\mathbb{#1}}$} \mspace{1.5mu}}}

\makeatletter
\newcommand{\subalign}[1]{%
	\vcenter{%
		\Let@ \restore@math@cr \default@tag
		\baselineskip\fontdimen10 \scriptfont\tw@
		\advance\baselineskip\fontdimen12 \scriptfont\tw@
		\lineskip\thr@@\fontdimen8 \scriptfont\thr@@
		\lineskiplimit\lineskip
		\ialign{\hfil$\m@th\scriptstyle##$&$\m@th\scriptstyle{}##$\crcr
			#1\crcr
		}%
	}
}
\makeatother

\title{\textsc{Gravity-Matter Feynman Rules for any Valence}}
\author{David Prinz\footnote{Department of Mathematics and Department of Physics at Humboldt University of Berlin and Max Planck Institute for Gravitational Physics (Albert Einstein Institute) in Potsdam-Golm, Germany; prinz@\{math.hu-berlin.de, physik.hu-berlin.de, aei.mpg.de\}}}
\date{November 4, 2021}

\begin{document}

\maketitle

\begin{abstract}
	This article derives and presents the Feynman rules for (effective) Quantum General Relativity coupled to the Standard Model for any vertex valence and with general gauge parameter \(\zeta\). The results are worked out for the metric decomposition \(g_{\mu \nu} = \eta_{\mu \nu} + \gcoupling h_{\mu \nu}\), a linearized de Donder gauge fixing and four dimensions of spacetime. To this end, we calculate the Feynman rules for gravitons, graviton-ghosts and for the couplings of gravitons to scalars, spinors, gauge bosons and gauge ghosts.
\end{abstract}

\section{Introduction}

The attempt to perturbatively quantize General Relativity (GR) is rather old: In fact, the approach to define the graviton field \(h_{\mu \nu}\) with gravitational coupling constant \(\gcoupling\) as the fluctuation around a fixed background metric \(b_{\mu \nu}\), i.e.\
\begin{equation}
h_{\mu \nu} := \frac{1}{\gcoupling} \left ( g_{\mu \nu} - b_{\mu \nu} \right ) \iff g_{\mu \nu} \equiv b_{\mu \nu} + \gcoupling h_{\mu \nu} \, ,
\end{equation}
--- oftentimes, and in particular in this article, chosen as the Minkowski metric \(b_{\mu \nu} := \eta_{\mu \nu}\) --- goes back to M. Fierz, W. Pauli and L. Rosenfeld in the 1930s \cite{Rovelli}. Then, R. Feynman \cite{Feynman_Hatfield_Morinigo_Wagner} and B. DeWitt \cite{DeWitt_I,DeWitt_II,DeWitt_III,DeWitt_IV} started to calculate the corresponding Feynman rules in the 1960s. However, D. Boulware, S. Deser and P. van Nieuwenhuizen \cite{Boulware_Deser_Nieuwenhuizen}, G. 't Hooft \cite{tHooft_QG} and M. Veltman \cite{Veltman} discovered serious problems in the perturbative expansion due to the non-renormalizability of Quantum General Relativity (QGR) in the 1970s. We refer to \cite{Rovelli} for a historical treatment.

Despite its age, it is still very hard to find references properly deriving and displaying Feynman rules for QGR, given via the Lagrange density
\begin{subequations}
\begin{align} \label{eqn:QGR_Lagrange_density_Introduction}
	\mathcal{L}_\text{QGR} & := \mathcal{L}_\text{GR} + \mathcal{L}_\text{GF} + \mathcal{L}_\text{Ghost}
	\intertext{with}
	\mathcal{L}_\text{GR} & := - \frac{1}{2 \gcoupling^2} R \dif V_g \, , \\
	\mathcal{L}_\text{GF} & := - \frac{1}{4 \gcoupling^2 \zeta}  \eta^{\mu \nu} \deDonder^{(1)}_\mu \deDonder^{(1)}_\nu \dif V_\eta
	\intertext{and}
	\mathcal{L}_\text{Ghost} & := - \frac{1}{2 \zeta} \eta^{\rho \sigma} \overline{C}^\mu \left ( \partial_\rho \partial_\sigma C_\mu \right ) \dif V_\eta - \frac{1}{2} \eta^{\rho \sigma} \overline{C}^\mu \left ( \partial_\mu \big ( \tensor{\Gamma}{^\nu _\rho _\sigma} C_\nu \big ) - 2 \partial_\rho \big ( \tensor{\Gamma}{^\nu _\mu _\sigma} C_\nu \big ) \right ) \dif V_\eta \, ,
\end{align}
\end{subequations}
where \(R := g^{\nu \sigma} \tensor{R}{^\mu _\sigma _\mu _\nu}\) is the Ricci scalar and \(\deDonder^{(1)}_\mu := \eta^{\rho \sigma} \Gamma_{\rho \sigma \mu}\) is the linearized de Donder gauge fixing functional. Additionally, \(\gravitonghost \in \Gamma \left ( \bbM, \Pi \left ( T \bbM \right ) \right )\) and \(\overline{\gravitonghost} \in \Gamma \left ( \bbM, \Pi \left ( T^* \bbM \right ) \right )\) are the graviton-ghost and graviton-antighost, respectively. Finally, \(\dif V_g := \sqrt{- \dt{g}} \dif t \wedge \dif x \wedge \dif y \wedge \dif z\) and \(\dif V_\eta := \dif t \wedge \dif x \wedge \dif y \wedge \dif z\) are the Riemannian and Minkowskian volume forms, respectively. We refer to \sectionref{sec:conventions_and_definitions} and \cite{Prinz_2} for a detailed introduction. The existing literature known to the author, \cite{Veltman,Sannan,Donoghue,Choi_Shim_Song,tHooft_PQG,Hamber,Schuster,Rodigast_Schuster_1,Rodigast_Schuster_2}, limits the vertex Feynman rules to valence five, directly sets the de Donder gauge fixing parameter to \(\zeta := 1\) and omits the ghost vertex Feynman rules completely. This article aims to fix this gap in the literature by deriving the Feynman rules for gravitons, their ghosts and for their interactions with matter from the Standard Model: The analysis is carried out for the metric decomposition \(g_{\mu \nu} = \eta_{\mu \nu} + \gcoupling h_{\mu \nu}\), arbitrary vertex valence, a linearized de Donder gauge fixing with general gauge parameter \(\zeta\) and in four dimensions of spacetime. Moreover, the gravitational interactions with matter from the Standard Model are then classified into 10 different types and their vertex Feynman rules are also derived and presented for any valence.

The main results are \thmref{thm:grav-fr} stating the graviton vertex Feynman rules, \thmref{thm:grav-prop} stating the corresponding graviton propagator Feynman rule, \thmref{thm:ghost-fr} stating the graviton-ghost vertex Feynman rules and \thmref{thm:ghost-prop} stating the corresponding graviton-ghost propagator Feynman rule. Additionally, the graviton-matter vertex Feynman rules are stated in \thmref{thm:matter-fr} on the level of 10 generic matter-model Lagrange densities, as classified by \lemref{lem:matter-model-Lagrange-densities}. The complete graviton-matter Feynman rules can then be obtained by adding the corresponding matter contributions, as listed e.g.\ in \cite{Romao_Silva}. Finally, we display the three- and four-valent graviton and graviton-ghost vertex Feynman rules explicitly in \exref{exmp:FR}.

General Relativity and Quantum Theory are both fundamental theories in modern physics. While some of their predictions agree with outstanding precision with the corresponding experimental data, there are still regimes where both theories break down conceptually. Notably, this is the case with models of the big bang or in the inside of black holes. In these situations, both theories are needed simultaneously to capture the entire physical reality: General Relativity is needed in order to describe the huge masses and energies that are involved and Quantum Theory is needed in order to describe the interactions of the respective particles in these very small spatial dimensions. Unfortunately, a combined theory of Quantum Gravitation has not been found yet: While, given the success of the Standard Model, a perturbative quantization seems to be the canonical choice, it comes with several problems, most notably its non-renormalizability. This fact has lead to various, more radical approaches to Quantum Gravity, such as Supergravity, String Theory or Loop Quantum Gravity. While any of these theories fixes conceptual problems of the perturbative approach, they create additional problems elsewhere due to further assumptions. Therefore, in this article, we go back to the foundation of Quantum General Relativity via its (effective) perturbative approach using Feynman rules. Feynman rules are calculated from the Lagrangian by extracting the potentials for all classically allowed interactions. Then, scattering amplitudes are calculated by applying the fundamental principle of Quantum Theory, namely that the sum over all unobserved intermediate states needs to be considered. This leads to the Feynman diagram expansion, where each non-tree Feynman diagram corresponds to a Feynman integral over the unobserved momenta of the virtual particles. We refer to \cite{Weinberg_1,Weinberg_2,Weinberg_3} for a more detailed treatment and to \cite{Jimenez} for the corresponding treatment of supersymmetric theories.

This research is intended as the starting point for several related approaches to the perturbative renormalization of (effective) Quantum General Relativity: It is generally possible to render any Feynman integral finite by applying an appropriate subtraction for each divergent (sub-) integral. This treatment of (sub-)divergences has been studied extensively in the Hopf algebra approach of Connes and Kreimer: Here, the subdivergences are treated via the corresponding coproduct \cite{Kreimer_Hopf_Algebra} and the renormalized Feynman rules are then obtained via an algebraic Birkhoff decomposition \cite{Connes_Kreimer_0}. Then, this reasoning was soon applied to gauge theories \cite{Kreimer_Anatomy}, which lead to the identification of Ward--Takahashi and Slavnov--Taylor identities with Hopf ideals in the corresponding Connes--Kreimer renormalization Hopf algebra \cite{vSuijlekom_QED,vSuijlekom_QCD,vSuijlekom_BV,Prinz_3}. Following this route, it was then suggested by Kreimer to apply this duality to General Relativity \cite{Kreimer_QG1}, which was motivated via a scalar toy model \cite{Kreimer_vSuijlekom} and then studied in detail by the author \cite{Prinz_2,Prinz_3}. In this approach, the non-renormalizability of General Relativity manifests itself by the necessity to introduce infinitely many counterterms. The aim is now to relate these counterterms by generalized Slavnov--Taylor identities, which correspond to the diffeomorphism invariance of the theory. A first step in this direction is the construction of tree-level cancellation identities, which requires the longitudinal and transversal decomposition of the graviton propagator via the general gauge parameter \(\zeta\) as variable. This approach was supported by recent calculations for the metric density decomposition of Goldberg and Capper et al.\ (\cite{Goldberg,Capper_Leibbrandt_Ramon-Medrano,Capper_Medrano,Capper_Namazie}) up to valence six \cite{Kissler}. With the present work, we provide a foundation for a purely combinatorial argument, which will be valid for all vertex valences. This will be studied in future work via the diffeomorphism-gauge BRST double complex \cite{Prinz_5} and the longitudinal and transversal structure of the gravitational Feynman rules \cite{Prinz_6}. Additionally, we remark that this reasoning is implicit in the construction of Kreimer's Corolla polynomial \cite{Kreimer_Yeats,Kreimer_Sars_vSuijlekom,Kreimer_Corolla}. This graph polynomial, which is based on half-edges, relates the amplitudes of Quantum Yang--Mills theories to the amplitudes of the scalar \(\phi^3_4\)-theory, by means of the parametric representation of Feynman integrals \cite{Sars}. More precisely, in this approach the cancellation identities are encoded into amplitudes by means of Feynman graph cohomology \cite{Berghoff_Knispel}. In particular, this approach has been successfully generalized to spontaneously broken gauge theories and thus to the complete bosonic part of the Standard Model \cite{Prinz_1}. The possibility to apply this construction also to (effective) Quantum General Relativity will be checked in future work. Finally, we believe that the results of this article are also of intellectual interest, as Feynman rules are an essential ingredient to perturbative Quantum Field Theories.

We remark the development of more concise formulations, aimed in particular for practical calculations: There are the KLT relations \cite{Kawai_Lewellen_Tye,Bern_Carrasco_Johansson_1,Bern_Carrasco_Johansson_2,Elvang_Huang}, which relate on-shell gravitational amplitudes with the amplitudes of the `double-copy' of a gauge theory, and are applied e.g.\ in \cite{Bern_et-al}. Furthermore, it is also possible to simplify the gravitational Feynman rules by a reformulation with different (possibly auxiliary) fields \cite{Cheung_Remmen,Tomboulis}, even on a de Sitter background \cite{Tsamis_Woodard}. Moreover, we remark the use of computer algebra programs, such as `XACT' \cite{Abreu_et-al} and `QGRAF' \cite{Blumlein_et-al}. For the projects mentioned in the previous paragraph, however, the original Feynman rules are needed to arbitrary vertex valence and with general gauge parameter \(\zeta\): This is because the KLT relations are only valid on-shell and thus rely on Cutkosky's Theorem \cite{Cutkosky}, cf.\ \cite{Bloch_Kreimer,Kreimer_Cutkosky}. Also, we want to study the direct relationship between combinatorial Green's functions and their counterterms, which becomes more complicated in the aforementioned reformulations with auxiliary fields. And finally, we are interested in a combinatorial proof that is valid to all vertex valences and thus excludes the use of computer algebra programs.

\section{Conventions and definitions} \label{sec:conventions_and_definitions}

We start this article with our conventions, in particular the used sign choices. Additionally, we recall important definitions for (effective) Quantum General Relativity and the Standard Model. This includes the Lagrange densities with the metric decomposition, the de Donder gauge fixing and the corresponding ghosts. Furthermore, we provide a proper definition of the graviton field and in particular of the background Minkowski spacetime. This is obtained with the rather restrictive assumption of \defnref{defn:simple_spacetime}, which we call `simple spacetime'. This setup is motivated with classical results from Ellis--Hawking in \propref{prop:ase-parallelizable} and Geroch in \propref{prop:parallelizable-spin} and the boundedness assumption from \assref{ass:bdns_gf}. Finally, we comment on the diffeomorphism invariance and display the action of the corresponding diffeomorphism BRST operators. We refer to \cite{Prinz_2} for a more fundamental introduction to (effective) Quantum General Relativity coupled to Quantum Electrodynamics. Additionally, we refer to \cite{Prinz_5} for a study of the diffeomorphism-gauge BRST double complex, to \cite{Prinz_6} for a study of transversality with respect to infinitesimal diffeomorphisms and to \cite{Prinz_7} for a generalization of Wigner's elementary particle classification to Linearized General Relativity.

\enter

\begin{con}[Sign choices] \label{con:sign_choices}
	We use the sign-convention \((-++)\), as classified by \cite{Misner_Thorne_Wheeler}, i.e.:
	\begin{enumerate}
		\item Minkowski metric: \(\eta_{\mu \nu} = \begin{pmatrix} 1 & 0 & 0 & 0 \\ 0 & - 1 & 0 & 0 \\ 0 & 0 & - 1 & 0 \\ 0 & 0 & 0 & - 1 \end{pmatrix}_{\mu \nu}\)
		\item Riemann tensor: \(\tensor{R}{^\rho _\sigma _\mu _\nu} = \partial_\mu \Gamma^\rho_{\nu \sigma} - \partial_\nu \Gamma^\rho_{\mu \sigma} + \Gamma^\rho_{\mu \lambda} \Gamma^\lambda_{\nu \sigma} - \Gamma^\rho_{\nu \lambda} \Gamma^\lambda_{\mu \sigma}\)
		\item Einstein field equations: \(G_{\mu \nu} = \kappa T_{\mu \nu}\)
	\end{enumerate}
	Additionally we use the plus-signed Clifford relation, i.e.\ \(\left \{ \gamma_m , \gamma_n \right \} = 2 \eta_{m n} \id_{\Sigma \sbbM}\), cf.\ \cite[Remark 2.15]{Prinz_2}.
\end{con}

\enter

\begin{defn}[Spacetime] \label{def:spacetime}
	Let \((M,\met)\) be a Lorentzian manifold. We call \((M,\met)\) a spacetime, if it is smooth, connected, 4-dimensional and time-orientable.\footnote{We fix the spacetime-dimension, as the gravitational Feynman rules depend directly on it.}
\end{defn}

\enter

\begin{defn}[Asymptotically simple (and empty) spacetime]
	Let \((M,\met)\) be an oriented and causal spacetime. We call \((M,\met)\) an asymptotically simple spacetime, if it admits a conformal extension \(\big ( \conext,\conmet \big )\) in the sense of Penrose \cite{Penrose_1,Penrose_2,Penrose_3,Penrose_4}: That is, if there exists an embedding \(\iota \colon M \hookrightarrow \conext\) and a smooth function \(\varsigma \in C^\infty \big ( \conext \big )\), such that:
	\begin{enumerate}
		\item \(\conext\) is a manifold with interior \(\iota \left ( M \right )\) and boundary \(\scri\), i.e.\ \(\conext \cong \iota \left ( M \right ) \sqcup \scri\)
		\item \(\eval[2]{\varsigma}_{\iota \left ( M \right )} > 0\), \(\eval[2]{\varsigma}_{\scri} \equiv 0\) and \(\eval[2]{\dif \varsigma}_{\scri} \not \equiv 0\); additionally \(\iota_* \met \equiv \varsigma^2 \conmet\)
		\item Each null geodesic of \(\big ( \conext,\conmet \big )\) has two distinct endpoints on \(\scri\)
	\end{enumerate}
	We call \((M,\met)\) an asymptotically simple and empty spacetime, if additionally:\footnote{This condition can be modified to allow electromagnetic radiation near \(\scri\). We remark that asymptotically simple and empty spacetimes are also called asymptotically flat.}
	\begin{enumerate}
		\setcounter{enumi}{3}
		\item \(\eval[2]{\left ( R_{\mu \nu} \right )}_{\iota^{-1} ( \widetilde{O} )} \equiv 0\), where \(\widetilde{O} \subset \conext\) is an open neighborhood of \(\scri \subset \conext\)
	\end{enumerate}
\end{defn}

\enter

\begin{prop} \label{prop:ase-parallelizable}
	Let \((M,\met)\) be an asymptotically simple and empty spacetime. Then \((M,\met)\) is globally hyperbolic and thus parallelizable.
\end{prop}

\begin{proof}
	The first part of the statement, i.e.\ that \((M,\met)\) is globally hyperbolic, is a classical result due to Ellis and Hawking \cite[Proposition 6.9.2]{Hawking_Ellis}. We conclude the second part, i.e.\ that \((M,\met)\) is parallelizable, by noting that we have additionally assumed spacetimes to be four-dimensional: Thus, being globally hyperbolic, there is a well-defined three-dimensional space-submanifold, which therefore is parallelizable as it is orientable by assumption.
\end{proof}

\enter

\begin{col}
	Any asymptotically simple and empty spacetime \((M,\met)\) is spin.
\end{col}

\begin{proof}
	This follows immediately from \propref{prop:ase-parallelizable}, as parallelizable manifolds are trivially spin.
\end{proof}

\enter

\begin{prop} \label{prop:parallelizable-spin}
	A spacetime \((M,\met)\) is spin if and only if it is globally hyperbolic. Equivalently, \((M,\met)\) is spin if and only if it is parallelizable.
\end{prop}

\begin{proof}
	These are two classical results by Geroch \cite{Geroch_1,Geroch_2}.
\end{proof}

\enter

\begin{defn}[Simple spacetime] \label{defn:simple_spacetime}
	Let \((M,\met)\) be a spacetime. We call the triple \((M,\met,\trivmap)\) a simple spacetime, if \(M\) is diffeomorphic to the Minkowski spacetime \(\bbM\) and \(\trivmap \colon M \to \bbM\) is a fixed such diffeomorphism (not necessarily an isometry), called trivializing map. Furthermore, we use \(\trivmap\) to pushforward the metric \(\met\) to the Minkowski spacetime \(\bbM\) via \(g := \trivmap_* \met \in \Gamma \big ( \bbM, \operatorname{Sym}^2 T^* \bbM \big )\) to obtain an equivalence between the physical spacetime \((M,\met)\) and its background Minkowski spacetime \((\bbM,g)\).
\end{defn}

\enter

\begin{ass} \label{ass:simple_spacetime}
	From now on, we assume spacetimes to be simple.
\end{ass}

\enter

\begin{rem} \label{rem:Minkowski_background}
	The rather restrictive setup of \assref{ass:simple_spacetime} is motivated by \propref{prop:ase-parallelizable} and \propref{prop:parallelizable-spin}: It is physically reasonable to consider asymptotically simple and empty spacetimes, as well as to demand a spin structure for fermionic particles. Thus, the spacetime \((M,\met)\) has the same asymptotic structure as the Minkowski spacetime \((\bbM,\eta)\) and is furthermore parallelizable. This implies that it is diffeomorphic to the Minkowski spacetime of the same dimension, modulo possible singularities. However, as we need the eigenvalues of the metric \(g\) to be bounded by \assref{ass:bdns_gf} for the following constructions, we exclude singularities in our setup. This assumption allows us to view particle fields, in particular the graviton field, as sections over Minkowski spacetime \(\sctn{\bbM,E}\), where \(\pi_E \colon E \to \bbM\) is a suitable vector bundle for the particle fields under consideration, cf.\ \defnref{defn:correspondence_Minkowski_spacetime}. In turn, this enables us to use Wigner's classification of elementary particles via irreducible representations of the Poincar\'{e} group \cite{Wigner}, which will be studied in \cite{Prinz_7}. Thus we can proceed as usual by constructing the Fock space to describe the quantum states of our corresponding Quantum Field Theory. Finally, this setup provides a well-defined Fourier transformation for particle sections, cf.\ \defnref{defn:fourier_transform}.
\end{rem}

\enter

\begin{defn}[Metric decomposition and graviton field] \label{defn:md_and_gf}
	Let \((M,\met,\trivmap)\) be a simple spacetime. Then we use the following metric decomposition on the background Minkowski spacetime \((\bbM,\eta)\)
	\begin{equation} \label{eqn:metric_decomposition}
		h_{\mu \nu} := \frac{1}{\gcoupling} \left ( g_{\mu \nu} - \eta_{\mu \nu} \right ) \iff g_{\mu \nu} \equiv \eta_{\mu \nu} + \gcoupling h_{\mu \nu} \, ,
	\end{equation}
	where \(\gcoupling := \sqrt{\kappa}\) is the graviton coupling constant (with \(\kappa := 8 \pi G\) the Einstein gravitational constant). Thus, the graviton field \(h_{\mu \nu}\) is given as a rescaled, symmetric \((0,2)\)-tensor field on the background Minkowski spacetime, i.e.\ as the section \(\gcoupling h \in \Gamma \big ( \bbM, \operatorname{Sym}^2 T^* \bbM \big )\).
\end{defn}

\enter

\begin{rem}
	Given the situation of \defnref{defn:md_and_gf}, the graviton field \(h\) depends directly on the choice of the trivializing map \(\trivmap\). It can be shown, however, that this dependence can be absorbed, if the theory is diffeomorphism-invariant \cite{Prinz_7}. Thus, this construction is in particular well-defined for Linearized General Relativity.
\end{rem}

\enter

\begin{ass} \label{ass:bdns_gf}
	Given the metric decomposition from \defnref{defn:md_and_gf}, we assume the following boundedness condition for the gravitational constant \(\gcoupling\) and the graviton field \(h_{\mu \nu}\):
	\begin{equation}
		\left | \gcoupling \right | \left \| h \right \|_{\max} := \left | \gcoupling \right | \max_{\lambda \in \operatorname{EW} \left ( h \right )} \left | \lambda \right | < 1 \, ,
	\end{equation}
	where \(\operatorname{EW} \left ( h \right )\) denotes the set of eigenvalues of \(h\). This will be relevant for preceding assertions to assure the convergence of series involving the graviton coupling constant \(\gcoupling\) and the graviton field \(h_{\mu \nu}\).
\end{ass}

\enter

\begin{defn}[Correspondence to Minkowski spacetime] \label{defn:correspondence_Minkowski_spacetime}
	Let \((M,\met,\trivmap)\) be a simple spacetime and \(\pi_E \colon E \to M\) a vector bundle for particle fields. Then we extend the vector bundle for particle fields via \((\trivmap \circ \pi_E) \colon E \to \bbM\) to a vector bundle over the background Minkowski spacetime \(\bbM\).
\end{defn}

\enter

\begin{con}[Lagrange density] \label{con:Lagrange_density}
	We choose the following signs and prefactors for the Lagrange density, which we consider as a functional for sections over the background Minkowski spacetime \(\bbM\) and where \(\dif V_g := \sqrt{- \dt{g}} \dif t \wedge \dif x \wedge \dif y \wedge \dif z\) and \(\dif V_\eta := \dif t \wedge \dif x \wedge \dif y \wedge \dif z\) denote the Riemannian and Minkowskian volume forms, respectively:
	\begin{enumerate}
		\item Einstein-Hilbert Lagrange density: \begin{equation} \mathcal{L}_\text{GR} := - \frac{1}{2 \gcoupling^2} R \dif V_g \, , \end{equation} with \(R := g^{\nu \sigma} \tensor{R}{^\mu _\sigma _\mu _\nu}\)
		\item Gauge fixing Lagrange density: \begin{equation} \mathcal{L}_\text{GF} := - \frac{1}{4 \gcoupling^2 \zeta}  \eta^{\mu \nu} \deDonder^{(1)}_\mu \deDonder^{(1)}_\nu \dif V_\eta \, , \end{equation} with \(\deDonder^{(1)}_\mu := \eta^{\rho \sigma} \Gamma_{\rho \sigma \mu} \equiv \gcoupling \eta^{\rho \sigma} \left ( \partial_\rho h_{\mu \sigma} - \frac{1}{2} \partial_\mu h_{\rho \sigma} \right )\)
		\item Ghost Lagrange density: \begin{equation} \begin{split} \mathcal{L}_\text{Ghost} & := - \frac{1}{2 \zeta} \eta^{\rho \sigma} \overline{C}^\mu \left ( \partial_\rho \partial_\sigma C_\mu \right ) \dif V_\eta \\ & \phantom{:=} - \frac{1}{2} \eta^{\rho \sigma} \overline{C}^\mu \left ( \partial_\mu \big ( \tensor{\Gamma}{^\nu _\rho _\sigma} C_\nu \big ) - 2 \partial_\rho \big ( \tensor{\Gamma}{^\nu _\mu _\sigma} C_\nu \big ) \right ) \dif V_\eta \, , \end{split} \end{equation} with \(\gravitonghost \in \Gamma \left ( \bbM, \Pi \left ( T \bbM \right ) \right )\) and \(\overline{\gravitonghost} \in \Gamma \left ( \bbM, \Pi \left ( T^* \bbM \right ) \right )\)
	\end{enumerate}
	The Lagrange density of (effective) Quantum General Relativity is then the sum of the three, i.e.\
	\begin{equation} \label{eqn:QGR_Lagrange_density}
		\begin{split}
		\mathcal{L}_\text{QGR} & := \mathcal{L}_\text{GR} + \mathcal{L}_\text{GF} + \mathcal{L}_\text{Ghost} \\
		& \phantom{:} \equiv - \frac{1}{2 \gcoupling^2} \left ( \sqrt{- \dt{g}} R + \frac{1}{2 \zeta}  \eta^{\mu \nu} \deDonder^{(1)}_\mu \deDonder^{(1)}_\nu \right ) \dif V_\eta \\
		& \phantom{:=} - \frac{1}{2} \eta^{\rho \sigma} \left ( \frac{1}{\zeta} \overline{C}^\mu \left ( \partial_\rho \partial_\sigma C_\mu \right ) + \overline{C}^\mu \left ( \partial_\mu \big ( \tensor{\Gamma}{^\nu _\rho _\sigma} C_\nu \big ) - 2 \partial_\rho \big ( \tensor{\Gamma}{^\nu _\mu _\sigma} C_\nu \big ) \right ) \right ) \dif V_\eta \, ,
		\end{split}
	\end{equation}
cf.\ \cite[Section 2.2]{Prinz_2}. We remark that the ghost Lagrange density is calculated via Faddeev-Popov's method \cite{Faddeev_Popov}, cf.\ \cite[Subsection 2.2.3]{Prinz_2}, which can be embedded into the more elaborate settings of BRST cohomology and BV formalism.
\end{con}

\enter

\begin{rem}
	The reason for the sign choices from \conref{con:Lagrange_density} are as follows: The minus sign for the Einstein-Hilbert Lagrange density is due to the sign choice for the Minkowski metric, cf.\ \conref{con:sign_choices}. Then, the minus sign for the gauge fixing Lagrange density is such that \(\zeta = 1\) corresponds to the de Donder gauge fixing. Finally, the sign for the ghost Lagrange density is, as usual, an arbitrary choice, and is chosen such that all Lagrange densities have the same sign.
\end{rem}

\enter

\begin{rem}
	Given the situation of \assref{ass:simple_spacetime}, the gravitational path integral then corresponds to an integral over the space of symmetric \((0,2)\)-tensor fields over the background Minkowski spacetime \(\bbM\). As the construction of such integral measures over function spaces is rather troublesome, we simply consider the \(\hbar \ll 0\) limit, where the Feynman graph expansion can be interpreted as a `perturbative definition' of the path integral. We refer to \cite{Hamber} for a more physical treatment.
\end{rem}

\enter

\begin{ass} \label{ass:diffeo_homotopic_to_identity}
	We assume from now on that diffeomorphisms are homotopic to the identity, i.e.\ \(\phi \in \diff\).
\end{ass}

\enter

\begin{rem}
	\assref{ass:diffeo_homotopic_to_identity} is motivated by the fact that diffeomorphisms homotopic to the identity are generated via the flows of compactly supported vector fields, \(X \in \vectc{M}\), and differ from the identity only on compactly supported domains. Thus, diffeomorphisms homotopic to the identity preserve the asymptotic structure of spacetimes. We remark that, different from finite dimensional Lie groups, the Lie exponential map
	\begin{align}
		\operatorname{exp} \, & : \quad \vectc{M} \to \diff
		\intertext{is no longer locally surjective, which leads to the notion of an evolution map}
		\operatorname{Evol} \, & : \quad C^\infty \left ( [0,1], \vectc{M} \right ) \to C^\infty \left ( [0,1], \diff \right )
	\end{align}
	that maps smooth curves in the Lie algebra to smooth curves in the corresponding Lie group. We refer to \cite{Schmeding} for further details.
\end{rem}

\enter

\begin{defn}[Transformation under (infinitesimal) diffeomorphisms] \label{defn:transformation_diffeo}
	Given the situation of \defnref{defn:md_and_gf} and \assref{ass:diffeo_homotopic_to_identity}, we define the action of diffeomorphisms \(\phi \in \diff\) on the graviton field via
	\begin{align}
		\left ( \trivmap \circ \phi \right )_* \left ( \gcoupling h \right ) & := \left ( \trivmap \circ \phi \right )_* g \, ,
		\intertext{such that the background Minkowski metric can be conveniently defined to be invariant, i.e.}
		\left ( \trivmap \circ \phi \right )_* \eta & := 0 \, ,
	\end{align}
	and on the other particle fields \(\particlefield \in \sctn{\bbM,E}\) as usual, i.e.\ via
	\begin{equation}
		\left ( \trivmap \circ \phi \right )_* \particlefield \, .
	\end{equation}
	In particular, the action of infinitesimal diffeomorphisms is given via the Lie derivative with respect the generating vector field \(X \in \vectc{\bbM}\), i.e.\
	\begin{align}
		\delta_X h_{\mu \nu} & \equiv \frac{1}{\gcoupling} \left ( \nabla^{(g)}_\mu X_\nu + \nabla^{(g)}_\nu X_\mu \right ) \, , \\
		\delta_X \eta_{\mu \nu} & \equiv 0
		\intertext{and}
		\delta_X \particlefield & \equiv \Lie_X \particlefield \, ,
	\end{align}
	where \(\nabla^{(g)}\) denotes the covariant derivative with respect to the connection \(\Gamma\) induced via \(g\) on \(\bbM\), i.e.\ via
	\begin{equation}
		\tensor{\Gamma}{_\mu _\nu ^\rho} := \frac{1}{2} g^{\rho \sigma} \left ( \partial_\mu g_{\sigma \nu} + \partial_\nu g_{\mu \sigma} - \partial_\sigma g_{\mu \nu} \right ) \, .
	\end{equation}
\end{defn}

\enter

\begin{rem} \label{rem:Lie_groupoid_and_algebroid}
	Using the setup from \asssnref{ass:simple_spacetime}{ass:bdns_gf}{ass:diffeo_homotopic_to_identity} and \defnssnref{defn:md_and_gf}{defn:correspondence_Minkowski_spacetime}{defn:transformation_diffeo}, we can view Linearized General Relativity coupled to matter from the Standard Model as a `generalized gauge theory' on the background Minkowski spacetime: The `gauge group' is then given via the pushforward of diffeomorphisms homotopic to the identity by the trivializing map, i.e.\ \(\mathcal{G} := \trivmap_* \diff \cong \diffbbM\). Furthermore, their infinitesimal actions are given via Lie derivatives with respect to compactly supported vector fields \(\vectc{\bbM}\). In particular, the right setting to study the gauge theoretic properties of such a theory is given via the Lie groupoid \(\left ( \mathcal{G} \times \mathcal{B} \right ) \rightrightarrows \mathcal{B}\) over the background Minkowski spacetime-matter bundle \(\mathcal{B} := \bbM \times E\). Additionally, the action of infinitesimal diffeomorphisms is embedded into this picture via the corresponding Lie algebroid \(((\vectc{\bbM} \times \mathcal{B}) \to \mathcal{B}, [\cdot, \cdot ], \rho)\): More precisely, \([\cdot, \cdot ]\) is the Lie bracket on \(\vectc{\bbM}\) and \(\rho \colon (\vectc{\bbM} \times \mathcal{B}) \to T \mathcal{B}\) the anchor map. Then, as in the case of `ordinary gauge theories' --- that is gauge theories coming from a principle bundle structure --- the invariance of the theory under diffeomorphisms provides an obstacle for the calculation of the propagator. We solve this issue by introducing a linearized de Donder gauge fixing together with the corresponding ghost and antighost fields, \(\gravitonghost \in \sctn{\bbM, \Pi \left ( T \bbM \right )}\) and \(\overline{\gravitonghost} \in \sctn{\bbM, \Pi \left ( T^* \bbM \right )}\), respectively. This viewpoint will be elaborated in \cite{Prinz_7}.
\end{rem}

\enter

\begin{lem}
	Given the situation of \defnref{defn:transformation_diffeo}, the diffeomorphism BRST operator \(P \in \vectg{\mathcal{B}}\), i.e.\ a vector field on the spacetime-matter bundle with ghost degree 1, can be consistently defined as follows:
	\begin{subequations}
	\begin{align}
		P h_{\mu \nu} & = \frac{1}{\gcoupling} \left ( \nabla^{(g)}_\mu \gravitonghost_\nu + \nabla^{(g)}_\nu \gravitonghost_\mu \right ) \\
		P \gravitonghost^\sigma & = \gravitonghost^\rho \left ( \partial_\rho \gravitonghost^\sigma \right ) \\
		P \overline{\gravitonghost}_\sigma & = B_\sigma \\
		P B_\sigma & = 0 \\
		P \eta_{\mu \nu} & = 0 \\
		P \particlefield & = \Lie_\gravitonghost \particlefield \, ,
	\end{align}
	\end{subequations}
	where \(\gravitonghost \in \sctn{\bbM, \Pi \left ( T \bbM \right )}\) is the graviton-ghost, \(\overline{\gravitonghost} \in \sctn{\bbM, \Pi \left ( T^* \bbM \right )}\) the graviton-antighost, \(B \in \sctn{\bbM, T^* \bbM}\) the Lautrup-Nakanishi auxiliary field and again \(\particlefield \in \sctn{\bbM,E}\) represents any other particle field.
\end{lem}

\begin{proof}
	This follows directly from the transformation properties of \defnref{defn:transformation_diffeo}, as BRST operators are defined to induce infinitesimal transformations with respect to the corresponding ghost field, together with the nilpotency condition, i.e.\ the two properties
	\begin{subequations}
	\begin{align}
		P \particlefield := \delta_\gravitonghost \particlefield
		\intertext{and}
		P^2 \equiv 0 \, .
	\end{align}
	\end{subequations}
\end{proof}

\enter

\begin{rem} \label{rem:diffeo_invariance_EH-Lagrange_density}
	Contrary to Yang--Mills Lagrange densities, which are strictly invariant under gauge transformations, the Einstein-Hilbert Lagrange density is not invariant under general diffeomorphisms as it is a tensor density of weight \(1\). More precisely, the action of an infinitesimal diffeomorphism adds a total derivative to the Einstein-Hilbert Lagrange density, if the corresponding vector field is not Killing.
\end{rem}

\enter

\begin{defn}[Fourier transformation] \label{defn:fourier_transform}
	Let \((M,\met,\trivmap)\) be a simple spacetime with background Minkowski spacetime \((\bbM,\eta)\). Using the correspondence from \defnref{defn:correspondence_Minkowski_spacetime}, we define the Fourier transformation for particle fields, i.e.\ sections \(\particlefield \in \sctn{\bbM,E}\), as usual:
\begin{equation}
	\mathscr{F} \, : \quad \Gamma \big ( \bbM, E \big ) \to \widehat{\Gamma} \big ( \bbM, E \big ) \, , \quad \particlefield \left ( x^\alpha \right ) \mapsto \hat{\particlefield} \left ( p^\alpha \right ) := \frac{1}{\left ( 2 \pi \right )^2} \int_\sbbM \particlefield \big ( y^\beta \big ) e^{-i \eta_{\beta \gamma} y^\beta p^\gamma} \dif V_\eta
\end{equation}
\end{defn}

\section{Expansion of the Lagrange density} \label{sec:expansion_lagrange_density}

Given the Quantum General Relativity Lagrange density
\begin{equation}
	\begin{split}
	\mathcal{L}_\text{QGR} & = - \frac{1}{2 \gcoupling^2} \left ( \sqrt{- \dt{g}} R + \frac{1}{2 \zeta}  \eta^{\mu \nu} \deDonder^{(1)}_\mu \deDonder^{(1)}_\nu \right ) \dif V_\eta \\
	& \phantom{:=} - \frac{1}{2} \eta^{\rho \sigma} \left ( \frac{1}{\zeta} \overline{C}^\mu \left ( \partial_\rho \partial_\sigma C_\mu \right ) + \overline{C}^\mu \left ( \partial_\mu \big ( \tensor{\Gamma}{^\nu _\rho _\sigma} C_\nu \big ) - 2 \partial_\rho \big ( \tensor{\Gamma}{^\nu _\mu _\sigma} C_\nu \big ) \right ) \right ) \dif V_\eta
	\end{split}
\end{equation}
from \conref{con:Lagrange_density}. In order to calculate the corresponding Feynman rules, we decompose \(\mathcal{L}_\text{QGR}\) with respect to its powers in the gravitational coupling constant \(\varkappa\) and the ghost field \(C\) as follows\footnote{We omit the term \(\mathcal{L}_\text{QGR}^{-1,0}\) as it is given by a total derivative.}
\begin{equation}
	\mathcal{L}_\text{QGR} \equiv \sum_{m = 0}^\infty \sum_{n = 0}^1 \mathcal{L}_\text{QGR}^{m,n} \, ,
\end{equation}
where we have set \(\mathcal{L}_\text{QGR}^{m,n} := \eval[1]{\left ( \mathcal{L}_\text{QGR} \right )}_{O (\varkappa^m C^n)}\). Given \(m \in \mathbb{N}_+\), the restricted Lagrange densities \(\mathcal{L}_\text{QGR}^{m,0}\) correspond to the potential terms for the interaction of \(\left ( m + 2 \right )\) gravitons and the restricted Lagrange densities \(\mathcal{L}_\text{QGR}^{m,1}\) correspond to the potential terms for the interaction of \(m\) gravitons with a graviton-ghost and graviton-antighost, while the terms \(m = 0\) and \(n \in \set{0,1}\) provide the kinetic terms for the graviton and graviton-ghost, respectively. The situation for the matter-model Lagrange densities from \lemref{lem:matter-model-Lagrange-densities} is then analogous.\footnote{The shift in \(m\) comes from the prefactor \(\textfrac{1}{\gcoupling^2}\) in \(\mathcal{L}_{\text{QGR}}\) and is convenient, because then propagators are of order \(\gcoupling^0\) and three-valent vertices of order \(\gcoupling^1\), etc.}

\enter

\begin{lem}[Inverse metric as Neumann series in the graviton field] \label{lem:inverse_metric_series}
	Given the metric decomposition from \defnref{defn:md_and_gf} and the boundedness condition from \assref{ass:bdns_gf}, the inverse metric is given via the Neumann series
	\begin{equation}
	g^{\mu \nu} = \sum_{k = 0}^\infty \left ( - \gcoupling \right )^k \left ( h^k \right )^{\mu \nu} \, ,
	\end{equation}
	where
	\begin{subequations}
		\begin{align}
		h^{\mu \nu} & := \eta^{\mu \rho} \eta^{\nu \sigma} h_{\rho \sigma} \, ,\\
		\left ( h^0 \vphantom{h^k} \right )^{\mu \nu} & := \eta^{\mu \nu}\\
		\intertext{and}
		\left ( h^k \right )^{\mu \nu} & := \underbrace{h^\mu_{\kappa_1} h^{\kappa_1}_{\kappa_2} \cdots h^{\kappa_{k-1} \nu}}_{\text{\(k\)-times}} \, , \; k \in \mathbb{N} \, .
		\end{align}
	\end{subequations}
\end{lem}

\begin{proof}
	We calculate
	\begin{equation}
	\begin{split}
		g_{\mu \nu} g^{\nu \rho} & = \left ( \eta_{\mu \nu} + \gcoupling h_{\mu \nu} \right ) \left ( \sum_{k = 0}^\infty \left ( - \gcoupling \right )^k \left ( h^k \right )^{\nu \rho} \right ) \\
		& = \eta_{\mu \nu} \eta^{\nu \rho} + \eta_{\mu \nu} \left ( \sum_{i = 1}^\infty \left ( - \gcoupling \right )^i \left ( h^i \right )^{\nu \rho} \right ) + \gcoupling h_{\mu \nu} \left ( \sum_{j = 0}^\infty \left ( - \gcoupling \right )^j \left ( h^j \right )^{\nu \rho} \right ) \\
		& = \delta_\mu^\rho - \gcoupling h_{\mu \nu} \left ( \sum_{i = 0}^\infty \left ( - \gcoupling \right )^i \left ( h^i \right )^{\nu \rho} \right ) + \gcoupling h_{\mu \nu} \left ( \sum_{j = 0}^\infty \left ( - \gcoupling \right )^j \left ( h^j \right )^{\nu \rho} \right ) \\
		& = \delta_\mu^\rho \, ,
	\end{split}
	\end{equation}
	as requested. Finally, we remark that the Neumann series
	\begin{equation}
		g^{\mu \nu} = \sum_{k = 0}^\infty \left ( - \gcoupling \right )^k \left ( h^k \right )^{\mu \nu}
	\end{equation}
	converges precisely for
	\begin{equation}
		\left | \gcoupling \right | \left \| h \right \|_{\max} := \left | \gcoupling \right | \max_{\lambda \in \operatorname{EW} \left ( h \right )} \left | \lambda \right | < 1 \, ,
	\end{equation}
	where \(\operatorname{EW} \left ( h \right )\) denotes the set of eigenvalues of \(h\), as stated.
\end{proof}

\enter

\begin{lem}[Vielbein and inverse vielbein as series in the graviton field] \label{lem:vielbeins_series}
	Given the metric decomposition from \defnref{defn:md_and_gf} and the boundedness condition from \assref{ass:bdns_gf}, the vielbein and inverse vielbein are given via the series
	\begin{subequations}
		\begin{align}
		e_\mu^m = \sum_{k = 0}^\infty \gcoupling^k \binom{\frac{1}{2}}{k} \left ( h^k \right )_\mu^m \, ,
		\intertext{with \(h_\mu^m := \eta^{m \nu} h_{\mu \nu}\), and}
		e_m^\mu = \sum_{k = 0}^\infty \gcoupling^k \binom{- \frac{1}{2}}{k} \left ( h^k \right )_m^\mu \, ,
		\end{align}
	\end{subequations}
	with \(h^\mu_m := \eta^{\mu \nu} \delta^\rho_m h_{\nu \rho}\).
\end{lem}

\begin{proof}
	We recall the defining equations for vielbeins and inverse vielbeins,
	\begin{align}
		g_{\mu \nu} & = \eta_{m n} e^m_\mu e^n_\nu
		\intertext{and}
		\eta_{m n} & = g_{\mu \nu} e^\mu_m e^\nu_n \, ,
	\end{align}
	cf.\ \cite[Definition 2.8]{Prinz_2}. Thus, we calculate
	\begin{equation}
	\begin{split}
		g_{\mu \nu} & = \eta_{m n} e^m_\mu e^n_\nu \\
		& = \eta_{m n} \left ( \sum_{i = 0}^\infty \gcoupling^i \binom{\frac{1}{2}}{i} \left ( h^i \right )_\mu^m \right ) \left ( \sum_{j = 0}^\infty \gcoupling^j \binom{\frac{1}{2}}{j} \left ( h^j \right )_\nu^n \right ) \\
		& = \sum_{i = 0}^\infty \sum_{j = 0}^\infty \gcoupling^{i + j} \binom{\frac{1}{2}}{i} \binom{\frac{1}{2}}{j} \left ( h^{i + j} \right )_{\mu \nu} \\
		& = \sum_{k = 0}^\infty \gcoupling^k \binom{1}{k} \left ( h^{k} \right )_{\mu \nu} \\
		& = \eta_{\mu \nu} + \gcoupling h_{\mu \nu} \, ,
	\end{split}
	\end{equation}
	where we have used Vandermonde's identity, and
	\begin{equation}
	\begin{split}
		g^{\mu \nu} & = \eta^{m n} e^\mu_m e^\nu_n \\
		& = \eta^{m n} \left ( \sum_{i = 0}^\infty \gcoupling^i \binom{- \frac{1}{2}}{i} \left ( h^i \right )_m^\mu \right ) \left ( \sum_{j = 0}^\infty \gcoupling^j \binom{- \frac{1}{2}}{j} \left ( h^j \right )_n^\nu \right ) \\
		& = \sum_{i = 0}^\infty \sum_{j = 0}^\infty \gcoupling^{i + j} \binom{- \frac{1}{2}}{i} \binom{- \frac{1}{2}}{j} \left ( h^{i + j} \right )^{\mu \nu} \\
		& = \sum_{k = 0}^\infty \gcoupling^k \binom{- 1}{k} \left ( h^k \right )^{\mu \nu} \\
		& = \sum_{k = 0}^\infty \left ( - \gcoupling \right )^k \left ( h^k \right )^{\mu \nu} \\
		& = g^{\mu \nu} \, ,
	\end{split}
	\end{equation}
	where we have again used Vandermonde's identity, the identity \(\binom{-1}{k} = \left ( -1 \right )^k\) and \lemref{lem:inverse_metric_series}. Finally, the series for the vielbein and inverse vielbein field converge precisely for
	\begin{equation}
		\left | \gcoupling \right | \left \| h \right \|_{\max} := \left | \gcoupling \right | \max_{\lambda \in \operatorname{EW} \left ( h \right )} \left | \lambda \right | < 1 \, ,
	\end{equation}
	where \(\operatorname{EW} \left ( h \right )\) denotes the set of eigenvalues of \(h\), as stated.
\end{proof}

\enter

\begin{prop}[Ricci scalar for the Levi-Civita connection, cf.\ \cite{Prinz_2}] \label{prop:ricci_scalar_for_the_levi_civita_connection}
	Using the Levi-Civita connection, the Ricci scalar is given via partial derivatives of the metric and its inverse as follows:
	\begin{equation} \label{eqn:ricci_scalar_metric}
	\begin{split}
		R & = g^{\mu \rho} g^{\nu \sigma} \left ( \partial_\mu \partial_\nu g_{\rho \sigma} - \partial_\mu \partial_\rho g_{\nu \sigma} \right ) \\
		& \hphantom{ = } + g^{\mu \rho} g^{\nu \sigma} g^{\kappa \lambda} \left ( \left ( \partial_\mu g_{\kappa \lambda} \right ) \left ( \partial_\nu g_{\rho \sigma} - \frac{1}{4} \partial_\rho g_{\nu \sigma} \right ) + \left ( \partial_\nu g_{\rho \kappa} \right ) \left ( \frac{3}{4} \partial_\sigma g_{\mu \lambda} - \frac{1}{2} \partial_\mu g_{\sigma \lambda} \right ) \right . \\
		& \hphantom{ = + g^{\mu \rho} g^{\nu \sigma} g^{\kappa \lambda} ( } \left . \vphantom{\left ( \frac{1}{2} \right )} - \left ( \partial_\mu g_{\rho \kappa} \right ) \left ( \partial_\nu g_{\sigma \lambda} \right ) \right )
	\end{split}
	\end{equation}
	Furthermore, we also consider the decomposition
	\begin{subequations}
	\begin{align}
		\begin{split}
		R & \equiv g^{\nu \sigma} \left ( \partial_\mu \Gamma^\mu_{\nu \sigma} - \partial_\nu \Gamma^\mu_{\mu \sigma} + \Gamma^\mu_{\mu \kappa} \Gamma^\kappa_{\nu \sigma} - \Gamma^\mu_{\nu \kappa} \Gamma^\kappa_{\mu \sigma} \right ) \\
		& =: R^{\partial \Gamma} + R^{\Gamma^2}
		\end{split}
		\intertext{with}
		R^{\partial \Gamma} & := g^{\nu \sigma} \left ( \partial_\mu \Gamma^\mu_{\nu \sigma} - \partial_\nu \Gamma^\mu_{\mu \sigma} \right ) \\
		\intertext{and}
		R^{\Gamma^2} & := g^{\nu \sigma} \left ( \Gamma^\mu_{\mu \kappa} \Gamma^\kappa_{\nu \sigma} - \Gamma^\mu_{\nu \kappa} \Gamma^\kappa_{\mu \sigma} \right ) \, .
	\end{align}
	\end{subequations}
	Then we obtain:
	\begin{equation}
	\begin{split}
		R^{\partial \Gamma} & = g^{\mu \rho} g^{\nu \sigma} \left ( \partial_\mu \partial_\nu g_{\rho \sigma} - \partial_\mu \partial_\rho g_{\nu \sigma} \right ) \\
		& \hphantom{ = } + g^{\mu \rho} g^{\nu \sigma} g^{\kappa \lambda} \left ( \left ( \partial_\mu g_{\rho \kappa} \right ) \left ( \frac{1}{2} \partial_\lambda g_{\nu \sigma} - \partial_\nu g_{\lambda \sigma} \right ) + \frac{1}{2} \left ( \partial_\nu g_{\mu \kappa} \right ) \left ( \partial_\sigma g_{\rho \lambda} \right ) \right )
	\end{split}
	\end{equation}
	and
	\begin{equation}
	\begin{split}
		R^{\Gamma^2} & = g^{\mu \rho} g^{\nu \sigma} g^{\kappa \lambda} \left ( \left ( \partial_\kappa g_{\mu \rho} \right ) \left ( \frac{1}{2} \partial_\nu g_{\sigma \lambda} - \frac{1}{4} \partial_\lambda g_{\nu \sigma} \right ) - \left ( \partial_\nu g_{\mu \kappa} \right ) \left ( \frac{1}{2} \partial_\rho g_{\sigma \lambda} - \frac{1}{4} \partial_\sigma g_{\rho \lambda} \right ) \right )
	\end{split}
	\end{equation}
\end{prop}

\begin{proof}
	The claim is verified by the calculations
	\begin{align}
		R & = R^{\partial \Gamma} + R^{\Gamma^2}
	\intertext{with}
	\begin{split} \label{eqn:R-partial-Gamma}
		R^{\partial \Gamma} & = g^{\nu \sigma} \left ( \partial_\mu \Gamma^\mu_{\nu \sigma} - \partial_\nu \Gamma^\mu_{\mu \sigma} \right ) \\
		& = \left ( \partial_\mu g^{\mu \rho} \right ) \left ( \partial_\nu g_{\rho \sigma} - \frac{1}{2} \partial_\rho g_{\nu \sigma} \right ) - \frac{1}{2} \left ( \partial_\nu g^{\mu \rho} \right ) \left ( \partial_\sigma g_{\mu \rho} \right ) \\
		& \hphantom{ = } + g^{\mu \rho} \left ( \partial_\mu \partial_\nu g_{\rho \sigma} - \frac{1}{2} \partial_\mu \partial_\rho g_{\nu \sigma} \right ) - \frac{1}{2} g^{\mu \rho} \left ( \partial_\nu \partial_\sigma g_{\mu \rho} \right ) \\
		& = g^{\mu \rho} g^{\nu \sigma} \left ( \partial_\mu \partial_\nu g_{\rho \sigma} - \partial_\mu \partial_\rho g_{\nu \sigma} \right ) \\
		& \hphantom{ = } + g^{\mu \rho} g^{\nu \sigma} g^{\kappa \lambda} \left ( \left ( \partial_\mu g_{\rho \kappa} \right ) \left ( \frac{1}{2} \partial_\lambda g_{\nu \sigma} - \partial_\nu g_{\lambda \sigma} \right ) + \frac{1}{2} \left ( \partial_\nu g_{\mu \kappa} \right ) \left ( \partial_\sigma g_{\rho \lambda} \right ) \right )
	\end{split}
	\intertext{and}
	\begin{split} \label{eqn:R-Gamma-2}
		R^{\Gamma^2} & = g^{\nu \sigma} \left ( \Gamma^\mu_{\mu \kappa} \Gamma^\kappa_{\nu \sigma} - \Gamma^\mu_{\nu \kappa} \Gamma^\kappa_{\mu \sigma} \right ) \\
		& = g^{\mu \rho} g^{\nu \sigma} g^{\kappa \lambda} \left ( \left ( \partial_\kappa g_{\mu \rho} \right ) \left ( \frac{1}{2} \partial_\nu g_{\sigma \lambda} - \frac{1}{4} \partial_\lambda g_{\nu \sigma} \right ) - \left ( \partial_\nu g_{\mu \kappa} \right ) \left ( \frac{1}{2} \partial_\rho g_{\sigma \lambda} - \frac{1}{4} \partial_\sigma g_{\rho \lambda} \right ) \right ) \, ,
	\end{split}
	\end{align}
	where we have used \(\left ( \partial_\rho g^{\nu \sigma} \right ) g_{\mu \sigma} = - g^{\nu \sigma} \left ( \partial_\rho g_{\mu \sigma} \right )\) in \eqnref{eqn:R-partial-Gamma} twice, which results from
	\begin{equation}
	\begin{split}
		0 & = \nabla^{(g)}_\rho \delta_\mu^\nu \\
		& = \partial_\rho \delta_\mu^\nu + \tensor{\Gamma}{_\rho _\sigma ^\nu} \delta_\mu^\sigma - \tensor{\Gamma}{_\rho _\mu ^\sigma} \delta_\sigma^\nu \\
		& = \partial_\rho \delta_\mu^\nu + \tensor{\Gamma}{_\rho _\mu ^\nu} - \tensor{\Gamma}{_\rho _\mu ^\nu} \\
		& = \partial_\rho \delta_\mu^\nu \\
		& = \partial_\rho \left ( g_{\mu \sigma} g^{\nu \sigma} \right ) \\
		& = \left ( \partial_\rho g_{\mu \sigma} \right ) g^{\nu \sigma} + g_{\mu \sigma} \left ( \partial_\rho g^{\nu \sigma} \right ) \, .
	\end{split}
	\end{equation}
\end{proof}

\enter

\begin{col} \label{col:ricci_scalar_for_the_levi_civita_connection_restriction}
	Given the situation of \propref{prop:ricci_scalar_for_the_levi_civita_connection}, the grade-\(m\) part in the gravitational coupling constant \(\gcoupling\) of the Ricci scalar \(R\) is given via
	\begin{subequations}
	\begin{align}
		\eval{R^{\partial \Gamma}}_{\order{\gcoupling^0}} & = \eval{R^{\Gamma^2}}_{\order{\gcoupling^0}} = \eval{R^{\Gamma^2}}_{\order{\gcoupling^1}} = 0 \, , \\
		\eval{R^{\partial \Gamma}}_{\order{\gcoupling^1}} & = \gcoupling \eta^{\mu \rho} \eta^{\nu \sigma} \left ( \partial_\mu \partial_\nu h_{\rho \sigma} - \partial_\mu \partial_\rho h_{\nu \sigma} \right ) \\
		\intertext{and for \(m > 1\)}
		\begin{split}
			\eval{R^{\partial \Gamma}}_{\order{\gcoupling^m}} & = - \left ( - \gcoupling \right )^m \sum_{i+j = m-1} \left ( h^i \right )^{\mu \rho} \left ( h^j \right )^{\nu \sigma} \left ( \partial_\mu \partial_\nu h_{\rho \sigma} - \partial_\mu \partial_\rho h_{\nu \sigma} \right )\\
			& \hphantom{ = } + \left ( - \gcoupling \right )^m \sum_{i+j+k = m-2} \left ( h^i \right )^{\mu \rho} \left ( h^j \right )^{\nu \sigma} \left ( h^k \right )^{\kappa \lambda} \left ( \left ( \partial_\mu h_{\rho \kappa} \right ) \left ( \frac{1}{2} \partial_\lambda h_{\nu \sigma} - \partial_\nu h_{\lambda \sigma} \right ) \right . \\ & \hphantom{ = + \left ( - \gcoupling \right )^m \sum_{i+j+k = m-2} \left ( h^i \right )^{\mu \rho} \left ( h^j \right )^{\nu \sigma} \left ( h^k \right )^{\kappa \lambda} ( } \left . + \frac{1}{2} \left ( \partial_\nu h_{\mu \kappa} \right ) \left ( \partial_\sigma h_{\rho \lambda} \right ) \right )
		\end{split}
		\intertext{and}
		\begin{split}
			\eval{R^{\Gamma^2}}_{\order{\gcoupling^m}} & = \left ( - \gcoupling \right )^m \sum_{i+j+k = m-2} \left ( h^i \right )^{\mu \rho} \left ( h^j \right )^{\nu \sigma} \left ( h^k \right )^{\kappa \lambda} \left ( \left ( \partial_\kappa h_{\mu \rho} \right ) \left ( \frac{1}{2} \partial_\nu h_{\sigma \lambda} - \frac{1}{4} \partial_\lambda h_{\nu \sigma} \right ) \right . \\ & \hphantom{ = \left ( - \gcoupling \right )^m \sum_{i+j+k = m-2} \left ( h^i \right )^{\mu \rho} \left ( h^j \right )^{\nu \sigma} \left ( h^k \right )^{\kappa \lambda} ( } \! \! \! \! \! \left . - \left ( \partial_\nu h_{\mu \kappa} \right ) \left ( \frac{1}{2} \partial_\rho h_{\sigma \lambda} - \frac{1}{4} \partial_\sigma h_{\rho \lambda} \right ) \right ) \, .
		\end{split}
	\end{align}
	\end{subequations}
\end{col}

\begin{proof}
	This follows directly from \propref{prop:ricci_scalar_for_the_levi_civita_connection} together with \lemref{lem:inverse_metric_series}.
\end{proof}

\enter

\begin{prop}[Metric expression for the de Donder gauge fixing] \label{prop:metric_expression_for_de_donder_gauge_fixing}
	Given the square of the de Donder gauge fixing,
	\begin{equation}
		\deDonder^2 := g^{\mu \nu} \deDonder_\mu \deDonder_\nu
	\end{equation}
	with \(\deDonder_\mu := g^{\rho \sigma} \Gamma_{\rho \sigma \mu}\), this can be rewritten as
	\begin{equation}
		\deDonder^2 = g^{\mu \rho} g^{\nu \sigma} g^{\kappa \lambda} \left ( \left ( \partial_\nu g_{\sigma \mu} \right ) \left ( \partial_\kappa g_{\lambda \rho} \right ) - \left ( \partial_\nu g_{\sigma \mu} \right ) \left ( \partial_\rho g_{\kappa \lambda} \right ) + \frac{1}{4} \left ( \partial_\mu g_{\nu \sigma} \right ) \left ( \partial_\rho g_{\kappa \lambda} \right ) \right ) \, .
	\end{equation}
	Furthermore, its quadratic part is given by
	\begin{equation}
	\begin{split}
		\deDonder_{(2)}^2 & := \eval{\deDonder^2}_{\order{\gcoupling^2}} \\
		& \phantom{:} \equiv \eta^{\mu \nu} \deDonder^{(1)}_\mu \deDonder^{(1)}_\nu
	\end{split}
	\end{equation}
	with \(\deDonder^{(1)}_\mu := \eta^{\rho \sigma} \Gamma_{\rho \sigma \mu}\), and can be rewritten as
	\begin{equation}
		\deDonder_{(2)}^2 = \eta^{\mu \rho} \eta^{\nu \sigma} \eta^{\kappa \lambda} \left ( \left ( \partial_\nu g_{\sigma \mu} \right ) \left ( \partial_\kappa g_{\lambda \rho} \right ) - \left ( \partial_\nu g_{\sigma \mu} \right ) \left ( \partial_\rho g_{\kappa \lambda} \right ) + \frac{1}{4} \left ( \partial_\mu g_{\nu \sigma} \right ) \left ( \partial_\rho g_{\kappa \lambda} \right ) \right ) \, .
	\end{equation}
\end{prop}

\begin{proof}
	The claim is verified by the calculation
	\begin{equation}
	\begin{split}
		\deDonder^2 & = g^{\mu \nu} \deDonder_\mu \deDonder_\nu \\
		& = \frac{1}{4} g^{\mu \rho} g^{\nu \sigma} g^{\kappa \lambda} \left ( \partial_\nu g_{\sigma \mu} + \partial_\sigma g_{\mu \nu} - \partial_\mu g_{\nu \sigma} \right ) \left ( \partial_\kappa g_{\lambda \rho} + \partial_\lambda g_{\rho \kappa} - \partial_\rho g_{\kappa \lambda} \right ) \\
		& = \frac{1}{4} g^{\mu \rho} g^{\nu \sigma} g^{\kappa \lambda} \left ( \left ( \partial_\nu g_{\sigma \mu} \right ) \left ( \partial_\kappa g_{\lambda \rho} \right ) + \left ( \partial_\nu g_{\sigma \mu} \right ) \left ( \partial_\lambda g_{\rho \kappa} \right ) - \left ( \partial_\nu g_{\sigma \mu} \right ) \left ( \partial_\rho g_{\kappa \lambda} \right ) \right . \\
		& \hphantom{= \frac{1}{4} g^{\mu \rho} g^{\nu \sigma} g^{\kappa \lambda} (} \left . + \left ( \partial_\sigma g_{\mu \nu} \right ) \left ( \partial_\kappa g_{\lambda \rho} \right ) + \left ( \partial_\sigma g_{\mu \nu} \right ) \left ( \partial_\lambda g_{\rho \kappa} \right ) - \left ( \partial_\sigma g_{\mu \nu} \right ) \left ( \partial_\rho g_{\kappa \lambda} \right ) \right . \\
		& \hphantom{= \frac{1}{4} g^{\mu \rho} g^{\nu \sigma} g^{\kappa \lambda} (} \left . - \left ( \partial_\mu g_{\nu \sigma} \right ) \left ( \partial_\kappa g_{\lambda \rho} \right ) - \left ( \partial_\mu g_{\nu \sigma} \right ) \left ( \partial_\lambda g_{\rho \kappa} \right ) + \left ( \partial_\mu g_{\nu \sigma} \right ) \left ( \partial_\rho g_{\kappa \lambda} \right ) \right ) \\
		& = g^{\mu \rho} g^{\nu \sigma} g^{\kappa \lambda} \left ( \left ( \partial_\nu g_{\sigma \mu} \right ) \left ( \partial_\kappa g_{\lambda \rho} \right ) - \left ( \partial_\nu g_{\sigma \mu} \right ) \left ( \partial_\rho g_{\kappa \lambda} \right ) + \frac{1}{4} \left ( \partial_\mu g_{\nu \sigma} \right ) \left ( \partial_\rho g_{\kappa \lambda} \right ) \right ) \, ,
	\end{split}
	\end{equation}
	together with the obvious restriction to \(\order{\gcoupling^2}\).
\end{proof}

\enter

\begin{col} \label{col:metric_expression_for_de_donder_gauge_fixing_restriction}
	Given the situation of \propref{prop:metric_expression_for_de_donder_gauge_fixing}, the grade-\(m\) part in the gravitational coupling constant \(\gcoupling\) of the square of the de Donder gauge fixing \(\deDonder^2\) is given for \(m < 2\) via
	\begin{subequations}
	\begin{align}
		\eval{\deDonder^2}_{\order{\gcoupling^m}} & = 0
		\intertext{and for \(m > 1\) via}
		\begin{split}
		\eval{\deDonder^2}_{\order{\gcoupling^m}} & = \left ( - \gcoupling \right )^m \sum_{i+j+k = m-2} \left ( h^i \right )^{\mu \rho} \left ( h^j \right )^{\nu \sigma} \left ( h^k \right )^{\kappa \lambda} \\ & \hphantom{\sum_{i+j+k = m-2}} \times \left ( \left ( \partial_\nu h_{\sigma \mu} \right ) \left ( \partial_\kappa h_{\lambda \rho} \right ) - \left ( \partial_\nu h_{\sigma \mu} \right ) \left ( \partial_\rho h_{\kappa \lambda} \right ) + \frac{1}{4} \left ( \partial_\mu h_{\nu \sigma} \right ) \left ( \partial_\rho h_{\kappa \lambda} \right ) \right ) \, .
		\end{split}
	\end{align}
	\end{subequations}
	In particular, the quadratic term \(\deDonder_{(2)}^2\) is given by
	\begin{equation}
	\begin{split}
		\deDonder_{(2)}^2 & := \eval{\deDonder^2}_{\order{\gcoupling^2}} \\
	& \phantom{:} = \gcoupling^2 \eta^{\mu \rho} \eta^{\nu \sigma} \eta^{\kappa \lambda} \left ( \left ( \partial_\nu h_{\sigma \mu} \right ) \left ( \partial_\kappa h_{\lambda \rho} \right ) - \left ( \partial_\nu h_{\sigma \mu} \right ) \left ( \partial_\rho h_{\kappa \lambda} \right ) + \frac{1}{4} \left ( \partial_\mu h_{\nu \sigma} \right ) \left ( \partial_\rho h_{\kappa \lambda} \right ) \right ) \, .
	\end{split}
	\end{equation}
\end{col}

\begin{proof}
This follows directly from \propref{prop:metric_expression_for_de_donder_gauge_fixing} together with \lemref{lem:inverse_metric_series}.
\end{proof}

\enter

\begin{prop}[Determinant of the metric as a series in the graviton field] \label{prop:determinant_metric}
	Given the metric decomposition from \defnref{defn:md_and_gf}, the negative of the determinant of the metric, \(- \dt{g}\), is given via
	\begin{equation}
		- \dt{g} = 1 + \first + \second + \third + \fourth \label{eqn:determinant_metric}
	\end{equation}
	with
	{\allowdisplaybreaks
	\begin{subequations} \label{eqns:first-fourth}
	\begin{align}
		\begin{split} \label{eqn:first}
			\first & := \gcoupling \tr{\eta h}\\
			& \hphantom{ : } \equiv \gcoupling \eta^{\mu \nu} h_{\mu \nu} \, ,
		\end{split}
		\\
		\begin{split}
			\second & := \gcoupling^2 \left ( \frac{1}{2} \tr{\eta h}^2 - \frac{1}{2} \tr{\left ( \eta h \right )^2} \right )\\
			& \hphantom{ : } \equiv \gcoupling^2 \left ( \frac{1}{2} \eta^{\mu \nu} \eta^{\rho \sigma} - \frac{1}{2} \eta^{\mu \sigma} \eta^{\rho \nu} \right ) h_{\mu \nu} h_{\rho \sigma} \, ,
		\end{split}
		\\
		\begin{split}
			\third & := \gcoupling^3 \left ( \frac{1}{6} \left ( \tr{\eta h} \right )^3 - \frac{1}{2} \tr{\eta h} \tr{\left ( \eta h \right )^2} + \frac{1}{3} \tr{\left ( \eta h \right )^3} \right )\\
			& \hphantom{ : } \equiv \gcoupling^3 \left ( \frac{1}{6} \eta^{\mu \nu} \eta^{\rho \sigma} \eta^{\lambda \tau} - \frac{1}{2} \eta^{\mu \nu} \eta^{\rho \tau} \eta^{\lambda \sigma} + \frac{1}{3} \eta^{\mu \tau} \eta^{\rho \nu} \eta^{\lambda \sigma} \right ) h_{\mu \nu} h_{\rho \sigma} h_{\lambda \tau}
		\end{split}
		\intertext{and}
		\begin{split} \label{eqn:fourth}
			\fourth & := \gcoupling^4 \left ( \frac{1}{24} \left ( \tr{\eta h} \right )^4 - \frac{1}{4} \left ( \tr{\eta h} \right )^2 \tr{\left ( \eta h \right )^2} + \frac{1}{3} \tr{\eta h} \tr{\left ( \eta h \right )^3} \right .\\ & \hphantom{ := } \left . + \frac{1}{8} \left ( \tr{\left ( \eta h \right )^2} \right )^2 - \frac{1}{4} \tr{\left ( \eta h \right )^4} \right )\\
			& \hphantom{ : } \equiv \gcoupling^4 \left ( \frac{1}{24} \eta^{\mu \nu} \eta^{\rho \sigma} \eta^{\lambda \tau} \eta^{\vartheta \varphi} - \frac{1}{4} \eta^{\mu \nu} \eta^{\rho \sigma} \eta^{\lambda \varphi} \eta^{\vartheta \tau} + \frac{1}{3} \eta^{\mu \nu} \eta^{\rho \varphi} \eta^{\lambda \sigma} \eta^{\vartheta \tau} \right . \\ & \hphantom{ := ( } \left . + \frac{1}{8} \eta^{\mu \sigma} \eta^{\rho \nu} \eta^{\lambda \varphi} \eta^{\vartheta \tau} - \frac{1}{4} \eta^{\mu \varphi} \eta^{\rho \nu} \eta^{\lambda \sigma} \eta^{\vartheta \tau} \right ) h_{\mu \nu} h_{\rho \sigma} h_{\lambda \tau} h_{\vartheta \varphi} \, .
		\end{split}
	\end{align}
	\end{subequations}
	}
	\end{prop}

\begin{proof}
	Given a \(4 \times 4\)-matrix \(\mtx \in \operatorname{Mat}_\mathbb{C} \left ( 4 \times 4 \right )\), from Newton's identities we get the relation
	\begin{equation} \label{eqn:Newtons_identities}
	\begin{split}
		\dt{\mtx} & = \frac{1}{4!} \operatorname{Det} \begin{pmatrix} \tr{\mtx} & 1 & 0 & 0 \\ \tr{\mtx^2} & \tr{\mtx} & 2 & 0 \\ \tr{\mtx^3} & \tr{\mtx^2} & \tr{\mtx} & 3 \\ \tr{\mtx^4} & \tr{\mtx^3} & \tr{\mtx^2} & \tr{\mtx} \end{pmatrix} \\
		& = \frac{1}{4!} \left ( \tr{\mtx}^4 - 6 \tr{\mtx}^2 \tr{\mtx^2} + 8 \tr{\mtx} \tr{\mtx^3} \right . \\
		& \hphantom{= \frac{1}{4!} (} \left . + 3 \tr{\mtx^2}^2 - 6 \tr{\mtx^4} \right ) \, .
	\end{split}
	\end{equation}
	Next, using the metric decomposition \(g = \eta + \gcoupling h\), we obtain\footnote{In accordance with index-notation, we set \(\delta\) to be the unit matrix.}
	\begin{equation}
	\begin{split}
		- \dt{g} & = - \dt{\eta + \gcoupling h}\\
		& = - \dt{\eta} \dt{\delta + \gcoupling \eta^{-1} h}\\
		& = \dt{\delta + \gcoupling \eta h} \, ,		
	\end{split}
	\end{equation}
	where we have used \(\dt{\eta} = - 1\) and \(\eta^{-1} = \eta\). Setting \(\mtx := \delta + \gcoupling \eta h\), using the linearity and cyclicity of the trace and the fact that \(\tr{\delta} = 4\), we get
	\begin{align}
		\tr{\delta + \gcoupling \eta h} & = 4 + \gcoupling \tr{\eta h}\\
		\tr{\left ( \delta + \gcoupling \eta h \right )^2} & = 4 + 2 \gcoupling \tr{\eta h} + \gcoupling^2 \tr{\left ( \eta h \right )^2}\\
		\tr{\left ( \delta + \gcoupling \eta h \right )^3} & = 4 + 3 \gcoupling \tr{\eta h} + 3 \gcoupling^2 \tr{\left ( \eta h \right )^2} + \gcoupling^3 \tr{\left ( \eta h \right )^3}\\
		\tr{\left ( \delta + \gcoupling \eta h \right )^4} & = 4 + 4 \gcoupling \tr{\eta h} + 6 \gcoupling^2 \tr{\left ( \eta h \right )^2} + 4 \gcoupling^3 \tr{\left ( \eta h \right )^3} + \gcoupling^4 \tr{\left ( \eta h \right )^4} \, .
	\end{align}
	Combining these results, we obtain
	\begin{equation}
	\begin{split}
		- \dt{g} & = 1 + \gcoupling \tr{\eta h} + \gcoupling^2 \left ( \frac{1}{2} \tr{\eta h}^2 - \frac{1}{2} \tr{\left ( \eta h \right )^2} \right ) \\
		& \hphantom{ = } + \gcoupling^3 \left ( \frac{1}{6} \tr{\eta h}^3 - \frac{1}{2} \tr{\eta h} \tr{\left ( \eta h \right )^2} + \frac{1}{3} \tr{\left ( \eta h \right )^3} \right ) \\
		& \hphantom{ = } + \gcoupling^4 \left ( \frac{1}{24} \tr{\eta h}^4 - \frac{1}{4} \tr{\eta h}^2 \tr{\eta h^2} + \frac{1}{3} \tr{\eta h} \tr{\left ( \eta h \right )^3} \right . \\ & \hphantom{ =  + \gcoupling^4 (} \left . + \frac{1}{8} \tr{\left ( \eta h \right )^2}^2 - \frac{1}{4} \tr{\left ( \eta h \right )^4} \right ) \, ,
	\end{split}
	\end{equation}
	which, when restricting to the powers in the coupling constant, yields the claimed result.
\end{proof}

\enter

\begin{col} \label{col:determinant_metric_restriction}
	Given the situation of \propref{prop:determinant_metric} and assume furthermore the boundedness condition from \assref{ass:bdns_gf}, the grade-\(m\) part in the gravitational coupling constant \(\gcoupling\) of the square-root of the negative of the determinant of the metric, \(- \dt{g}\), is given via
	\begin{equation}
	\begin{split}
		\eval{\sqrt{- \dt{g}}}_{\order{\gcoupling^m}} & = \sum_{\substack{i + j + k + l = m\\i \geq j \geq k \geq l \geq 0}} \sum_{p = 0}^{j - k} \sum_{q = 0}^{k - l} \sum_{r = 0}^{q} \sum_{s = 0}^l \sum_{t = 0}^s \sum_{u = 0}^t \sum_{v = 0}^u \\
		& \hphantom{ = } \binom{\frac{1}{2}}{i} \binom{i}{j} \binom{j}{k} \binom{k}{l} \binom{j - k}{p} \binom{k - l}{q} \binom{q}{r} \binom{l}{s} \binom{s}{t} \binom{t}{u} \binom{u}{v} \\
		& \hphantom{ = } \times \left ( - 1 \right )^{p + q - r + s - t + v} 2^{- j + l + r + s + 2t - 3u + v} 3^{- k + q - r + s - t + u} \\
		& \hphantom{ = } \times \mathfrak{a}^{i + j + k + l - 2p - 2q - r - 2s - t - u} \mathfrak{b}^{p + q - r + s - t + 2u - 2v} \mathfrak{c}^{r + t - u} \mathfrak{d}^v
	\end{split}
	\end{equation}
	with
	\begin{subequations} \label{eqns:varfirst-varfourth}
	\begin{align}
		\begin{split}
			\mathfrak{a} & := \gcoupling \tr{\eta h} \\
			& \hphantom{ : } \equiv \gcoupling \eta^{\mu \nu} h_{\mu \nu} \, ,
		\end{split}
		\\
		\begin{split}
			\mathfrak{b} & := \gcoupling^2 \tr{\left ( \eta h \right )^2} \\
			& \hphantom{ : } \equiv \gcoupling^2 \eta^{\mu \sigma} \eta^{\rho \nu} h_{\mu \nu} h_{\rho \sigma} \, ,
		\end{split}
		\\
		\begin{split}
			\mathfrak{c} & := \gcoupling^3 \tr{\left ( \eta h \right )^3} \\
			& \hphantom{ : } \equiv \gcoupling^3 \eta^{\mu \tau} \eta^{\rho \nu} \eta^{\lambda \sigma} h_{\mu \nu} h_{\rho \sigma} h_{\lambda \tau}
		\end{split}
		\intertext{and}
		\begin{split}
			\mathfrak{d} & := \gcoupling^4 \tr{\left ( \eta h \right )^4} \\
			& \hphantom{ : } \equiv \gcoupling^4 \eta^{\mu \varphi} \eta^{\rho \nu} \eta^{\lambda \sigma} \eta^{\vartheta \tau} h_{\mu \nu} h_{\rho \sigma} h_{\lambda \tau} h_{\vartheta \varphi} \, .
		\end{split}
	\end{align}
	\end{subequations}
\end{col}

\begin{proof}
	We use \eqnref{eqn:determinant_metric},
	\begin{equation}
		- \dt{g} = 1 + \first + \second + \third + \fourth \, ,
	\end{equation}
	and plug it into the Taylor series of the square-root around \(x = 0\),\footnote{Here we need the assumption \(\left | \gcoupling \right | \left \| h \right \|_{\max} := \left | \gcoupling \right | \max_{\lambda \in \operatorname{EW} \left ( h \right )} \left | \lambda \right | < 1\), where \(\operatorname{EW} \left ( h \right )\) denotes the set of eigenvalues of \(h\), to assure convergence.}
	\begin{equation}
		\sqrt{x} = \sum_{i = 0}^\infty \binom{\frac{1}{2}}{i} \left ( x - 1 \right )^i \, ,
	\end{equation}
	to obtain
	\begin{equation}
		\sqrt{- \dt{g}} = \sum_{i = 0}^\infty \binom{\frac{1}{2}}{i} \left ( \first + \second + \third + \fourth \right )^i \, .
	\end{equation}
	Applying the binomial theorem iteratively three times, we get
	\begin{equation}
	\begin{split}
		\sqrt{- \dt{g}} & = \sum_{i = 0}^\infty \binom{\frac{1}{2}}{i} \left ( \first + \second + \third + \fourth \right )^i\\
		& = \sum_{i = 0}^\infty \sum_{j = 0}^i \binom{\frac{1}{2}}{i} \binom{i}{j} \first^{i - j} \left ( \second + \third + \fourth \right )^j\\
		& = \sum_{i = 0}^\infty \sum_{j = 0}^i \sum_{k = 0}^j \binom{\frac{1}{2}}{i} \binom{i}{j} \binom{j}{k} \first^{i - j} \second^{j - k} \left ( \third + \fourth \right )^k\\
		& = \sum_{i = 0}^\infty \sum_{j = 0}^i \sum_{k = 0}^j \sum_{l = 0}^k \binom{\frac{1}{2}}{i} \binom{i}{j} \binom{j}{k} \binom{k}{l} \first^{i - j} \second^{j - k} \third^{k - l} \fourth^l \, .
	\end{split}
	\end{equation}
	Observe, that from \eqnsaref{eqn:determinant_metric}{eqns:first-fourth} we have the relations
	\begin{subequations}
	\begin{align}
		\eval{- \dt{g}}_{\order{\gcoupling}} & \equiv \first \\
		\eval{- \dt{g}}_{\order{\gcoupling^2}} & \equiv \second \\
		\eval{- \dt{g}}_{\order{\gcoupling^3}} & \equiv \third
		\intertext{and}
		\eval{- \dt{g}}_{\order{\gcoupling^4}} & \equiv \fourth \, ,
	\end{align}
	\end{subequations}
	and thus the restriction to the grade-\(m\) part in the gravitational coupling constant \(\gcoupling\) is given via the integer solutions to
	\begin{equation}
	\begin{split}
		m & \overset{!}{=} i - j + 2 j - 2 k + 3 k - 3 l + 4 l\\
		& = i + j + k + l
	\end{split}
	\end{equation}
	with \(i \geq j \geq k \geq l\), i.e.\
	\begin{equation} \label{eqn:riemannean_volume_form_grade_m}
		\eval{\sqrt{- \dt{g}}}_{\order{\gcoupling^m}} = \sum_{\substack{i + j + k + l = m\\i \geq j \geq k \geq l \geq 0}} \binom{\frac{1}{2}}{i} \binom{i}{j} \binom{j}{k} \binom{k}{l} \first^{i - j} \second^{j - k} \third^{k - l} \fourth^l \, .
	\end{equation}
	Finally, using Newton's identities, i.e.\ the relations from \eqnsaref{eqns:first-fourth}{eqns:varfirst-varfourth},
	\begin{subequations}
	\begin{align}
		\first & \equiv \mathfrak{a} \, , \\
		\second & \equiv \frac{1}{2} \mathfrak{a}^2 - \frac{1}{2} \mathfrak{b} \, , \\
		\third & \equiv \frac{1}{6} \mathfrak{a}^3 - \frac{1}{2} \mathfrak{a} \mathfrak{b} + \frac{1}{3} \mathfrak{c}
		\intertext{and}
		\fourth & \equiv \frac{1}{24} \mathfrak{a}^4 - \frac{1}{4} \mathfrak{a}^2 \mathfrak{b} + \frac{1}{3} \mathfrak{a} \mathfrak{c} + \frac{1}{8} \mathfrak{b}^2 - \frac{1}{4} \mathfrak{d} \, ,
	\end{align}
	\end{subequations}
	we obtain, using again the Binomial theorem iteratively seven times,
	\begin{equation}
	\begin{split}
		\first^{i - j} \second^{j - k} \third^{k - l} \fourth^l & = \sum_{p = 0}^{j - k} \sum_{q = 0}^{k - l} \sum_{r = 0}^{q} \sum_{s = 0}^l \sum_{t = 0}^s \sum_{u = 0}^t \sum_{v = 0}^u\\
		& \hphantom{ = } \binom{j - k}{p} \binom{k - l}{q} \binom{q}{r} \binom{l}{s} \binom{s}{t} \binom{t}{u} \binom{u}{v}\\
		& \hphantom{ = } \times \left ( - 1 \right )^{p + q - r + s - t + v} 2^{- j + l + r + s + 2t - 3u + v} 3^{- k + q - r + s - t + u} \\
		& \hphantom{ = } \times \mathfrak{a}^{i + j + k + l - 2p - 2q - r - 2s - t - u} \mathfrak{b}^{p + q - r + s - t + 2u - 2v} \mathfrak{c}^{r + t - u} \mathfrak{d}^v \, ,
	\end{split}
	\end{equation}
	and thus finally
	\begin{equation}
	\begin{split}
		\eval{\sqrt{- \dt{g}}}_{\order{\gcoupling^m}} & = \sum_{\substack{i + j + k + l = m\\i \geq j \geq k \geq l \geq 0}} \sum_{p = 0}^{j - k} \sum_{q = 0}^{k - l} \sum_{r = 0}^{q} \sum_{s = 0}^l \sum_{t = 0}^s \sum_{u = 0}^t \sum_{v = 0}^u \\
		& \hphantom{ = } \binom{\frac{1}{2}}{i} \binom{i}{j} \binom{j}{k} \binom{k}{l} \binom{j - k}{p} \binom{k - l}{q} \binom{q}{r} \binom{l}{s} \binom{s}{t} \binom{t}{u} \binom{u}{v} \\
		& \hphantom{ = } \times \left ( - 1 \right )^{p + q - r + s - t + v} 2^{- j + l + r + s + 2t - 3u + v} 3^{- k + q - r + s - t + u} \\
		& \hphantom{ = } \times \mathfrak{a}^{i + j + k + l - 2p - 2q - r - 2s - t - u} \mathfrak{b}^{p + q - r + s - t + 2u - 2v} \mathfrak{c}^{r + t - u} \mathfrak{d}^v \, ,
	\end{split}
	\end{equation}
	as claimed.
\end{proof}

\section{Feynman rules}

Given the Quantum General Relativity Lagrange density
\begin{equation}
	\begin{split}
	\mathcal{L}_\text{QGR} & = - \frac{1}{2 \gcoupling^2} \left ( \sqrt{- \dt{g}} R + \frac{1}{2 \zeta}  \eta^{\mu \nu} \deDonder^{(1)}_\mu \deDonder^{(1)}_\nu \right ) \dif V_\eta \\
	& \phantom{:=} - \frac{1}{2} \eta^{\rho \sigma} \left ( \frac{1}{\zeta} \overline{C}^\mu \left ( \partial_\rho \partial_\sigma C_\mu \right ) + \overline{C}^\mu \left ( \partial_\mu \big ( \tensor{\Gamma}{^\nu _\rho _\sigma} C_\nu \big ) - 2 \partial_\rho \big ( \tensor{\Gamma}{^\nu _\mu _\sigma} C_\nu \big ) \right ) \right ) \dif V_\eta
	\end{split}
\end{equation}
from \conref{con:Lagrange_density} and the decomposition into its powers in the gravitational coupling constant \(\varkappa\) and the ghost field \(C\)
\begin{equation}
	\mathcal{L}_\text{QGR} \equiv \sum_{m = 0}^\infty \sum_{n = 0}^1 \mathcal{L}_\text{QGR}^{m,n}
\end{equation}
from the introduction of \sectionref{sec:expansion_lagrange_density}. Then, we extend the Lagrange densities \(\mathcal{L}_\text{QGR}^{m,n}\) for given \(m \in \mathbb{N}_+\), which were interpreted in the introduction of \sectionref{sec:expansion_lagrange_density} as potential terms for either \(\left ( m + 2 \right )\) gravitons or \(m\) gravitons and a graviton-ghost and graviton-antighost, to either \(\left ( m + 2 \right )\) distinguishable gravitons or \(m\) distinguishable gravitons and a graviton-ghost and graviton-antighost via symmetrization, depending on \(n \in \set{0,1}\). This then reflects the bosonic character of gravitons and allows the calculation of the corresponding Feynman rules as the remaining matrix elements of these potential terms. We start by introducing the notation and then present the Feynman rules.

\enter

\begin{defn}
	We denote the graviton \(m\)-point vertex Feynman rule with ingoing momenta \(\set{p_1^\sigma, \cdots, p_m^\sigma}\) via \(\gravfr_m^{\mu_1 \nu_1 \vert \cdots \vert \mu_m \nu_m} \left ( p_1^\sigma, \cdots, p_m^\sigma \right )\).\footnote{The vertical bars in \(\mu_1 \nu_1 \vert \cdots \vert \mu_m \nu_m\) are added solely for better readability.} It is defined as follows:
	\begin{equation}
		\gravfr_m^{\mu_1 \nu_1 \vert \cdots \vert \mu_m \nu_m} \left ( p_1^\sigma, \cdots, p_m^\sigma \right ) := \imaginary \left ( \prod_{i = 1}^m \frac{\bar{\delta}}{\bar{\delta} \hat{h}_{\mu_i \nu_i}} \right ) \mathscr{F} \left ( \overline{\mathcal{L}}_\text{QGR}^{(m-2),0} \right ) \, ,
	\end{equation}
	where the prefactor \(\imaginary\) is a convention from the path integral, \(\textfrac{\bar{\delta}}{\bar{\delta} \hat{h}_{\mu_i \nu_i}}\) denotes the symmetrized functional derivative with respect to the Fourier transformed graviton field \(\hat{h}_{\mu_i \nu_i}\) together with the additional agreement (represented by the bar \(\textfrac{\bar{\delta}}{\bar{\delta} \cdot}\)) that the possible preceding momentum is also labelled by the particle number \(i\), e.g.\
	\begin{equation}
		\frac{\bar{\delta}}{\bar{\delta} \hat{h}_{\mu_i \nu_i}} \left ( p_\kappa \hat{h}_{\rho \sigma} \right ) := \frac{1}{2} p_\kappa^i \left ( \hat{\delta}_\rho^{\mu_i} \hat{\delta}_\sigma^{\nu_i} + \hat{\delta}_\sigma^{\mu_i} \hat{\delta}_\rho^{\nu_i} \right ) \, ,
	\end{equation}
	and \(\overline{\mathcal{L}}_\text{QGR}^{(m-2),0}\) is the symmetrized extension of \(\mathcal{L}_\text{QGR}^{(m-2),0}\) to \(m\) distinguishable gravitons. Furthermore, we denote the graviton propagator Feynman rule with momentum \(p^\sigma\), gauge parameter \(\zeta\) and regulator for Landau singularities \(\epsilon\) via \(\gravprop_{\mu_1 \nu_1 \vert \mu_2 \nu_2} \left ( p^\sigma; \zeta; \epsilon \right )\). It is defined as the inverse of the matrix element for the graviton kinetic term:\footnote{We use momentum conservation to set \(p_1^\sigma := p^\sigma\) and \(p_2^\sigma := - p^\sigma\) in the expression \(\gravfr_2^{\mu_2 \nu_2 \vert \mu_3 \nu_3} \left ( p_1^\sigma, p_2^\sigma; \zeta \right )\).}
	\begin{equation}
		\gravprop_{\mu_1 \nu_1 \vert \mu_2 \nu_2} \left ( p^\sigma; \zeta; 0 \right ) \gravfr_2^{\mu_2 \nu_2 \vert \mu_3 \nu_3} \left ( p^\sigma; \zeta \right ) = \frac{1}{2} \left ( \hat{\delta}_{\mu_1}^{\mu_3} \hat{\delta}_{\nu_1}^{\nu_3} + \hat{\delta}_{\mu_1}^{\nu_3} \hat{\delta}_{\nu_1}^{\mu_3} \right ) \, ,
	\end{equation}
	where each tuple \(\mu_i \nu_i\) is treated as one index, which excludes the a priori possible term \(\hat{\eta}_{\mu_1 \nu_1} \hat{\eta}^{\mu_3 \nu_3}\) on the right-hand side. Moreover, we denote the graviton-ghost \(m\)-point vertex Feynman rule with ingoing momenta \(\set{p_1^\sigma, \cdots, p_m^\sigma}\) via \(\ghostfr_m^{\rho_1 \vert \rho_2 \| \mu_3 \nu_3 \vert \cdots \vert \mu_m \nu_m} \left ( p_1^\sigma, \cdots, p_m^\sigma \right )\), where particle 1 is the graviton-ghost, particle 2 is the graviton-antighost and the rest are gravitons. It is defined as follows:
	\begin{equation}
		\ghostfr_m^{\rho_1 \vert \rho_2 \| \mu_3 \nu_3 \vert \cdots \vert \mu_m \nu_m} \left ( p_1^\sigma, \cdots, p_m^\sigma \right ) := \imaginary \left ( \frac{\bar{\delta}}{\bar{\delta} \widehat{\gravitonghost}_{\rho_1}} \frac{\bar{\delta}}{\bar{\delta} \widehat{\overline{\gravitonghost}}_{\rho_2}} \prod_{i = 3}^m \frac{\bar{\delta}}{\bar{\delta} \hat{h}_{\mu_i \nu_i}} \right ) \mathscr{F} \left ( \overline{\mathcal{L}}_\text{QGR}^{m,1} \right ) \, ,
	\end{equation}
where, additionally to the above mentioned setting, \(\textfrac{\bar{\delta}}{\bar{\delta} \widehat{\gravitonghost}_{\rho_1}}\) and \(\textfrac{\bar{\delta}}{\bar{\delta} \widehat{\overline{\gravitonghost}}_{\rho_2}}\) denotes the functional derivative with respect to the Fourier transformed graviton-ghost field \(\widehat{\gravitonghost}_{\rho_1}\) and Fourier transformed graviton-antighost field \(\widehat{\overline{\gravitonghost}}_{\rho_2}\), respectively, and \(\overline{\mathcal{L}}_\text{QGR}^{m,1}\) is the symmetrized extension of \(\mathcal{L}_\text{QGR}^{m,1}\) to \(m\) distinguishable gravitons. Additionally, we denote the graviton-ghost propagator Feynman rule with momentum \(p^\sigma\) and regulator for Landau singularities \(\epsilon\) via \(\ghostprop_{\rho_1 \vert \rho_2} \left ( p^\sigma; \epsilon \right )\). It is defined as the inverse of the matrix element for the graviton-ghost kinetic term:\footnote{Again, we use momentum conservation to set \(p_1^\sigma := p^\sigma\) and \(p_2^\sigma := - p^\sigma\) in the expression \(\ghostfr_2^{\mu_2 \nu_2 \vert \mu_3 \nu_3} \left ( p_1^\sigma, p_2^\sigma \right )\).}
	\begin{equation}
		\ghostprop_{\rho_1 \vert \rho_2} \left ( p^\sigma; 0 \right ) \ghostfr_2^{\rho_2 \vert \rho_3} \left ( p^\sigma \right ) = \hat{\delta}_{\rho_1}^{\rho_3} \, .
	\end{equation}
	Finally, we denote the graviton-matter \(m\)-point vertex Feynman rule of type \(k\) from \lemref{lem:matter-model-Lagrange-densities} with ingoing momenta \(\set{p_1^\sigma, \cdots, p_m^\sigma}\) via \(\matterfrk_m^{\kappa \dots \tau \| o \dots t \triplevert \mu_1 \nu_1 \vert \cdots \vert \mu_m \nu_m} \left ( p_1^\sigma, \cdots, p_m^\sigma \right )\), where we count only graviton particles, as the matter-contributions are condensed into the tensors \(\tensor[_k]{\! T}{}\), whose Feynman rule contributions can be found e.g.\ in \cite{Romao_Silva}. They are defined as follows:
	\begin{multline}
		\matterfrk_m^{\kappa \dots \tau \| o \dots t \triplevert \mu_1 \nu_1 \vert \cdots \vert \mu_m \nu_m} \left ( p_1^\sigma, \cdots, p_m^\sigma \right ) := \\ \imaginary \left ( \frac{\bar{\delta}}{\bar{\delta} \tensor[_k]{\! \widehat{T}}{_\kappa _\dots _\tau _\| _o _\dots _t}} \prod_{i = 1}^m \frac{\bar{\delta}}{\bar{\delta} \hat{h}_{\mu_i \nu_i}} \right ) \mathscr{F} \left ( \tensor[_k]{{\overline{\mathcal{L}}^{m,0}_{\text{QGR-SM}}}}{} \right ) \, ,
	\end{multline}
	where we use again the above mentioned setting.
\end{defn}

\enter

\begin{con}
	We consider all momenta \(\set{p_1^\sigma, \cdots, p_m^\sigma}\) incoming and we assume momentum conservation on quadratic Feynman rules, i.e.\ set \(p_1^\sigma := p^\sigma\) and \(p_2^\sigma := - p^\sigma\).
\end{con}

\subsection{Preparations for gravitons and their ghosts} \label{ssec:preparations_graviton_graviton_ghost}

In this subsection we prepare all necessary objects for the graviton and graviton-ghost Feynman rules.

\enter

\begin{lem} \label{lem:traces_FR}
	Introducing the notation
	\begin{equation}
	\mathfrak{T}_n^{\mu_1 \nu_1 \vert \cdots \vert \mu_n \nu_n} := \left ( \prod_{i = 1}^n \frac{\bar{\delta}}{\bar{\delta} \hat{h}_{\mu_i \nu_i}} \right ) \mathscr{F} \left ( \tr{\left ( \eta h \right )^n} \right ) \, ,
	\end{equation}
	we obtain
	\begin{subequations}
	\begin{align}
		\mathfrak{T}_n^{\mu_1 \nu_1 \vert \cdots \vert \mu_n \nu_n} & = \frac{1}{2^n} \sum_{\mu_i \leftrightarrow \nu_i} \sum_{s \in S_n} \mathfrak{t}_n^{\mu_{s(1)} \nu_{s(1)} \vert \cdots \vert \mu_{s(n)} \nu_{s(n)}}
		\intertext{with}
		\mathfrak{t}_n^{\mu_1 \nu_1 \vert \cdots \vert \mu_n \nu_n} & = \gcoupling^n \left ( \hat{\delta}^{\nu_1}_{\nu_{n+1}} \prod_{a = 1}^n \hat{\eta}^{\mu_a \nu_{a+1}} \right ) \, .
	\end{align}
	\end{subequations}
	Furthermore, introducing the notation
	\begin{equation}
		\mathfrak{H}_n^{\mu \nu \triplevert \mu_1 \nu_1 \vert \cdots \vert \mu_n \nu_n} := \left ( \prod_{i = 1}^n \frac{\bar{\delta}}{\bar{\delta} \hat{h}_{\mu_i \nu_i}} \right ) \mathscr{F} \left ( \left ( h^n \right )^{\mu \nu} \right ) \, ,
	\end{equation}
	we obtain
	\begin{subequations}
	\begin{align}
		\mathfrak{H}_0^{\mu \nu} & = \eta^{\mu \nu}
		\intertext{and for \(n > 0\)}
		\mathfrak{H}_n^{\mu \nu \triplevert \mu_1 \nu_1 \vert \cdots \vert \mu_n \nu_n} & = \frac{1}{2^n} \sum_{\mu_i \leftrightarrow \nu_i} \sum_{s \in S_n} \mathfrak{h}_n^{\mu \nu \vert \mu_{s(1)} \nu_{s(1)} \vert \cdots \vert \mu_{s(n)} \nu_{s(n)}}
		\intertext{with}
		\mathfrak{h}_n^{\mu \nu \triplevert \mu_1 \nu_1 \vert \cdots \vert \mu_n \nu_n} & = \gcoupling^n \left ( \hat{\delta}^{\mu}_{\mu_0} \hat{\delta}^{\nu}_{\nu_{n+1}} \prod_{a = 0}^n \hat{\eta}^{\mu_a \nu_{a+1}} \right ) \, .
	\end{align}
	\end{subequations}
	Moreover, introducing the notation
	\begin{equation}
		\left ( \mathfrak{H}_n^\prime \right )_\rho^{\mu \nu \triplevert \mu_1 \nu_1 \vert \cdots \vert \mu_n \nu_n} \left ( p_1^\sigma, \cdots, p_n^\sigma \right ) := \left ( \prod_{i = 1}^n \frac{\bar{\delta}}{\bar{\delta} \hat{h}_{\mu_i \nu_i}} \right ) \mathscr{F} \left ( \partial_\rho \left ( \left ( h^n \right )^{\mu \nu} \right ) \right ) \, ,
	\end{equation}
	we obtain
	\begin{subequations}
	\begin{align}
		\left ( \mathfrak{H}_0^\prime \right )_\rho^{\mu \nu} & = 0 \\
		\intertext{and for \(n > 0\)}
		\begin{split}
			\left ( \mathfrak{H}_n^\prime \right )_\rho^{\mu \nu \triplevert \mu_1 \nu_1 \vert \cdots \vert \mu_n \nu_n} \left ( p_1^\sigma, \cdots, p_n^\sigma \right ) & = \\ & \hphantom{=} \mkern-44mu \frac{1}{2^n} \sum_{\mu_i \leftrightarrow \nu_i} \sum_{s \in S_n} \left ( \mathfrak{h}_n^\prime \right )_\rho^{\mu \nu \triplevert \mu_{s(1)} \nu_{s(1)} \vert \cdots \vert \mu_{s(n)} \nu_{s(n)}} \left ( p_{s(1)}^\sigma, \cdots, p_{s(n)}^\sigma \right )
		\end{split}
		\intertext{with}
		\left ( \mathfrak{h}_n^\prime \right )_\rho^{\mu \nu \triplevert \mu_1 \nu_1 \vert \cdots \vert \mu_n \nu_n} \left ( p_1^\sigma, \cdots, p_n^\sigma \right ) & = \gcoupling^n \left ( \sum_{m = 1}^n p^n_\rho \right ) \left ( \hat{\delta}^{\mu}_{\mu_0} \hat{\delta}^{\nu}_{\nu_{n+1}} \prod_{a = 0}^n \hat{\eta}^{\mu_a \nu_{a+1}} \right ) \, .
	\end{align}
	\end{subequations}
\end{lem}

\begin{proof}
	This follows from directly from the definition.
\end{proof}

\enter

\begin{col} \label{col:inverse_metric_vielbeins_FR}
	Given the situation of \lemref{lem:traces_FR}, we have
	{\allowdisplaybreaks
	\begin{align}
		\left ( \prod_{i = 1}^n \frac{\bar{\delta}}{\bar{\delta} \hat{h}_{\mu_i \nu_i}} \right ) \mathscr{F} \left ( \eval{g^{\mu \nu}}_{\order{\gcoupling^n}} \right ) & = \left ( - 1 \right )^n \mathfrak{H}_n^{\mu \nu \triplevert \mu_1 \nu_1 \vert \cdots \vert \mu_n \nu_n} \, ,
		\\
		\left ( \prod_{i = 1}^n \frac{\bar{\delta}}{\bar{\delta} \hat{h}_{\mu_i \nu_i}} \right ) \mathscr{F} \left ( \eval{e_\rho^r}_{\order{\gcoupling^n}} \right ) & = \binom{\frac{1}{2}}{n} \left ( \mathfrak{H}_n \right )_\rho^{r \triplevert \mu_1 \nu_1 \vert \cdots \vert \mu_n \nu_n} \, ,
		\\
		\left ( \prod_{i = 1}^n \frac{\bar{\delta}}{\bar{\delta} \hat{h}_{\mu_i \nu_i}} \right ) \mathscr{F} \left ( \eval{e^\rho_r}_{\order{\gcoupling^n}} \right ) & = \binom{- \frac{1}{2}}{n} \left ( \mathfrak{H}_n \right )_r^{\rho \triplevert \mu_1 \nu_1 \vert \cdots \vert \mu_n \nu_n} \, ,
		\\
		\left ( \prod_{i = 1}^n \frac{\bar{\delta}}{\bar{\delta} \hat{h}_{\mu_i \nu_i}} \right ) \mathscr{F} \left ( \eval{\left ( \partial_\sigma e_\rho^r \right )}_{\order{\gcoupling^n}} \right ) & = \binom{\frac{1}{2}}{n} \hat{\eta}_{\mu \rho} \hat{\delta}_\nu^r \left ( \mathfrak{H}_n^\prime \right )_\sigma^{\mu \nu \triplevert \mu_1 \nu_1 \vert \cdots \vert \mu_n \nu_n} \left ( p_1^\sigma, \cdots, p_n^\sigma \right )
		\intertext{and}
		\left ( \prod_{i = 1}^n \frac{\bar{\delta}}{\bar{\delta} \hat{h}_{\mu_i \nu_i}} \right ) \mathscr{F} \left ( \eval{\left ( \partial_\sigma e^\rho_r \right )}_{\order{\gcoupling^n}} \right ) & = \binom{- \frac{1}{2}}{n} \hat{\delta}_\mu^\rho \hat{\eta}_{\nu r} \left ( \mathfrak{H}_n^\prime \right )_\sigma^{\mu \nu \triplevert \mu_1 \nu_1 \vert \cdots \vert \mu_n \nu_n} \left ( p_1^\sigma, \cdots, p_n^\sigma \right ) \, .
	\end{align}
}
\end{col}

\begin{proof}
	This follows directly from Lemmata \ref{lem:inverse_metric_series}, \ref{lem:vielbeins_series} and \ref{lem:traces_FR}.
\end{proof}

\enter

\begin{lem} \label{lem:Christoffel_FR}
	Introducing the notation
	\begin{equation}
		\boldsymbol{\Gamma}^{\mu_1 \nu_1}_{\mu \nu \rho} \left ( p_1^\sigma \right ) := \frac{\bar{\delta}}{\bar{\delta} \hat{h}_{\mu_1 \nu_1}} \mathscr{F} \left ( \Gamma_{\mu \nu \rho} \right )
	\end{equation}
	with
	\begin{equation}
	\begin{split}
		\Gamma_{\mu \nu \rho} & := g_{\rho \sigma} \tensor{\Gamma}{_\mu _\nu ^\sigma} \\
		& \hphantom{:} \equiv \frac{1}{2} \left ( \partial_\mu g_{\nu \rho} + \partial_\nu g_{\rho \mu} - \partial_\rho g_{\mu \nu} \right )
	\end{split}
	\end{equation}
	we obtain
	\begin{equation}
	\begin{split}
		\boldsymbol{\Gamma}^{\mu_1 \nu_1}_{\mu \nu \rho} \left ( p_1^\sigma \right ) & = \frac{\gcoupling}{4} \left ( p^1_\mu \left ( \hat{\delta}_\rho^{\mu_1} \hat{\delta}_\nu^{\nu_1} + \hat{\delta}_\nu^{\mu_1} \hat{\delta}_\rho^{\nu_1} \right ) + p^1_\nu \left ( \hat{\delta}_\mu^{\mu_1} \hat{\delta}_\rho^{\nu_1} + \hat{\delta}_\rho^{\mu_1} \hat{\delta}_\mu^{\nu_1} \right ) \right . \\
		& \hphantom{= \frac{\gcoupling}{4} (} \left . - p^1_\rho \left ( \hat{\delta}_\mu^{\mu_1} \hat{\delta}_\nu^{\nu_1} + \hat{\delta}_\nu^{\mu_1} \hat{\delta}_\mu^{\nu_1} \right ) \right ) \, .
	\end{split}
	\end{equation}
\end{lem}

\begin{proof}
	This follows from directly from the expression
	\begin{equation}
		\widehat{\Gamma}_{\mu \nu \rho} = \frac{\gcoupling}{2} \left ( p_\mu \hat{h}_{\nu \rho} + p_\nu \hat{h}_{\rho \mu} - p_\rho \hat{h}_{\mu \nu} \right ) \, .
	\end{equation}
\end{proof}

\enter

\begin{lem} \label{lem:Ricci_scalar_FR}
	Introducing the notation
	\begin{equation}
		\mathfrak{R}_n^{\mu_1 \nu_1 \vert \cdots \vert \mu_n \nu_n} \left ( p_1^\sigma, \cdots, p_n^\sigma \right ) := \left ( \prod_{i = 1}^n \frac{\bar{\delta}}{\bar{\delta} \hat{h}_{\mu_i \nu_i}} \right ) \mathscr{F} \left ( \eval{R}_{\order{\gcoupling^{n}}} \right ) \, ,
	\end{equation}
	we obtain
	\begin{subequations}
	\begin{align}
		\mathfrak{R}_0 & = 0 \, , \\
		\mathfrak{R}_1^{\mu_1 \nu_1} \left ( p_1^\sigma \right ) & = - \gcoupling \left ( p_1^{\mu_1} p_1^{\nu_1} - p_1^2 \hat{\eta}^{\mu_1 \nu_1} \right )
		\intertext{and for \(n > 1\)}
		\mathfrak{R}_n^{\mu_1 \nu_1 \vert \cdots \vert \mu_n \nu_n} \left ( p_1^\sigma, \cdots, p_n^\sigma \right ) & = \frac{1}{2^n} \sum_{\mu_i \leftrightarrow \nu_i} \sum_{s \in S_n} \mathfrak{r}_n^{\mu_{s(1)} \nu_{s(1)} \vert \cdots \vert \mu_{s(n)} \nu_{s(n)}} \left ( p_{s(1)}^\sigma, \cdots, p_{s(n)}^\sigma \right )
		\intertext{with}
		\begin{split}
			\mathfrak{r}_n^{\mu_1 \nu_1 \vert \cdots \vert \mu_n \nu_n} \left ( p_1^\sigma, \cdots, p_n^\sigma \right ) & = \left ( \mathfrak{r}_n^{\partial \Gamma} \right )^{\mu_1 \nu_1 \vert \cdots \vert \mu_n \nu_n} \left ( p_1^\sigma, \cdots, p_n^\sigma \right ) \\ & \hphantom{=} + \left ( \mathfrak{r}_n^{\Gamma^2} \right )^{\mu_1 \nu_1 \vert \cdots \vert \mu_n \nu_n} \left ( p_1^\sigma, \cdots, p_n^\sigma \right ) \, ,
		\end{split}
		\\
		\begin{split}
			\left ( \mathfrak{r}_n^{\partial \Gamma} \right )^{\mu_1 \nu_1 \vert \cdots \vert \mu_n \nu_n} \left ( p_1^\sigma, \cdots, p_n^\sigma \right ) & = \left ( - \gcoupling \right )^n \sum_{i+j = n-1} \left ( \hat{\delta}^\rho_{\nu_{i+1}} \prod_{a = 0}^i \hat{\eta}^{\mu_a \nu_{a+1}} \right ) \left ( \hat{\delta}^\nu_{\mu_i} \prod_{b = i}^{i+j} \hat{\eta}^{\mu_b \nu_{b+1}} \right ) \\
			& \hphantom{= - \gcoupling^n \sum_{i+j = n-1}} \times \left ( p^n_{\mu_0} p^n_\nu \hat{\delta}_\rho^{\mu_n} - p^n_{\mu_0} p^n_\rho \hat{\delta}_\nu^{\mu_n} \right ) \\
			& \hphantom{=} \mkern-72mu - \left ( - \gcoupling \right )^n \sum_{i+j+k = n-2} \left ( \hat{\delta}^\rho_{\nu_{i+1}} \prod_{a = 0}^i \hat{\eta}^{\mu_a \nu_{a+1}} \right ) \left ( \hat{\delta}^\nu_{\mu_i} \hat{\delta}^\sigma_{\nu_{i+j+1}} \prod_{b = i}^{i+j} \hat{\eta}^{\mu_b \nu_{b+1}} \right ) \\
			 & \hphantom{= - \left ( - \gcoupling \right )^n \sum_{i+j+k = n-2}} \mkern-104mu \times \left ( \hat{\delta}^\kappa_{\mu_{i+j}} \hat{\delta}^\lambda_{\nu_{i+j+k+1}} \prod_{c = i+j}^{i+j+k} \hat{\eta}^{\mu_c \nu_{c+1}} \right ) \\
			 & \hphantom{= - \left ( - \gcoupling \right )^n \sum_{i+j+k = n-2}} \mkern-104mu \times \left ( \left ( p^{n-1}_{\mu_0} \hat{\delta}_\rho^{\mu_{n-1}} \hat{\delta}_\kappa^{\nu_{n-1}} \right ) \left ( \frac{1}{2} p^n_\lambda \hat{\delta}_\nu^{\mu_n} \hat{\delta}_\sigma^{\nu_n} - p^n_\nu \hat{\delta}_\lambda^{\mu_n} \hat{\delta}_\sigma^{\nu_n} \right ) \right . \\
			 & \hphantom{= - \left ( - \gcoupling \right )^n \sum_{i+j+k = n-2} \times (} \mkern-104mu \left . + \frac{1}{2} \left ( p^{n-1}_\nu \hat{\delta}_{\mu_0}^{\mu_{n-1}} \hat{\delta}_\kappa^{\nu_{n-1}} \right ) \left ( p^n_\sigma \hat{\delta}_\rho^{\mu_n} \hat{\delta}_\lambda^{\nu_n} \right ) \vphantom{\left ( \frac{1}{2} \right )} \right )
		\end{split}
		\intertext{and}
		\begin{split}
			\left ( \mathfrak{r}_n^{\Gamma^2} \right )^{\mu_1 \nu_1 \vert \cdots \vert \mu_n \nu_n} \left ( p_1^\sigma, \cdots, p_n^\sigma \right ) & = \\
			& \hphantom{=} \mkern-120mu - \left ( - \gcoupling \right )^n \sum_{i+j+k = n-2} \left ( \hat{\delta}^\rho_{\nu_{i+1}} \prod_{a = 0}^i \hat{\eta}^{\mu_a \nu_{a+1}} \right ) \left ( \hat{\delta}^\nu_{\mu_i} \hat{\delta}^\sigma_{\nu_{i+j+1}} \prod_{b = i}^{i+j} \hat{\eta}^{\mu_b \nu_{b+1}} \right ) \\
			& \hphantom{= - \left ( - \gcoupling \right )^n \sum_{i+j+k = n-2}} \mkern-152mu \times \left ( \hat{\delta}^\kappa_{\mu_{i+j}} \hat{\delta}^\lambda_{\nu_{i+j+k+1}} \prod_{c = i+j}^{i+j+k} \hat{\eta}^{\mu_c \nu_{c+1}} \right ) \\
			& \hphantom{= - \left ( - \gcoupling \right )^n \sum_{i+j+k = n-2}} \mkern-152mu \times \left ( \left ( p^{n-1}_\kappa \hat{\delta}_{\mu_0}^{\mu_{n-1}} \hat{\delta}_\rho^{\nu_{n-1}} \right ) \left ( \frac{1}{2} p^n_\nu \hat{\delta}_\sigma^{\mu_n} \hat{\delta}_\lambda^{\nu_n} - \frac{1}{4} p^n_\lambda \hat{\delta}_\nu^{\mu_n} \hat{\delta}_\sigma^{\nu_n} \right ) \right . \\
			& \hphantom{= - \left ( - \gcoupling \right )^n \sum_{i+j+k = n-2} \times (} \mkern-152mu \left . - \left ( p^{n-1}_\nu \hat{\delta}_{\mu_0}^{\mu_{n-1}} \hat{\delta}_\kappa^{\nu_{n-1}} \right ) \left ( \frac{1}{2} p^n_\rho \hat{\delta}_\sigma^{\mu_n} \hat{\delta}_\lambda^{\nu_n} - \frac{1}{4} p^n_\sigma \hat{\delta}_\rho^{\mu_n} \hat{\delta}_\lambda^{\nu_n} \right ) \vphantom{\left ( \frac{1}{2} \right )} \right ) \, .
		\end{split}
	\end{align}
	\end{subequations}
\end{lem}

\begin{proof}
	This follows directly from Corollaries \ref{col:ricci_scalar_for_the_levi_civita_connection_restriction} and \ref{col:inverse_metric_vielbeins_FR}. Furthermore, we remark the global minus sign due to the Fourier transform and the omission of Kronecker symbols, if possible.
\end{proof}

\enter

\begin{lem} \label{lem:de_Donder_gauge_fixing_FR}
	Introducing the notation
	\begin{equation}
		\deDonderFR_n^{\mu_1 \nu_1 \vert \cdots \vert \mu_n \nu_n} \left ( p_1^\sigma, \cdots, p_n^\sigma \right ) := \left ( \prod_{i = 1}^n \frac{\bar{\delta}}{\bar{\delta} \hat{h}_{\mu_i \nu_i}} \right ) \mathscr{F} \left ( \eval{\deDonder^2}_{\order{\gcoupling^{n}}} \right ) \, ,
	\end{equation}
	we obtain
	\begin{subequations}
	\begin{align}
		\deDonderFR_0 & = 0 \, , \\
		\deDonderFR_1^{\mu_1 \nu_1} \left ( p_1^\sigma \right ) & = 0
		\intertext{and for \(n > 1\)}
		\! \! \! \deDonderFR_n^{\mu_1 \nu_1 \vert \cdots \vert \mu_n \nu_n} \left ( p_1^\sigma, \cdots, p_n^\sigma \right ) & = \frac{1}{2^n} \sum_{\mu_i \leftrightarrow \nu_i} \sum_{s \in S_n} \deDonderpreFR_n^{\mu_{s(1)} \nu_{s(1)} \vert \cdots \vert \mu_{s(n)} \nu_{s(n)}} \left ( p_{s(1)}^\sigma, \cdots, p_{s(n)}^\sigma \right )
		\intertext{with}
		\begin{split}
			\! \! \! \deDonderpreFR_n^{\mu_1 \nu_1 \vert \cdots \vert \mu_n \nu_n} \left ( p_1^\sigma, \cdots, p_n^\sigma \right ) & = - \left ( - \gcoupling \right )^n \sum_{i+j+k = n-2} \left ( \hat{\delta}^\rho_{\nu_{i+1}} \prod_{a = 0}^i \hat{\eta}^{\mu_a \nu_{a+1}} \right ) \\ & \hphantom{=} \times \left ( \hat{\delta}^\nu_{\mu_i} \hat{\delta}^\sigma_{\nu_{i+j+1}} \prod_{b = i}^{i+j} \hat{\eta}^{\mu_b \nu_{b+1}} \right ) \left ( \hat{\delta}^\kappa_{\mu_{i+j}} \hat{\delta}^\lambda_{\nu_{i+j+k+1}} \prod_{c = i+j}^{i+j+k} \hat{\eta}^{\mu_c \nu_{c+1}} \right ) \\
			& \mkern-72mu \hphantom{=} \times \left ( \left ( p^{n-1}_\nu \hat{\delta}_\sigma^{\mu_{n-1}} \hat{\delta}_{\mu_0}^{\nu_{n-1}} \right ) \left ( p^n_\kappa \hat{\delta}_\lambda^{\mu_n} \hat{\delta}_\rho^{\nu_n} \right ) - \left ( p^{n-1}_\nu \hat{\delta}_\sigma^{\mu_{n-1}} \hat{\delta}_{\mu_0}^{\nu_{n-1}} \right ) \left ( p^n_\rho \hat{\delta}_\kappa^{\mu_n} \hat{\delta}_\lambda^{\nu_n} \right ) \right . \\ & \mkern-72mu \hphantom{= \times (} \left . + \frac{1}{4} \left ( p^{n-1}_{\mu_0} \hat{\delta}_\nu^{\mu_{n-1}} \hat{\delta}_\sigma^{\nu_{n-1}} \right ) \left ( p^n_\rho \hat{\delta}_\kappa^{\mu_n} \hat{\delta}_\lambda^{\nu_n} \right ) \right ) \, .
		\end{split}
	\end{align}
	\end{subequations}
	In particular, the quadratic part is given by (using momentum conservation, i.e.\ setting \(p_1^\sigma := p^\sigma\) and \(p_2^\sigma := - p^\sigma\))
	\begin{equation} \label{eqn:de_donder-quadratic}
	\begin{split}
		\deDonderFR_2^{\mu_1 \nu_1 \vert \mu_2 \nu_2} \left ( p^\sigma, - p^\sigma \right ) & = \gcoupling^2 \big ( p^{\mu_1} p^{\nu_1} \hat{\eta}^{\mu_2 \nu_2} + p^{\mu_2} p^{\nu_2} \hat{\eta}^{\mu_1 \nu_1} \big ) \\
		& \hphantom{=} - \frac{1}{2} \gcoupling^2 \big ( p^{\mu_1} p^{\mu_2} \hat{\eta}^{\nu_1 \nu_2} + p^{\mu_1} p^{\nu_2} \hat{\eta}^{\nu_1 \mu_2} + p^{\nu_1} p^{\mu_2} \hat{\eta}^{\mu_1 \nu_2} + p^{\nu_1} p^{\nu_2} \hat{\eta}^{\mu_1 \mu_2} \big ) \\
		& \hphantom{=} - \frac{1}{2} \gcoupling^2 \big ( p^2 \hat{\eta}^{\mu_1 \nu_1} \hat{\eta}^{\mu_2 \nu_2} \big ) \, .
	\end{split}
	\end{equation}
\end{lem}

\begin{proof}
	This follows directly from Corollaries \ref{col:metric_expression_for_de_donder_gauge_fixing_restriction} and \ref{col:inverse_metric_vielbeins_FR}. Furthermore, we remark the global minus sign due to the Fourier transform and the omission of Kronecker symbols, if possible.
\end{proof}

\enter

\begin{lem} \label{lem:Riemannian_volume_form_FR}
	Introducing the notation
	\begin{equation}
		\mathfrak{V}_n^{\mu_1 \nu_1 \vert \cdots \vert \mu_n \nu_n} := \left ( \prod_{i = 1}^n \frac{\delta}{\delta \hat{h}_{\mu_i \nu_i}} \right ) \mathscr{F} \left ( \eval{\sqrt{- \dt{g}}}_{\order{\gcoupling^{n}}} \right ) \, ,
	\end{equation}
	we obtain
	\begin{subequations}
	\begin{align}
		\mathfrak{V}_n^{\mu_1 \nu_1 \vert \cdots \vert \mu_n \nu_n} & = \frac{1}{2^n} \sum_{\mu_i \leftrightarrow \nu_i} \sum_{s \in S_n} \mathfrak{v}_n^{\mu_{s(1)} \nu_{s(1)} \vert \cdots \vert \mu_{s(n)} \nu_{s(n)}}
		\intertext{with}
		\begin{split}
			\mathfrak{v}_n^{\mu_1 \nu_1 \vert \cdots \vert \mu_n \nu_n} & = \gcoupling^n \sum_{\substack{i + j + k + l = m\\i \geq j \geq k \geq l \geq 0}} \sum_{p = 0}^{j - k} \sum_{q = 0}^{k - l} \sum_{r = 0}^{q} \sum_{s = 0}^l \sum_{t = 0}^s \sum_{u = 0}^t \sum_{v = 0}^u \\
			& \hphantom{=} \mkern-54mu \binom{\frac{1}{2}}{i} \binom{i}{j} \binom{j}{k} \binom{k}{l} \binom{j - k}{p} \binom{k - l}{q} \binom{q}{r} \binom{l}{s} \binom{s}{t} \binom{t}{u} \binom{u}{v} \\
			& \hphantom{=} \mkern-54mu \times \left ( - 1 \right )^{p + q - r + s - t + v} 2^{- j + l + r + s + 2t - 3u + v} 3^{- k + q - r + s - t + u} \\
			& \hphantom{=} \mkern-54mu \times \left ( \prod_{a = 1}^{\boldsymbol{a}} \hat{\eta}^{\mu_a \nu_a} \right ) \left ( \prod_{b = {\boldsymbol{a} + 1}}^{\boldsymbol{a} + \boldsymbol{b}} \hat{\eta}^{\mu_{b} \mu_{b + \boldsymbol{b}}} \hat{\eta}^{\nu_{b} \nu_{b + \boldsymbol{b}}} \right ) \left ( \prod_{c = {\boldsymbol{a} + 2 \boldsymbol{b} + 1}}^{\boldsymbol{a} + 2 \boldsymbol{b} + \boldsymbol{c}} \hat{\eta}^{\mu_{c} \nu_{c + \boldsymbol{c}}} \hat{\eta}^{\mu_{c + \boldsymbol{c}} \nu_{c + 2 \boldsymbol{c}}} \hat{\eta}^{\mu_{c + 2 \boldsymbol{c}} \nu_{c}} \right ) \\
			& \hphantom{=} \mkern-54mu \times \left ( \prod_{d = {\boldsymbol{a} + 2 \boldsymbol{b} + 3 \boldsymbol{c} + 1}}^{\boldsymbol{a} + 2 \boldsymbol{b} + 3 \boldsymbol{c} + \boldsymbol{d}} \hat{\eta}^{\mu_{d} \nu_{d + \boldsymbol{d}}} \hat{\eta}^{\mu_{d + \boldsymbol{d}} \nu_{d + 2 \boldsymbol{d}}} \hat{\eta}^{\mu_{d + 2 \boldsymbol{d}} \nu_{d + 3 \boldsymbol{d}}} \hat{\eta}^{\mu_{d + 3 \boldsymbol{d}} \nu_{d}} \right )
		\end{split}
		\intertext{and}
		\begin{split}
			\boldsymbol{a} & := i + j + k + l - 2p - 2q - r - 2s - t - u \\
			\boldsymbol{b} & := p + q - r + s - t + 2u - 2v \\
			\boldsymbol{c} & := r + t - u \\
			\boldsymbol{d} & := v \, .
		\end{split}
	\end{align}
	\end{subequations}
\end{lem}

\begin{proof}
	This follows directly from Corollaries \ref{col:determinant_metric_restriction} and \ref{col:inverse_metric_vielbeins_FR}.
\end{proof}

\subsection{Preparations for for gravitons and matter} \label{ssec:preparations_gravitons-matter}

In this subsection we prepare all necessary objects for the graviton-matter Feynman rules. As will be discussed in detail in the following four Subsubsections, the gravitational interactions with matter from the Standard Model can be classified into the following 10 Lagrange densities, henceforth refered to as matter-model Lagrange densities of type \(k\). We calculate only the gravitational interactions for the matter-model Lagrange densities and refer for the corresponding matter contributions to \cite{Romao_Silva} in order to keep this article at a reasonable length.

\enter

\begin{lem} \label{lem:matter-model-Lagrange-densities}
	Consider (effective) Quantum General Relativity coupled to the Standard Model (QGR-SM). Then the interaction Lagrange densities between gravitons and matter particles are of the following 10 types:\footnote{We remark that the tensors \(\tensor[_k]{\! T}{}\) are not related to Hilbert stress-energy tensors. More precisely, they are defined as the graviton-free matter contributions of the corresponding Lagrange densities.}
	{\allowdisplaybreaks
	\begin{align}
		\tensor[_1]{\mathcal{L}}{_{\text{QGR-SM}}} & := \tensor[_1]{\! T}{} \dif V_g \, , \\
		\tensor[_2]{\mathcal{L}}{_{\text{QGR-SM}}} & := \left ( g^{\mu \nu} \, \tensor[_2]{\! T}{_\mu _\nu} \right ) \dif V_g \, , \\
		\tensor[_3]{\mathcal{L}}{_{\text{QGR-SM}}} & := \left ( g^{\mu \nu} g^{\rho \sigma} \, \tensor[_3]{\! T}{_\mu _\nu _\rho _\sigma} \right ) \dif V_g \, , \\
		\tensor[_4]{\mathcal{L}}{_{\text{QGR-SM}}} & := \left ( g^{\mu \nu} \tensor{\Gamma}{_\mu _\nu ^\tau} \, \tensor[_4]{\! T}{_\tau} \right ) \dif V_g \, , \\
		\tensor[_5]{\mathcal{L}}{_{\text{QGR-SM}}} & := \left ( g^{\mu \nu} g^{\rho \sigma} \tensor{\Gamma}{_\mu _\nu ^\tau} \, \tensor[_5]{\! T}{_\rho _\sigma _\tau} \right ) \dif V_g \, , \\
		\tensor[_6]{\mathcal{L}}{_{\text{QGR-SM}}} & := \left ( g^{\mu \nu} g^{\rho \sigma} \tensor{\Gamma}{_\mu _\nu ^\kappa} \tensor{\Gamma}{_\rho _\sigma ^\lambda} \, \tensor[_6]{\! T}{_\kappa _\lambda} \right ) \dif V_g \, , \\
		\tensor[_7]{\mathcal{L}}{_{\text{QGR-SM}}} & := \left ( e_0^o \, \tensor[_7]{\! T}{_o} \right ) \dif V_g \, , \\
		\tensor[_8]{\mathcal{L}}{_{\text{QGR-SM}}} & := \left ( e_0^o e^{\rho r} \, \tensor[_8]{\! T}{_o _\rho _r} \right ) \dif V_g \, , \\
		\tensor[_9]{\mathcal{L}}{_{\text{QGR-SM}}} & := \left ( e_0^o e^{\rho r} e^{\sigma s} \left ( \partial_\rho e_\sigma^t \right ) \, \tensor[_9]{\! T}{_o _r _s _t} \right ) \dif V_g \, , \\
		\intertext{and}
		\tensor[_{10}]{\mathcal{L}}{_{\text{QGR-SM}}} & := \left ( e_0^o e^{\rho r} e^{\sigma s} e_\tau^t \tensor{\Gamma}{_\rho _\sigma ^\tau} \, \tensor[_{10}]{\! T}{_o _r _s _t} \right ) \dif V_g \, .
	\end{align}
	}
\end{lem}

\begin{proof}
	A direct computation shows that the scalar particles form the Standard Model are of type 1 and 2. Furthermore, the spinor particles from the Standard Model are of type 7, 8, 9 and 10. Moreover, the bosonic gauge boson particles from the Standard Model are of type 3, 5 and 6 and additionally 1, 2 and 4 for spontaneous symmetry breaking. Finally, the gauge ghosts are of type 2 and 6 and additionally 1 for spontaneous symmetry breaking. This is discussed in detail in the following four Subsubsections.
\end{proof}

\subsubsection{Gravitons and scalar particles} \label{sssec:gravitons_and_scalar_particles}

Scalar particles from the Standard Model are the Higgs and Goldstone bosons.\footnote{The gauge ghosts are discussed in \sssecref{sssec:gravitons_and_gauge_ghosts}.} In the following we describe the interaction of gravitons with a real scalar field and with a vector of complex scalar fields, subjected to the action of a gauge group \(G\) (which leads to spontaneous symmetry breaking). Geometrically they are described via sections \(\phi \in \Gamma \left ( \bbM, \mathbb{R} \right )\) and \(\Phi \in \Gamma \left ( \bbM, \mathbb{C}^i \right )\), respectively, where \(i = \operatorname{DimRep} \left ( \rho \right )\) is the dimension of the representation of the gauge group \(G\) on \(\mathbb{C}^i\), acting fiberwise via \(\rho \colon G \times \mathbb{C}^i \to \mathbb{C}^i\). Then, the corresponding Lagrange densities are given by
\begin{align}
	\mathcal{L}_\text{GR-\(\mathbb{R}\)-Scalar} & = \left ( \frac{1}{2} g^{\mu \nu} \left ( \partial_\mu \phi \right ) \left ( \partial_\nu \phi \right ) + \sum_{i \in \boldsymbol{I}_\phi} \frac{\alpha_i}{i!} \phi^i \right ) \dif V_g
	\intertext{and}
	\mathcal{L}_\text{GR-\(\mathbb{C}^i\)-Scalar} & = \left ( g^{\mu \nu} \big ( \nabla^{G \times_\rho \mathbb{C}^i}_\mu \Phi \big )^\dagger \big ( \nabla^{G \times_\rho \mathbb{C}^i}_\nu \Phi \big ) + \sum_{i \in \boldsymbol{I}_\Phi} \frac{\alpha_i}{i!} \big ( \Phi^\dagger \Phi \big )^i \right ) \dif V_g \, , \label{eqn:vector_complex_scalar_field}
\end{align}
where \(\boldsymbol{I}_\phi\) and \(\boldsymbol{I}_\Phi\) denote the interaction sets with particle mass \(- \alpha_2\) and coupling constants \(\alpha_i\) for \(i \neq 2\), \(\dagger\) denotes Hermitian conjugation and
\begin{equation}
	\nabla^{G \times_\rho \mathbb{C}^i}_\mu := \partial_\mu + \imaginary \mathrm{g} A^a_\mu \mathfrak{b}_a
\end{equation}
is the covariant derivative on the \(\mathbb{C}^i\)-bundle, with connection form \(\imaginary \mathrm{g} A \in \Gamma \big ( \bbM, T^* \bbM \otimes \operatorname{End} ( \mathbb{C}^i ) \big )\). The Higgs bundle from the Standard Model is of the form \eqnref{eqn:vector_complex_scalar_field} with further interactions coming from the gauge fixing of the corresponding Electroweak gauge bosons, cf.\ \sssecref{sssec:gravitons_and_gauge_bosons}. These interactions correspond to type 1 and 2 from \lemref{lem:matter-model-Lagrange-densities}. More precisely, we have
\begin{align}
	\tensor[_1]{\! T}{} & := \sum_{i \in \boldsymbol{I}_\phi} \frac{\alpha_i}{i!} \phi^i + \sum_{i \in \boldsymbol{I}_\Phi} \frac{\alpha_i}{i!} \big ( \Phi^\dagger \Phi \big )^i
	\intertext{and}
	\tensor[_2]{\! T}{_\mu _\nu} & := \frac{1}{2} \left ( \partial_\mu \phi \right ) \left ( \partial_\nu \phi \right ) + \big ( \nabla^{G \times_\rho \mathbb{C}^i}_\mu \Phi \big )^\dagger \big ( \nabla^{G \times_\rho \mathbb{C}^i}_\nu \Phi \big ) \, .
\end{align}

\subsubsection{Gravitons and spinor particles}

Spinor particles from the Standard Model are leptons and quarks. In the following we describe the interaction of gravitons with spinor fields and with a vector of spinor fields, subjected to the action of a gauge group \(G\). Geometrically they are described via sections \(\psi \in \Gamma \left ( \bbM, \Sigma \bbM \right )\) and \(\Psi \in \Gamma \left ( \bbM, \Sigma \bbM^{\oplus j} \right )\), respectively, where \(j = \operatorname{DimRep} \left ( \varrho \right )\) is the dimension of the representation of the gauge group \(G\) on \(\Sigma M^{\oplus j}\), acting fiberwise via \(\varrho \colon G \times \Sigma M^{\oplus j} \to \Sigma M^{\oplus j}\). The corresponding dual spinor fields are defined via
\begin{align}
	\overline{\psi} & := e_0^o \left ( \gamma_o \psi \right )^\dagger
	\intertext{and}
	\overline{\Psi} & := e_0^o \left ( \boldsymbol{\gamma}_o \Psi \right )^\dagger \, ,
\end{align}
where \(e_0^o\) is a vielbein with its curved index fixed to \(\mu \equiv 0\) and flat index \(o\), i.e.\ a vielbein contracted with the normalized timelike vector field \(e \left ( \dif t \right )\), and \(\gamma_m\) and \(\boldsymbol{\gamma}_m\) are the Dirac matrices for the Minkowski background metric \(\eta\) on \(\bbM\) and \(\bbM^{\oplus j}\), respectively. Thus, dual spinor fields depend on the metric via the vielbein \(e_0^o\) with fixed timelike curved index.\footnote{We emphasize the placement of \(\gamma_o\) and \(\boldsymbol{\gamma}_o\) in the following equations, as only the timelike Dirac matrices \(\gamma_0\) and \(\boldsymbol{\gamma}_0\) are hermitian, whereas the other Dirac matrices are antihermitian.} We remark that if the spacetime \((M,\met)\) is globally hyperbolic, it is possible to choose charts in which \(e_0^o \equiv \delta_0^o\), as is done implicitly in e.g.\ \cite{Choi_Shim_Song,Schuster,Rodigast_Schuster_1,Rodigast_Schuster_2}. However it should be noted that in this setting the theory is no longer invariant under general diffeomorphisms, but only under the subgroup of diffeomorphisms preserving global hyperbolicity. As we do not want to restrict our analysis to such charts and diffeomorphisms, we set
\begin{equation}
	\overline{\psi}_o := \left ( \gamma_o \psi \right )^\dagger
\end{equation}
for later use. Then, the corresponding Lagrange densities are given by
\begin{align}
	\mathcal{L}_\text{GR-Spinor} & = \Big ( \overline{\psi} \big ( \imaginary \slashed{\nabla}^{\Sigma M} - m_\psi \big ) \psi \Big ) \dif V_g
	\intertext{and}
	\mathcal{L}_\text{GR-Spinor\(^j\)} & = \Big ( \overline{\Psi} \big ( \imaginary \slashed{\nabla}^{G \times_\varrho \Sigma M^{\oplus j}} - \boldsymbol{m}_\Psi \big ) \Psi \Big ) \dif V_g \, ,
\end{align}
where \(\boldsymbol{m}_\Psi\) is a diagonal \(j \times j\)-matrix with entries given via the corresponding spinor particle masses, and with the Dirac operators given via
\begin{align}
	\slashed{\nabla}^{\Sigma M} & := e^{\mu m} \gamma_m \left ( \partial_\mu + \varpi_\mu \right )
	\intertext{and}
	\slashed{\nabla}^{G \times_\varrho \Sigma M^{\oplus j}} & := e^{\mu m} \boldsymbol{\gamma}_m \left ( \partial_\mu + \varpi_\mu \right ) + e^{\mu m} \gamma_m \left ( \imaginary \mathrm{g} A^a_\mu \mathfrak{b}_a \right ) \, ,
\end{align}
where \(\varpi_\mu \in \Gamma \left ( \bbM, T^* \bbM \otimes \operatorname{End} \left ( \Sigma M \right ) \right )\) is the spin connection form and \(\imaginary \mathrm{g} A \in \Gamma \big ( \bbM, T^* \bbM \otimes \operatorname{End} \left ( \Sigma M^{\oplus j} \right ) \! \big )\) the corresponding gauge group connection form. These interactions correspond to type 7, 8, 9 and 10 from \lemref{lem:matter-model-Lagrange-densities}. More precisely, we have
\begin{align}
	\tensor[_7]{\! T}{_o} & := - m_\psi \overline{\psi}_o \psi - \overline{\Psi} \boldsymbol{m}_\Psi \Psi \, , \\
	\tensor[_8]{\! T}{_o _\rho _r} & := \overline{\psi}_o \gamma_r \left ( \partial_\rho \psi \right ) + \overline{\Psi}_o \boldsymbol{\gamma}_r \left ( \partial_\rho \Psi \right ) \, , \\
	\tensor[_9]{\! T}{_o _r _s _t} & := - \frac{\imaginary}{4} \overline{\psi}_o \left ( \gamma_r \sigma_{s t} \right ) \psi - \frac{\imaginary}{4} \overline{\Psi}_o \left ( \boldsymbol{\gamma}_r \boldsymbol{\sigma}_{s t} \right ) \Psi \, ,
	\intertext{with \(\sigma_{s t} := \frac{\imaginary}{2} \left [ \gamma_s, \gamma_t \right ]\) and \(\boldsymbol{\sigma}_{s t} := \frac{\imaginary}{2} \left [ \boldsymbol{\gamma}_s, \boldsymbol{\gamma}_t \right ]\), and}
	\begin{split}
		\tensor[_{10}]{\! T}{_o _r _s _t} & := - \frac{\imaginary}{4} \overline{\psi}_o \left ( \gamma_r \sigma_{s t} \right ) \psi - \frac{\imaginary}{4} \overline{\Psi}_o \left ( \boldsymbol{\gamma}_r \boldsymbol{\sigma}_{s t} \right ) \Psi \\
		& \hphantom{:} \equiv \tensor[_9]{\! T}{_o _r _s _t} \, .
	\end{split}
\end{align}
We remark that the interaction of leptons and quarks with the Higgs and Goldstone bosons are given by
\begin{equation}
	\mathcal{L}_\text{Yukawa} = - \left ( \sum_{\{ \phi, \overline{\psi}_o, \psi \} \in \boldsymbol{I}_Y} \alpha_{\{ \phi, \overline{\psi}_o, \psi \}} \phi \overline{\psi}_o \psi \right ) \dif V_g
\end{equation}
which represent the Yukawa interaction terms for the interaction set \(\boldsymbol{I}_Y\), with corresponding coupling constants \(\alpha_{\{ \phi, \overline{\psi}_o, \psi \}}\). These interactions are of type 7 from \lemref{lem:matter-model-Lagrange-densities}. More precisely, we have
\begin{equation}
	\tensor[_7]{\! T}{_o} := - \sum_{\{ \phi, \overline{\psi}_o, \psi \} \in \boldsymbol{I}_Y} \alpha_{\{ \phi, \overline{\psi}_o, \psi \}} \phi \overline{\psi}_o \psi \, .
\end{equation}

\subsubsection{Gravitons and gauge bosons} \label{sssec:gravitons_and_gauge_bosons}

Gauge bosons from the Standard Model are the photon, the \(Z\)- and \(W^\pm\)-bosons, and the gluons. In the following we describe the interaction of gravitons with gauge bosons from a Quantum Yang--Mills theory. We denote the Yang--Mills gauge group by \(G\) and its Lie algebra by \(\mathfrak{g}\). Geometrically, gauge bosons are described via connection forms \(\imaginary \mathrm{g} A \in \Gamma \left ( \bbM, T^* \bbM \otimes \mathfrak{g} \right )\) on the underlying principle bundle. More precisely, they are given as the components with respect to a basis choice \(\set{\mathfrak{b}_a}\) on \(\mathfrak{g}\). Then, the corresponding Lagrange densities are given by\footnote{We emphasize the minus sign coming from the square of \(F^a_{\mu \nu} := \imaginary \mathrm{g} \big ( \partial_\mu A^a_\nu - \partial_\nu A^a_\mu - \mathrm{g} f^{abc} A^a_\mu A^b_\nu \big )\) in our conventions. Furthermore, we remark that this obviously also includes abelian gauge theories, such as electrodynamics, by setting \(\mathfrak{g}\) to be abelian, i.e.\ \(f^{abc} \equiv 0\). \label{ftn:YM-ED}}
\begin{align}
	\mathcal{L}_\text{GR-YM} & = \left ( \frac{1}{4 \mathrm{g}^2} \delta_{ab} g^{\mu \nu} g^{\rho \sigma} F^a_{\mu \rho} F^b_{\nu \sigma} \right ) \dif V_g
	\intertext{and the Lorenz gauge fixing by\footnotemark}
	\mathcal{L}_\text{GR-YM-GF} & = \left ( \frac{1}{2 \xi} \delta_{ab} g^{\mu \nu} g^{\rho \sigma} \big ( \nabla^{(g)}_\mu A^a_\nu \big ) \big ( \nabla^{(g)}_\rho A^b_\sigma \big ) \right ) \dif V_g \, . \label{eqn:GR-YM-GF}
\end{align}
\footnotetext{It is convenient to use the covariant Lorenz gauge fixing \(g^{\mu \nu} \nabla^{(g)}_\mu A^a_\nu \overset{!}{=} 0\), as this choice avoids couplings from graviton-ghosts to gauge ghosts \cite{Prinz_5}.}
\noindent These interactions correspond to type 3, 5 and 6 from \lemref{lem:matter-model-Lagrange-densities}. More precisely, we have\footnote{We remark the minus sign due to the covariant derivative on forms and the additional factor of 2 due to the binomial theorem in \eqnref{eqn:gauge_boson_vii}.}
\begin{align}
	\tensor[_3]{\! T}{_\mu _\nu _\rho _\sigma} & := \frac{1}{4 \mathrm{g}^2} \delta_{ab} F^a_{\mu \rho} F^b_{\nu \sigma} + \frac{1}{2 \xi} \delta_{ab} \big ( \partial_\mu A^a_\nu \big ) \big ( \partial_\rho A^b_\sigma \big ) \, , \\
	\tensor[_5]{\! T}{_\mu _\nu _\tau} & := - \frac{1}{\xi} \delta_{ab} \big ( \partial_\mu A^a_\nu \big ) A^b_\tau \label{eqn:gauge_boson_vii}
	\intertext{and}
	\tensor[_6]{\! T}{_\kappa _\lambda} & := \frac{1}{2 \xi} \delta_{ab} A^a_\kappa A^b_\lambda \, .
\end{align}
We remark that the Lorenz gauge fixing Lagrange densities for the \(Z\)- and \(W^\pm\)-bosons need slight modifications due to the spontaneous symmetry breaking and are given by
\begin{align}
	\mathcal{L}_\text{\(Z\)-Boson-GF} & = \left ( \frac{1}{2 \xi_Z} g^{\mu \nu} g^{\rho \sigma} \big ( \nabla^{(g)}_\mu Z_\nu ) \big ( \nabla^{(g)}_\rho Z_\sigma \big ) + m_Z \phi_Z g^{\mu \nu} \big ( \nabla^{(g)}_\mu Z_\nu \big ) + \frac{\xi_Z}{2} m_Z^2 \phi_Z^2 \right ) \dif V_g
	\intertext{and}
	\begin{split}
		\mathcal{L}_\text{\(W\)-Boson-GF} & = \left ( \frac{1}{\xi_W} g^{\mu \nu} g^{\rho \sigma} \big ( \nabla^{(g)}_\mu W^-_\nu \big ) \big ( \nabla^{(g)}_\rho W^+_\sigma \big ) + \xi_W m_W^2 \phi_{W^-} \phi_{W^+} \right . \\ & \phantom{= (} \left . + \imaginary m_W g^{\mu \nu} \left ( \phi_{W^+} \big ( \nabla^{(g)}_\mu W^-_\nu \big ) - \phi_{W^-} \big ( \nabla^{(g)}_\mu W^+_\nu \big ) \right ) \right ) \dif V_g \, ,
	\end{split}
\end{align}
where \(\xi_s\) is the corresponding gauge fixing parameter and \(m_s\) the corresponding mass for \(s \in \set{Z, W^+, W^-}\), and \(\phi_Z\), \(\phi_{W^+}\) and \(\phi_{W^-}\) are the Goldstone bosons. These interactions additionally require type 1, 2 and 4 from \lemref{lem:matter-model-Lagrange-densities}. More precisely, we have
\begin{align}
	\tensor[_1]{\! T}{} & := \frac{\xi_Z}{2} m_Z^2 \phi_Z^2 + \xi_W m_W^2 \phi_{W^-} \phi_{W^+} \, , \\
	\tensor[_2]{\! T}{_\mu _\nu} & := \xi_Z m_Z \phi_Z \big ( \partial_\mu Z_\nu \big ) + \imaginary \xi_W m_W \left ( \phi_{W^+} \big ( \partial_\mu W^-_\nu \big ) - \phi_{W^-} \big ( \partial_\mu W^+_\nu \big ) \right ) \, , \\
	\tensor[_3]{\! T}{_\mu _\nu _\rho _\sigma} & := \frac{1}{2 \xi_Z} \big ( \partial_\mu Z_\nu \big ) \big ( \partial_\rho Z_\sigma \big ) + \frac{1}{\xi_W} \big ( \partial_\mu W^-_\nu \big ) \big ( \partial_\rho W^+_\sigma \big )
	\intertext{and}
	\tensor[_4]{\! T}{_\tau} & := \left ( \xi_s m_s \right ) \phi^s A^{-s}_\tau \, .
\end{align}
We refer to \sssecref{sssec:gravitons_and_scalar_particles} for further interactions between \(Z\)- and \(W^\pm\)-bosons and Higgs and Goldstone bosons coming from the covariant derivative on the Higgs bundle.

\subsubsection{Gravitons and gauge ghosts} \label{sssec:gravitons_and_gauge_ghosts}

Gauge ghosts and gauge antighosts from the Standard Model, accompanying their corresponding gauge bosons \(\imaginary \mathrm{g} A \in \Gamma \left ( \bbM, T^* \bbM \otimes \mathfrak{g} \right )\), are fermionic \(\mathfrak{g}\)-valued scalar particles \(c \in \Gamma \left ( \bbM, \Pi \left ( \mathfrak{g} \right ) \right )\) and \(\overline{c} \in \Gamma \left ( \bbM, \Pi \left ( \mathfrak{g}^* \right ) \right )\). Then, the corresponding Lagrange density is given by
\begin{equation}
	\mathcal{L}_\text{GR-YM-Ghost} = \left ( g^{\mu \nu} \overline{c}_a \big ( \nabla^{(g)}_\mu \partial_\nu c^a \big ) + \imaginary \mathrm{g} g^{\mu \nu} \tensor{f}{^a _b _c} \overline{c}_a \big ( \nabla^{(g)}_\mu A^b_\nu c^c \big ) \right ) \dif V_g \, .
\end{equation}
These interactions correspond to type 2 and 4 from \lemref{lem:matter-model-Lagrange-densities}. More precisely, we have\footnote{The ghost Lagrange densities are calculated with Faddeev--Popov's method \cite{Faddeev_Popov}, cf.\ \cite[Subsubsection 2.2.3]{Prinz_2}. We mention that this construction can be embedded into a more general context, using BRST and anti-BRST operators \cite{Baulieu_Thierry-Mieg}.}
\begin{align}
	\tensor[_2]{\! T}{_\mu _\nu} & := \overline{c}_a \big ( \partial_\mu \partial_\nu c^a \big ) + \imaginary \mathrm{g} \tensor{f}{^a _b _c} \overline{c}_a \big ( \partial_\mu A^b_\nu c^c \big )
	\intertext{and}
	\tensor[_4]{\! T}{_\tau} & := \overline{c}_a \big ( \partial_\tau c^a \big ) + \imaginary \mathrm{g} \tensor{f}{^a _b _c} \overline{c}_a A^b_\tau c^c \, .
\end{align}
We remark that the interaction of Electroweak gauge ghosts with the Higgs and Goldstone bosons are given by
\begin{equation}
	\mathcal{L}_\text{EW-Ghost} = \left ( \sum_{\set{s_1, s_2, s_3} \in \boldsymbol{I}_\text{EW-Ghost}} \left ( \xi_{s_2} m_{s_2} \right ) \phi^{s_1} \overline{c}^{s_2} c^{s_3} \right ) \dif V_g \, ,
\end{equation}
where \(\xi_{s_i}\) is the corresponding gauge fixing parameter, \(m_{s_i}\) the corresponding mass for \(s_i \in \set{A, Z, W^+, W^-, H}\) and \(\boldsymbol{I}_\text{EW-Ghost}\) is the corresponding interaction set. These interactions are of type 1 from \lemref{lem:matter-model-Lagrange-densities}. More precisely, we have
\begin{equation}
	\tensor[_1]{\! T}{} := \sum_{\set{s_1, s_2, s_3} \in \boldsymbol{I}_\text{EW-Ghost}} \left ( \xi_{s_2} m_{s_2} \right ) \phi^{s_1} \overline{c}^{s_2} c^{s_3} \, .
\end{equation}
We comment that with our chosen covariant Lorenz gauge fixing in \eqnref{eqn:GR-YM-GF} there are no interactions between graviton-ghosts and gauge ghosts present. This is due to the fact, that the gauge-fixing Lagrange density is a tensor density of weight 1, cf.\ \cite{Prinz_5}.

\subsection{Feynman rules for gravitons and their ghosts}

Having done all preparations in \ssecref{ssec:preparations_graviton_graviton_ghost}, we now list the corresponding Feynman rules for gravitons and their ghosts.

\enter

\begin{thm} \label{thm:grav-fr}
	Given the metric decomposition \(g_{\mu \nu} = \eta_{\mu \nu} + \gcoupling h_{\mu \nu}\) and assume \(\left | \gcoupling \right | \left \| h \right \|_{\max} := \left | \gcoupling \right | \max_{\lambda \in \operatorname{EW} \left ( h \right )} \left | \lambda \right | < 1\), where \(\operatorname{EW} \left ( h \right )\) denotes the set of eigenvalues of \(h\). Then the graviton \(2\)-point vertex Feynman rule for (effective) Quantum General Relativity reads (where \(\zeta\) denotes the gauge parameter and we use momentum conservation on the quadratic term, i.e.\ set \(p_1^\sigma := p^\sigma\) and \(p_2^\sigma := - p^\sigma\)):
		\begin{equation}
		\begin{split}
			\gravfr_2^{\mu_1 \nu_1 \vert \mu_2 \nu_2} \left ( p^\sigma; \zeta \right ) & = \frac{\imaginary}{4} \left ( 1 - \frac{1}{\zeta} \right ) \left ( p^{\mu_1} p^{\nu_1} \hat{\eta}^{\mu_2 \nu_2} + p^{\mu_2} p^{\nu_2} \hat{\eta}^{\mu_1 \nu_1} \right ) \\
			& \hphantom{=} \mkern-55mu - \frac{\imaginary}{8} \left ( 1 - \frac{1}{\zeta} \right ) \left ( p^{\mu_1} p^{\mu_2} \hat{\eta}^{\nu_1 \nu_2} + p^{\mu_1} p^{\nu_2} \hat{\eta}^{\nu_1 \mu_2} + p^{\nu_1} p^{\mu_2} \hat{\eta}^{\mu_1 \nu_2} + p^{\nu_1} p^{\nu_2} \hat{\eta}^{\mu_1 \mu_2} \right ) \\
			& \hphantom{=} \mkern-55mu - \frac{\imaginary}{4} \left ( 1 - \frac{1}{2 \zeta} \right ) \left ( p^2 \hat{\eta}^{\mu_1 \nu_1} \hat{\eta}^{\mu_2 \nu_2} \right ) \\
			& \hphantom{=} \mkern-55mu + \frac{\imaginary}{8} \left ( p^2 \hat{\eta}^{\mu_1 \mu_2} \hat{\eta}^{\nu_1 \nu_2} + p^2 \hat{\eta}^{\mu_1 \nu_2} \hat{\eta}^{\nu_1 \mu_2} \right )
		\end{split}
		\end{equation}
	Furthermore, the graviton \(n\)-point vertex Feynman rules with \(n > 2\) for (effective) Quantum General Relativity read (where \(\pregravfr_n\) denotes the corresponding unsymmetrized Feynman rules and \(\boldsymbol{\delta}_{m_1 \neq n}\) is set to \(0\) if \(m_1 = n\) and to \(1\) else, eliminating contributions coming from total derivatives):
		\begin{subequations}
		\begin{align}
			\gravfr_n^{\mu_1 \nu_1 \vert \cdots \vert \mu_n \nu_n} \left ( p_1^\sigma, \cdots, p_n^\sigma \right ) & = \frac{\imaginary}{2^n} \sum_{\mu_i \leftrightarrow \nu_i} \sum_{s \in S_n} \pregravfr_n^{\mu_{s(1)} \nu_{s(1)} \vert \cdots \vert \mu_{s(n)} \nu_{s(n)}} \left ( p_{s(1)}^\sigma, \cdots, p_{s(n)}^\sigma \right )
			\intertext{with}
			\begin{split}
				\pregravfr_n^{\mu_1 \nu_1 \vert \cdots \vert \mu_n \nu_n} \left ( p_1^\sigma, \cdots, p_n^\sigma \right ) & = \\
				& \mkern-210mu \frac{\left ( - \gcoupling \right )^{n-2}}{2} \sum_{m_1 + m_2 = n} \left \{ \sum_{i = 0}^{m_1 - 1} \left ( \hat{\delta}^\mu_{\mu_0} \hat{\delta}^\rho_{\nu_{i+1}} \prod_{a = 0}^i \hat{\eta}^{\mu_a \nu_{a+1}} \right ) \left ( \hat{\delta}^\nu_{\mu_i} \hat{\delta}^\sigma_{\nu_{m_1}} \prod_{b = i}^{m_1 - 1} \hat{\eta}^{\mu_b \nu_{b+1}} \right ) \right . \\
				& \mkern-210mu \hphantom{\frac{\left ( - \gcoupling \right )^{n-2}}{2} \sum \{ \sum} \times \boldsymbol{\delta}_{m_1 \neq n} \left [ p^{m_1}_\mu p^{m_1}_\nu \hat{\delta}_\rho^{\mu_{m_1}} \hat{\delta}_\sigma^{\nu_{m_1}} - p^{m_1}_\mu p^{m_1}_\rho \hat{\delta}_\nu^{\mu_{m_1}} \hat{\delta}_\sigma^{\nu_{m_1}} \right ] \\
				& \mkern-210mu \hphantom{\frac{\left ( - \gcoupling \right )^{n-2}}{2} \sum} - \sum_{j + k + l = m_1 - 2} \left ( \hat{\delta}^\mu_{\mu_0} \hat{\delta}^\rho_{\nu_{j+1}} \prod_{a = 0}^j \hat{\eta}^{\mu_a \nu_{a+j}} \right ) \left ( \hat{\delta}^\nu_{\mu_j} \hat{\delta}^\sigma_{\nu_{j+k+1}} \prod_{b = j}^{j+k} \hat{\eta}^{\mu_b \nu_{b+1}} \right ) \\
				& \mkern-210mu \hphantom{\frac{\left ( - \gcoupling \right )^{n-2}}{2} \sum +} \times \left ( \hat{\delta}^\kappa_{\mu_{j+k}} \hat{\delta}^\lambda_{\nu_{m_1 - 1}} \prod_{c = j+k}^{m_1 - 2} \hat{\eta}^{\mu_c \nu_{c+1}} \right ) \\
				& \mkern-210mu \hphantom{\frac{\left ( - \gcoupling \right )^{n-2}}{2} \sum +} \times \left ( \boldsymbol{\delta}_{m_1 \neq n} \left [ \left ( p^{n-1}_{\mu} \hat{\delta}_\rho^{\mu_{n-1}} \hat{\delta}_\kappa^{\nu_{n-1}} \right ) \left ( \frac{1}{2} p^n_\lambda \hat{\delta}_\nu^{\mu_n} \hat{\delta}_\sigma^{\nu_n} - p^n_\nu \hat{\delta}_\lambda^{\mu_n} \hat{\delta}_\sigma^{\nu_n} \right ) \right . \right . \\
				& \mkern-210mu \hphantom{\frac{\left ( - \gcoupling \right )^{n-2}}{2} \sum + \times ( + \boldsymbol{\delta}_{m_1 \neq n} [} \left . + \frac{1}{2} \left ( p^{n-1}_\nu \hat{\delta}_{\mu}^{\mu_{n-1}} \hat{\delta}_\kappa^{\nu_{n-1}} \right ) \left ( p^n_\sigma \hat{\delta}_\rho^{\mu_n} \hat{\delta}_\lambda^{\nu_n} \right ) \vphantom{\left ( \frac{1}{2} \right )} \right ] \\
				& \mkern-210mu \hphantom{\frac{\left ( - \gcoupling \right )^{n-2}}{2} \sum + \times (} + \left ( p^{n-1}_\kappa \hat{\delta}_{\mu}^{\mu_{n-1}} \hat{\delta}_\rho^{\nu_{n-1}} \right ) \left ( \frac{1}{2} p^n_\nu \hat{\delta}_\sigma^{\mu_n} \hat{\delta}_\lambda^{\nu_n} - \frac{1}{4} p^n_\lambda \hat{\delta}_\nu^{\mu_n} \hat{\delta}_\sigma^{\nu_n} \right ) \\
				& \mkern-210mu \hphantom{\frac{\left ( - \gcoupling \right )^{n-2}}{2} \sum + \times (} - \left . \left ( p^{n-1}_\nu \hat{\delta}_{\mu}^{\mu_{n-1}} \hat{\delta}_\kappa^{\nu_{n-1}} \right ) \left ( \frac{1}{2} p^n_\rho \hat{\delta}_\sigma^{\mu_n} \hat{\delta}_\lambda^{\nu_n} - \frac{1}{4} p^n_\sigma \hat{\delta}_\rho^{\mu_n} \hat{\delta}_\lambda^{\nu_n} \right ) \right \} \\
				& \mkern-210mu \hphantom{\frac{\left ( - \gcoupling \right )^{n-2}}{2} \sum \{} \boldsymbol{\times} \left \{ \sum_{\substack{i + j + k + l = m_2\\i \geq j \geq k \geq l \geq 0}} \sum_{p = 0}^{j - k} \sum_{q = 0}^{k - l} \sum_{r = 0}^{q} \sum_{s = 0}^l \sum_{t = 0}^s \sum_{u = 0}^t \sum_{v = 0}^u \right . \\
				& \mkern-210mu \hphantom{\frac{\left ( - \gcoupling \right )^{n-2}}{2} \sum \{ \boldsymbol{\times}} \binom{\frac{1}{2}}{i} \binom{i}{j} \binom{j}{k} \binom{k}{l} \binom{j - k}{p} \binom{k - l}{q} \binom{q}{r} \binom{l}{s} \binom{s}{t} \binom{t}{u} \binom{u}{v} \\
				& \mkern-210mu \hphantom{\frac{\left ( - \gcoupling \right )^{n-2}}{2} \sum \{ \boldsymbol{\times}} \times \left ( - 1 \right )^{p + q - r + s - t + v} 2^{- j + l + r + s + 2t - 3u + v} 3^{- k + q - r + s - t + u} \\
				& \mkern-210mu \hphantom{\frac{\left ( - \gcoupling \right )^{n-2}}{2} \sum \{ \boldsymbol{\times}} \times \left ( \prod_{a = m_1 + 1}^{m_1 + \boldsymbol{a}} \hat{\eta}^{\mu_a \nu_a} \right ) \left ( \prod_{b = m_1 + \boldsymbol{a} + 1}^{m_1 + \boldsymbol{a} + \boldsymbol{b}} \hat{\eta}^{\mu_{b} \mu_{b + \boldsymbol{b}}} \hat{\eta}^{\nu_{b} \nu_{b + \boldsymbol{b}}} \right ) \\
				& \mkern-210mu \hphantom{\frac{\left ( - \gcoupling \right )^{n-2}}{2} \sum \{ \boldsymbol{\times}} \times \left ( \prod_{c = m_1 + \boldsymbol{a} + 2 \boldsymbol{b} + 1}^{m_1 + \boldsymbol{a} + 2 \boldsymbol{b} + \boldsymbol{c}} \hat{\eta}^{\mu_{c} \nu_{c + \boldsymbol{c}}} \hat{\eta}^{\mu_{c + \boldsymbol{c}} \nu_{c + 2 \boldsymbol{c}}} \hat{\eta}^{\mu_{c + 2 \boldsymbol{c}} \nu_{c}} \right ) \\
				& \mkern-210mu \hphantom{\frac{\left ( - \gcoupling \right )^{n-2}}{2} \sum \{ \boldsymbol{\times}} \left . \times \left ( \prod_{d = m_1 + \boldsymbol{a} + 2 \boldsymbol{b} + 3 \boldsymbol{c} + 1}^{m_1 + \boldsymbol{a} + 2 \boldsymbol{b} + 3 \boldsymbol{c} + \boldsymbol{d}} \hat{\eta}^{\mu_{d} \nu_{d + \boldsymbol{d}}} \hat{\eta}^{\mu_{d + \boldsymbol{d}} \nu_{d + 2 \boldsymbol{d}}} \hat{\eta}^{\mu_{d + 2 \boldsymbol{d}} \nu_{d + 3 \boldsymbol{d}}} \hat{\eta}^{\mu_{d + 3 \boldsymbol{d}} \nu_{d}} \right ) \right \}
			\end{split}
			\intertext{and}
			\begin{split}
				\boldsymbol{a} & := i + j + k + l - 2p - 2q - r - 2s - t - u \\
				\boldsymbol{b} & := p + q - r + s - t + 2u - 2v \\
				\boldsymbol{c} & := r + t - u \\
				\boldsymbol{d} & := v
			\end{split}
		\end{align}
		\end{subequations}
\end{thm}

\begin{proof}
	This follows from the combination of Lemmata \ref{lem:traces_FR}, \ref{lem:Ricci_scalar_FR}, \ref{lem:de_Donder_gauge_fixing_FR} and \ref{lem:Riemannian_volume_form_FR}, since we have
	\begin{subequations}
	\begin{align}
		\gravfr_n^{\mu_1 \nu_1 \vert \cdots \vert \mu_n \nu_n} \left ( p_1^\sigma, \cdots, p_n^\sigma \right ) & = \frac{\imaginary}{2^n} \sum_{\mu_i \leftrightarrow \nu_i} \sum_{s \in S_n} \pregravfr_n^{\mu_{s(1)} \nu_{s(1)} \vert \cdots \vert \mu_{s(n)} \nu_{s(n)}} \left ( p_{s(1)}^\sigma, \cdots, p_{s(n)}^\sigma \right )
		\intertext{with}
		\begin{split}
			\pregravfr_n^{\mu_1 \nu_1 \vert \cdots \vert \mu_n \nu_n} \left ( p_1^\sigma, \cdots, p_n^\sigma \right ) & = \\ & \hphantom{=} \mkern-120mu - \frac{1}{2 \gcoupling^2} \sum_{m = 1}^n \left ( \mathfrak{R}_m^{\mu_1 \nu_1 \vert \cdots \vert \mu_m \nu_m} \left ( p_1^\sigma, \cdots, p_m^{\sigma_m} \right ) \times \mathfrak{V}_{n - m}^{\mu_{n - m} \nu_{n - m} \vert \cdots \vert \mu_n \nu_n} \right ) \\ & \hphantom{= - \frac{\imaginary}{\gcoupling^2} \sum_{m = 1}^n} \mkern-120mu + \boldsymbol{\delta}_{n = 2} \frac{1}{2 \zeta} \deDonderFR_2^{\mu_1 \nu_1 \vert \mu_2 \nu_2} \left ( p_1^\sigma, p_2^\sigma \right ) \, ,
		\end{split}
	\end{align}
	\end{subequations}
	where \(\boldsymbol{\delta}_{n = 2}\) is set to 1 for \(n = 2\) and to 0 else, and modulo total derivatives which come from the \(R^{\partial \Gamma}\) contributions of degree \(n\).
\end{proof}

\enter

\begin{rem} \label{rem:one-valent-FR}
	The one-valent Feynman rule actually reads
	\begin{equation}
		\gravfr_1^{\mu_1 \nu_1} \left ( p_1^\sigma \right ) = \frac{\imaginary}{2 \gcoupling} \left ( p_1^{\mu_1} p_1^{\nu_1} - p_1^2 \hat{\eta}^{\mu_1 \nu_1} \right ) \, .
	\end{equation}
	However this term comes from a total derivative in the Lagrange density and can thus be set to zero. Equivalently, on the level of Feynman rules, it vanishes due to momentum conservation.
\end{rem}

\enter

\begin{thm} \label{thm:grav-prop}
	Given the situation of \thmref{thm:grav-fr}, the graviton propagator Feynman rule for (effective) Quantum General Relativity reads:
	\begin{equation} \label{eqn:grav-prop}
	\begin{split}
		\gravprop_{\mu_1 \nu_1 \vert \mu_2 \nu_2} \left ( p^\sigma; \zeta; \epsilon \right ) & = - \frac{2 \imaginary}{p^2 + \imaginary \epsilon} \left [ \vphantom{\frac{1}{p^2}} \left ( \hat{\eta}_{\mu_1 \mu_2} \hat{\eta}_{\nu_1 \nu_2} + \hat{\eta}_{\mu_1 \nu_2} \hat{\eta}_{\nu_1 \mu_2} - \hat{\eta}_{\mu_1 \nu_1} \hat{\eta}_{\mu_2 \nu_2} \right ) \right . \\
		& \! \! \! \! \! \! \! \! \! \! \! \! \! \! \! \! \! \hphantom{\frac{2}{p^2 + \imaginary \epsilon}} \left . - \left ( \frac{1 - \zeta}{p^2} \right ) \left ( \hat{\eta}_{\mu_1 \mu_2} p_{\nu_1} p_{\nu_2} + \hat{\eta}_{\mu_1 \nu_2} p_{\nu_1} p_{\mu_2} + \hat{\eta}_{\nu_1 \mu_2} p_{\mu_1} p_{\nu_2} + \hat{\eta}_{\nu_1 \nu_2} p_{\mu_1} p_{\mu_2} \right ) \right ]
	\end{split}
	\end{equation}
\end{thm}

\begin{proof}
	To calculate the graviton propagator, we recall
	\begin{equation} \label{eqn:grav-fr-quadratic}
	\begin{split}
		\gravfr_2^{\mu_1 \nu_1 \vert \mu_2 \nu_2} \left ( p^\sigma; \zeta \right ) & = \frac{\imaginary}{4} \left ( 1 - \frac{1}{\zeta} \right ) \left ( p^{\mu_1} p^{\nu_1} \hat{\eta}^{\mu_2 \nu_2} + p^{\mu_2} p^{\nu_2} \hat{\eta}^{\mu_1 \nu_1} \right ) \\
		& \hphantom{=} \mkern-55mu - \frac{\imaginary}{8} \left ( 1 - \frac{1}{\zeta} \right ) \left ( p^{\mu_1} p^{\mu_2} \hat{\eta}^{\nu_1 \nu_2} + p^{\mu_1} p^{\nu_2} \hat{\eta}^{\nu_1 \mu_2} + p^{\nu_1} p^{\mu_2} \hat{\eta}^{\mu_1 \nu_2} + p^{\nu_1} p^{\nu_2} \hat{\eta}^{\mu_1 \mu_2} \right ) \\
		& \hphantom{=} \mkern-55mu - \frac{\imaginary}{4} \left ( 1 - \frac{1}{2 \zeta} \right ) \left ( p^2 \hat{\eta}^{\mu_1 \nu_1} \hat{\eta}^{\mu_2 \nu_2} \right ) \\
		& \hphantom{=} \mkern-55mu + \frac{\imaginary}{8} \left ( p^2 \hat{\eta}^{\mu_1 \mu_2} \hat{\eta}^{\nu_1 \nu_2} + p^2 \hat{\eta}^{\mu_1 \nu_2} \hat{\eta}^{\nu_1 \mu_2} \right )
	\end{split}
	\end{equation}
	from \thmref{thm:grav-fr} and then invert it to obtain the propagator, i.e.\ such that\footnote{Where we treat the tuples of indices \(\mu_i \nu_i\) as one index, i.e.\ exclude the a priori possible term \(\hat{\eta}^{\mu_1 \nu_1} \hat{\eta}_{\mu_3 \nu_3}\) on the right hand side.}
	\begin{equation}
		\gravfr_2^{\mu_1 \nu_1 \vert \mu_2 \nu_2} \left ( p^\sigma; \zeta \right ) \gravprop_{\mu_2 \nu_2 \vert \mu_3 \nu_3} \left ( p^\sigma; \zeta; 0 \right ) = \frac{1}{2} \left ( \hat{\delta}^{\mu_1}_{\mu_3} \hat{\delta}^{\nu_1}_{\nu_3} + \hat{\delta}^{\mu_1}_{\nu_3} \hat{\delta}^{\nu_1}_{\mu_3} \right )
	\end{equation}
	holds, and we obtain \eqnref{eqn:grav-prop}.
\end{proof}

\enter

\begin{thm} \label{thm:ghost-fr}
	Given the situation of \thmref{thm:grav-fr}, the graviton-ghost \(2\)-point vertex Feynman rule for (effective) Quantum General Relativity reads:
	\begin{equation}
		\ghostfr_2^{\rho_1 \rho_2} \left ( p^\sigma \right ) = \frac{\imaginary}{2 \zeta} p^2 \hat{\eta}^{\rho_1 \rho_2}
	\end{equation}
	Furthermore, the graviton-ghost \(n\)-point vertex Feynman rules with \(n > 2\) for (effective) Quantum General Relativity read   (where \(\preghostfr_n\) denotes the corresponding unsymmetrized Feynman rules):
	\begin{subequations}
	\begin{align}
		\begin{split}
			\ghostfr_n^{\rho_1 \vert \rho_2 \| \mu_3 \nu_3 \vert \cdots \vert \mu_n \nu_n} \left ( p_1^\sigma, \cdots, p_n^\sigma \right ) & = \\
			& \hphantom{=} \mkern-72mu \frac{\imaginary}{2^{n-2}} \sum_{\mu_i \leftrightarrow \nu_i} \sum_{\substack{s \in S_{n-2}\\\tilde{s}(i) := s(i-2)+2}} \preghostfr_n^{\rho_1 \vert \rho_2 \| \mu_{\tilde{s}(3)} \nu_{\tilde{s}(3)} \vert \cdots \vert \mu_{\tilde{s}(n)} \nu_{\tilde{s}(n)}} \left ( p_1^\sigma, \cdots, p_n^\sigma \right )
		\end{split}
		\intertext{with}
		\begin{split}
			\preghostfr_n^{\rho_1 \vert \rho_2 \| \mu_3 \nu_3 \vert \cdots \vert \mu_n \nu_n} \left ( p_1^\sigma, \cdots, p_n^\sigma \right ) & = \frac{\left ( - \gcoupling \right )^{n-2}}{4} \left \{ \Bigg ( \hat{\delta}^{\rho_1}_{\mu_3} \hat{\delta}^{\mu}_{\nu_{n+1}} \prod_{a = 3}^n \hat{\eta}^{\mu_a \nu_{a+1}} \Bigg ) \hat{\eta}^{\rho_2 \nu} \hat{\eta}^{\rho \sigma} \right . \\
			& \hphantom{=} \times \Bigg [ p^{(2)}_\nu \bigg ( p^{(3)}_\rho \hat{\delta}_\mu^{\mu_3} \hat{\delta}_\sigma^{\nu_3} + p^{(3)}_\sigma \hat{\delta}_\rho^{\mu_3} \hat{\delta}_\mu^{\nu_3} - p^{(3)}_\mu \hat{\delta}_\rho^{\mu_3} \hat{\delta}_\sigma^{\nu_3} \bigg ) \\
			& \hphantom{= \times [} \left . - 2 p^{(2)}_\rho \bigg ( p^{(3)}_\sigma \hat{\delta}_\nu^{\mu_3} \hat{\delta}_\mu^{\nu_3} + p^{(3)}_\nu \hat{\delta}_\mu^{\mu_3} \hat{\delta}_\sigma^{\nu_3} - p^{(3)}_\mu \hat{\delta}_\sigma^{\mu_3} \hat{\delta}_\nu^{\nu_3} \bigg ) \Bigg ] \right \} \, ,
		\end{split}
	\end{align}
	\end{subequations}
	where particle \(1\) is the graviton-ghost, particle \(2\) is the graviton-antighost and the other particles are gravitons.
\end{thm}

\begin{proof}
	The case \(n = 2\) is immediate and the cases \(n > 2\) follow directly from Lemmata \ref{lem:traces_FR} and \ref{lem:Christoffel_FR}, since we have for \(n > 2\)
	\begin{equation}
	\begin{split}
		\ghostfr_n^{\rho_1 \vert \rho_2 \| \mu_3 \nu_3 \vert \cdots \vert \mu_n \nu_n} \left ( p_1^\sigma, \cdots, p_n^\sigma \right ) & = \\
		& \hphantom{=} \mkern-72mu \frac{\imaginary}{2^{n-2}} \sum_{\mu_i \leftrightarrow \nu_i} \sum_{\substack{s \in S_{n-2}\\\tilde{s}(i) := s(i-2)+2}} \preghostfr_n^{\rho_1 \vert \rho_2 \| \mu_{\tilde{s}(3)} \nu_{\tilde{s}(3)} \vert \cdots \vert \mu_{\tilde{s}(n)} \nu_{\tilde{s}(n)}} \left ( p_1^\sigma, \cdots, p_n^\sigma \right )
	\end{split}
	\end{equation}
	with
	\begin{multline}
		\preghostfr_n^{\rho_1 \vert \rho_2 \| \mu_3 \nu_3 \vert \cdots \vert \mu_n \nu_n} \left ( p_1^\sigma, \cdots, p_n^\sigma \right ) = \\ \frac{\left ( - 1 \right )^{n}}{2} \mathfrak{h}_{n-3}^{\rho_1 \mu \triplevert \mu_4 \nu_4 \vert \cdots \vert \mu_n \nu_n} \hat{\eta}^{\rho_2 \nu} \hat{\eta}^{\rho \sigma} \left [ \left ( \sum_{\substack{k = 1\\k \neq 2}}^n p_\nu^k \right ) \boldsymbol{\Gamma}^{\mu_3 \nu_3}_{\rho \sigma \mu} \left ( p_1^\sigma \right ) - 2 \left ( \sum_{\substack{k = 1\\k \neq 2}}^n p_\rho^k \right ) \boldsymbol{\Gamma}^{\mu_3 \nu_3}_{\nu \sigma \mu} \left ( p_1^\sigma \right ) \right ] \, ,
	\end{multline}
	and then used momentum conservation twice, i.e.\ the relation
	\begin{equation}
		\left ( \sum_{\substack{k = 1\\k \neq 2}}^n p_\tau^{(k)} \right ) = - p^{(2)}_\tau \, .
	\end{equation}
\end{proof}

\enter

\begin{thm} \label{thm:ghost-prop}
	Given the situation of \thmref{thm:grav-fr}, the graviton-ghost propagator Feynman rule for (effective) Quantum General Relativity reads:
	\begin{equation} \label{eqn:grav-ghost-prop}
		\ghostprop_{\rho_1 \vert \rho_2} \left ( p^2, \epsilon \right ) = - \frac{2 \imaginary \zeta}{p^2 + \imaginary \epsilon} \hat{\eta}_{\rho_1 \rho_2}
	\end{equation}
\end{thm}

\begin{proof}
	To calculate the graviton propagator, we recall
	\begin{equation}
		\ghostfr_2^{\rho_1 \vert \rho_2} \left ( p^\sigma \right ) = \frac{\imaginary}{2 \zeta} p^2 \hat{\eta}^{\rho_1 \rho_2}
	\end{equation}
	from \thmref{thm:ghost-fr} and then invert it to obtain the propagator, i.e.\ such that
	\begin{equation}
		\ghostfr_2^{\rho_1 \vert \rho_2} \left ( p^\sigma \right ) \ghostprop_{\rho_2 \vert \rho_3} \left ( p^2, 0 \right ) = \hat{\delta}^{\rho_1}_{\rho_3}
	\end{equation}
	holds, and we obtain \eqnref{eqn:grav-ghost-prop}.
\end{proof}

\enter

\begin{exmp} \label{exmp:FR}
	Given the situation of \thmref{thm:grav-fr}, the three- and four-valent graviton vertex Feynman rules read as follows:\footnote{We have used momentum conservation, i.e.\ performed a partial integration on the Lagrange density for General Relativity, to obtain a more compact form.}
	\begin{subequations}
	\begin{align}
		& \gravfr_3^{\mu_1 \nu_1 \vert \mu_2 \nu_2 \vert \mu_3 \nu_3} \left ( p_1^\sigma, p_2^\sigma, p_3^\sigma \right ) = \frac{\imaginary}{8} \sum_{\mu_i \leftrightarrow \nu_i} \sum_{s \in S_3} \pregravfr_3^{\mu_{s(1)} \nu_{s(1)} \vert \mu_{s(2)} \nu_{s(2)} \vert \mu_{s(3)} \nu_{s(3)}} \left ( p_{s(1)}^\sigma, p_{s(2)}^\sigma \right )
		\intertext{with}
		\begin{split}
			& \pregravfr_3^{\mu_1 \nu_1 \vert \mu_2 \nu_2 \vert \mu_3 \nu_3} \left ( p_1^\sigma, p_2^\sigma \right ) = \frac{\gcoupling}{4} \Bigg \{ \frac{1}{2} p_1^{\mu_3} p_2^{\nu_3} \hat{\eta}^{\mu_1 \mu_2} \hat{\eta}^{\nu_1 \nu_2} - p_1^{\mu_3} p_2^{\mu_1} \hat{\eta}^{\nu_1 \mu_2} \hat{\eta}^{\nu_2 \nu_3} \\
			& \phantom{\pregravfr_3^{\mu_1 \nu_1 \vert \mu_2 \nu_2 \vert \mu_3 \nu_3} \left ( p_1^\sigma, p_2^\sigma \right ) = \frac{\gcoupling}{4} \Bigg \{} + \left ( p_1 \cdot p_2 \right ) \bigg ( - \frac{1}{2} \hat{\eta}^{\mu_1 \nu_1 } \hat{\eta}^{\mu_2 \mu_3 } \hat{\eta}^{\nu_2 \nu_3 } + \hat{\eta}^{\mu_1 \nu_2 } \hat{\eta}^{\mu_2 \nu_3 } \hat{\eta}^{\mu_3 \nu_1 } \\ & \phantom{\pregravfr_3^{\mu_1 \nu_1 \vert \mu_2 \nu_2 \vert \mu_3 \nu_3} \left ( p_1^\sigma, p_2^\sigma \right ) = \frac{\gcoupling}{4} \Bigg \{ + \left ( p_1 \cdot p_2 \right ) \bigg (} - \frac{1}{4} \hat{\eta}^{\mu_1 \mu_2 } \hat{\eta}^{\nu_1 \nu_2 } \hat{\eta}^{\mu_3 \nu_3 } + \frac{1}{8} \hat{\eta}^{\mu_1 \nu_1 } \hat{\eta}^{\mu_2 \nu_2 } \hat{\eta}^{\mu_3 \nu_3 } \bigg ) \! \Bigg \}
		\end{split}
	\end{align}
	\end{subequations}
	and
	\begin{subequations}
	\begin{align}
		& \gravfr_4^{\mu_1 \nu_1 \vert \cdots \vert \mu_4 \nu_4} \left ( p_1^\sigma, \cdots , p_4^\sigma \right ) = \frac{\imaginary}{16} \sum_{\mu_i \leftrightarrow \nu_i} \sum_{s \in S_4} \pregravfr_4^{\mu_{s(1)} \nu_{s(1)} \vert \cdots \vert \mu_{s(4)} \nu_{s(4)}} \left ( p_{s(1)}^\sigma, p_{s(2)}^\sigma \right )
		\intertext{with}
		\begin{split}
		& \pregravfr_4^{\mu_1 \nu_1 \vert \cdots \vert \mu_4 \nu_4} \left ( p_1^\sigma, p_2^\sigma \right ) = \frac{\gcoupling}{4} \Bigg \{ - p_1^{\mu_3} p_2^{\nu_3} \hat{\eta}^{\mu_1 \mu_2} \hat{\eta}^{\nu_1 \mu_4} \hat{\eta}^{\nu_2 \nu_4} + 2 p_1^{\mu_3} p_2^{\mu_1} \hat{\eta}^{\nu_1 \mu_2} \hat{\eta}^{\nu_2 \mu_4} \hat{\eta}^{\nu_3 \nu_4} \\
		& \phantom{\pregravfr_4^{\mu_1 \nu_1 \vert \cdots \vert \mu_4 \nu_4} \left ( p_1^\sigma, p_2^\sigma \right ) = \frac{\gcoupling}{4} \Bigg \{} - \frac{1}{2} p_1^{\mu_3} p_2^{\mu_1} \hat{\eta}^{\nu_1 \mu_2} \hat{\eta}^{\nu_2 \nu_3} \hat{\eta}^{\mu_4 \nu_4} + p_1^{\mu_3} p_2^{\nu_3} \hat{\eta}^{\mu_1 \mu_2} \hat{\eta}^{\nu_1 \nu_2} \hat{\eta}^{\mu_4 \nu_4} \\
		& \phantom{\pregravfr_4^{\mu_1 \nu_1 \vert \cdots \vert \mu_4 \nu_4} \left ( p_1^\sigma, p_2^\sigma \right ) = \frac{\gcoupling}{4} \Bigg \{} - \frac{1}{2} p_1^{\mu_3} p_2^{\mu_4} \hat{\eta}^{\mu_1 \mu_2} \hat{\eta}^{ \nu_1 \nu_2} \hat{\eta}^{\nu_3 \nu_4} + p_1^{\mu_3} p_2^{\mu_4} \hat{\eta}^{\mu_1 \mu_2} \hat{\eta}^{\nu_1 \nu_3} \hat{\eta}^{\nu_2 \nu_4} \\
		& \phantom{\pregravfr_4^{\mu_1 \nu_1 \vert \cdots \vert \mu_4 \nu_4} \left ( p_1^\sigma, p_2^\sigma \right ) = \frac{\gcoupling}{4} \Bigg \{} - \frac{1}{2} p_1^{\mu_2} p_2^{\mu_3} \hat{\eta}^{\mu_1 \nu_1} \hat{\eta}^{\nu_2 \mu_4} \hat{\eta}^{\nu_3 \nu_4} + \frac{1}{4} p_1^{\mu_2} p_2^{\mu_1} \hat{\eta}^{\nu_1 \nu_2} \hat{\eta}^{\mu_3 \mu_4} \hat{\eta}^{\nu_3 \nu_4} \\
		& \phantom{\pregravfr_4^{\mu_1 \nu_1 \vert \cdots \vert \mu_4 \nu_4} \left ( p_1^\sigma, p_2^\sigma \right ) =} \! \! \! \! + \left ( p_1 \cdot p_2 \right ) \bigg ( - \frac{1}{16} \hat{\eta}^{\mu_1 \mu_2} \hat{\eta}^{\nu_1 \nu_2} \hat{\eta}^{\mu_3 \nu_3} \hat{\eta}^{\mu_4 \nu_4} + \frac{1}{8} \hat{\eta}^{\mu_1 \mu_2} \hat{\eta}^{\nu_1 \nu_2} \hat{\eta}^{\mu_3 \mu_4} \hat{\eta}^{\nu_3 \nu_4} \\
		& \phantom{\pregravfr_4^{\mu_1 \nu_1 \vert \cdots \vert \mu_4 \nu_4} \left ( p_1^\sigma, p_2^\sigma \right ) = + \left ( p_1 \cdot p_2 \right ) \bigg (} \! \! \! \! + \frac{1}{2} \hat{\eta}^{\mu_1 \mu_2} \hat{\eta}^{\nu_1 \mu_3} \hat{\eta}^{\nu_2 \nu_3} \hat{\eta}^{\mu_4 \nu_4} - \hat{\eta}^{\mu_1 \mu_2} \hat{\eta}^{\nu_1 \mu_3} \hat{\eta}^{\nu_2 \mu_4} \hat{\eta}^{\nu_3 \nu_4} \\
		& \phantom{\pregravfr_4^{\mu_1 \nu_1 \vert \cdots \vert \mu_4 \nu_4} \left ( p_1^\sigma, p_2^\sigma \right ) = + \left ( p_1 \cdot p_2 \right ) \bigg (} \! \! \! \! + \frac{1}{2} \hat{\eta}^{\mu_1 \mu_3} \hat{\eta}^{\nu_1 \nu_3} \hat{\eta}^{\mu_2 \mu_4} \hat{\eta}^{\nu_2 \nu_4} - \frac{1}{2} \hat{\eta}^{\mu_1 \mu_3} \hat{\eta}^{\nu_1 \mu_4} \hat{\eta}^{\mu_2 \nu_3} \hat{\eta}^{\nu_2 \nu_4} \\
		& \phantom{\pregravfr_4^{\mu_1 \nu_1 \vert \cdots \vert \mu_4 \nu_4} \left ( p_1^\sigma, p_2^\sigma \right ) = + \left ( p_1 \cdot p_2 \right ) \bigg (} \! \! \! \! - \frac{1}{4} \hat{\eta}^{\mu_1 \nu_1} \hat{\eta}^{\mu_2 \mu_3} \hat{\eta}^{\nu_2 \nu_3} \hat{\eta}^{\mu_4 \nu_4} + \frac{1}{2} \hat{\eta}^{\mu_1 \nu_1} \hat{\eta}^{\mu_2 \mu_3} \hat{\eta}^{\nu_2 \mu_4} \hat{\eta}^{\nu_3 \nu_4} \\
		& \phantom{\pregravfr_4^{\mu_1 \nu_1 \vert \cdots \vert \mu_4 \nu_4} \left ( p_1^\sigma, p_2^\sigma \right ) = + \left ( p_1 \cdot p_2 \right ) \bigg (} \! \! \! \! + \frac{1}{32} \hat{\eta}^{\mu_1 \nu_1} \hat{\eta}^{\mu_2 \nu_2} \hat{\eta}^{\mu_3 \nu_3} \hat{\eta}^{\mu_4 \nu_4} - \frac{1}{8} \hat{\eta}^{\mu_1 \nu_1} \hat{\eta}^{\mu_2 \nu_2} \hat{\eta}^{\mu_3 \mu_4} \hat{\eta}^{\nu_3 \nu_4} \bigg ) \Bigg \}
		\end{split}
	\end{align}
	\end{subequations}
	We remark that the three- and four-valent graviton vertex Feynman rules agree with the cited literature modulo prefactors and minus signs. Additionally, given the situation of \thmref{thm:ghost-fr}, the three- and four-valent graviton-ghost vertex Feynman rules read as follows:
	\begin{equation}
	\begin{split}
		\ghostfr_3^{\rho_1 \vert \rho_2 \| \mu_3 \nu_3} \left ( p_2^\sigma, p_3^\sigma \right ) & = \frac{\imaginary \gcoupling}{4} \Bigg \{ - p_2^{\rho_2} \bigg ( p_3^{\mu_3} \hat{\eta}^{\rho_1 \nu_3} + p_3^{\nu_3} \hat{\eta}^{\rho_1 \mu_3} - p_3^{\rho_1} \hat{\eta}^{\mu_3 \nu_3} \bigg ) \\
		& \phantom{= \frac{\imaginary \gcoupling}{8} \Bigg \{} - p_3^{\rho_1} \bigg ( p_2^{\mu_3} \hat{\eta}^{\rho_2 \nu_3} + p_2^{\nu_3} \hat{\eta}^{\rho_2 \mu_3} \bigg ) + p_3^{\rho_2} \bigg ( p_2^{\mu_3} \hat{\eta}^{\rho_1 \nu_3} + p_2^{\nu_3} \hat{\eta}^{\rho_1 \mu_3} \bigg ) \\
		& \phantom{= \frac{\imaginary \gcoupling}{8} \Bigg \{} + \left ( p_2 \cdot p_3 \right ) \bigg ( \hat{\eta}^{\rho_1 \mu_3} \hat{\eta}^{\rho_2 \nu_3} + \hat{\eta}^{\rho_1 \nu_3} \hat{\eta}^{\rho_2 \mu_3} \bigg ) \Bigg \}
	\end{split}
	\end{equation}
	and
	\begin{equation}
	\begin{split}
		& \ghostfr_4^{\rho_1 \vert \rho_2 \| \mu_3 \nu_3 \vert \mu_4 \nu_4} \left ( p_2^\sigma, p_3^\sigma, p_4^\sigma \right ) = \\
		& \phantom{=} \frac{\imaginary \gcoupling^2}{8} \Bigg \{ p_2^{\rho_2} \bigg ( p_3^{\nu_3} \hat{\eta}^{\rho_1 \mu_4} \hat{\eta}^{\mu_3 \nu_4} + p_3^{\mu_3} \hat{\eta}^{\rho_1 \mu_4} \hat{\eta}^{\nu_3 \nu_4} + p_3^{\nu_3} \hat{\eta}^{\rho_1 \nu_4} \hat{\eta}^{\mu_3 \mu_4} \\
		& \phantom{\frac{\imaginary \gcoupling^2}{16} \Bigg \{ \, p_2^{\rho_2} \bigg (} + p_3^{\mu_3} \hat{\eta}^{\rho_1 \nu_4} \hat{\eta}^{\nu_3 \mu_4}	- p_3^{\mu_4} \hat{\eta}^{\rho_1 \nu_4} \hat{\eta}^{\mu_3 \nu_3} - p_3^{\nu_4} \hat{\eta}^{\rho_1 \mu_4} \hat{\eta}^{\mu_3 \nu_3} \bigg ) \\
		&  \phantom{\frac{\imaginary \gcoupling^2}{16} \Bigg \{} + p_2^{\rho_2} \bigg ( p_4^{\nu_4} \hat{\eta}^{\rho_1 \mu_3} \hat{\eta}^{\mu_4 \nu_3} + p_4^{\mu_4} \hat{\eta}^{\rho_1 \mu_3} \hat{\eta}^{\nu_4 \nu_3} + p_4^{\nu_4} \hat{\eta}^{\rho_1 \nu_3} \hat{\eta}^{\mu_4 \mu_3} \\
		& \phantom{\frac{\imaginary \gcoupling^2}{16} \Bigg \{ \, p_2^{\rho_2} \bigg (} + p_4^{\mu_4} \hat{\eta}^{\rho_1 \nu_3} \hat{\eta}^{\nu_4 \mu_3}	- p_4^{\mu_3} \hat{\eta}^{\rho_1 \nu_3} \hat{\eta}^{\mu_4 \nu_4} - p_4^{\nu_3} \hat{\eta}^{\rho_1 \mu_3} \hat{\eta}^{\mu_4 \nu_4} \bigg ) \\
		& \phantom{\frac{\imaginary \gcoupling^2}{16} \Bigg \{} - p_3^{\rho_2} \bigg ( p_2^{\mu_3} \hat{\eta}^{\rho_1 \mu_4} \hat{\eta}^{\nu_3 \nu_4} + p_2^{\nu_3} \hat{\eta}^{\rho_1 \mu_4} \hat{\eta}^{\mu_3 \nu_4} + p_2^{\mu_3} \hat{\eta}^{\rho_1 \nu_4} \hat{\eta}^{\nu_3 \mu_4} + p_2^{\nu_3} \hat{\eta}^{\rho_1 \nu_4} \hat{\eta}^{\mu_3 \mu_4} \bigg ) \\
		& \phantom{\frac{\imaginary \gcoupling^2}{16} \Bigg \{} - p_4^{\rho_2} \bigg ( p_2^{\mu_4} \hat{\eta}^{\rho_1 \mu_3} \hat{\eta}^{\nu_4 \nu_3} + p_2^{\nu_4} \hat{\eta}^{\rho_1 \mu_3} \hat{\eta}^{\mu_4 \nu_3} + p_2^{\mu_4} \hat{\eta}^{\rho_1 \nu_3} \hat{\eta}^{\nu_4 \mu_3} + p_2^{\nu_4} \hat{\eta}^{\rho_1 \nu_3} \hat{\eta}^{\mu_4 \mu_3} \bigg ) \\
		& \phantom{\frac{\imaginary \gcoupling^2}{16} \Bigg \{} + p_2^{\mu_3} p_3^{\mu_4} \hat{\eta}^{\rho_1 \nu_4} \hat{\eta}^{\rho_2 \nu_3} + p_2^{\nu_3} p_3^{\mu_4} \hat{\eta}^{\rho_1 \nu_4} \hat{\eta}^{\rho_2 \mu_3} + p_2^{\mu_3} p_3^{\nu_4} \hat{\eta}^{\rho_1 \mu_4} \hat{\eta}^{\rho_2 \nu_3} + p_2^{\nu_3} p_3^{\nu_4} \hat{\eta}^{\rho_1 \mu_4} \hat{\eta}^{\rho_2 \mu_3} \\
		& \phantom{\frac{\imaginary \gcoupling^2}{16} \Bigg \{} + p_2^{\mu_4} p_4^{\mu_3} \hat{\eta}^{\rho_1 \nu_3} \hat{\eta}^{\rho_2 \nu_4} + p_2^{\mu_4} p_4^{\nu_3} \hat{\eta}^{\rho_1 \mu_3} \hat{\eta}^{\rho_2 \nu_4} + p_2^{\nu_4} p_4^{\mu_3} \hat{\eta}^{\rho_1 \nu_3} \hat{\eta}^{\rho_2 \mu_4} + p_2^{\nu_4} p_4^{\nu_3} \hat{\eta}^{\rho_1 \mu_3} \hat{\eta}^{\rho_2 \mu_4} \\
		& \phantom{\frac{\imaginary \gcoupling^2}{16} \Bigg \{} - \left ( p_2 \cdot p_3 \right ) \bigg ( \hat{\eta}^{\rho_1 \mu_4} \hat{\eta}^{\rho_2 \mu_3} \hat{\eta}^{\nu_3 \nu_4} +  \hat{\eta}^{\rho_1 \mu_4} \hat{\eta}^{\rho_2 \nu_3} \hat{\eta}^{\mu_3 \nu_4} \\
		& \phantom{\frac{\imaginary \gcoupling^2}{16} \Bigg \{ - \left ( p_2 \cdot p_3 \right ) \bigg (} + \hat{\eta}^{\rho_1 \nu_4} \hat{\eta}^{\rho_2 \mu_3} \hat{\eta}^{\nu_3 \mu_4} +  \hat{\eta}^{\rho_1 \nu_4} \hat{\eta}^{\rho_2 \nu_3} \hat{\eta}^{\mu_3 \mu_4} \bigg ) \\
		& \phantom{\frac{\imaginary \gcoupling^2}{16} \Bigg \{} - \left ( p_2 \cdot p_4 \right ) \bigg ( \hat{\eta}^{\rho_1 \mu_3} \hat{\eta}^{\rho_2 \mu_4} \hat{\eta}^{\nu_4 \nu_3} +  \hat{\eta}^{\rho_1 \mu_3} \hat{\eta}^{\rho_2 \nu_4} \hat{\eta}^{\mu_4 \nu_3} \\
		& \phantom{\frac{\imaginary \gcoupling^2}{16} \Bigg \{ - \left ( p_2 \cdot p_4 \right ) \bigg (} +  \hat{\eta}^{\rho_1 \nu_3} \hat{\eta}^{\rho_2 \mu_4} \hat{\eta}^{\nu_4 \mu_3} +  \hat{\eta}^{\rho_1 \nu_3} \hat{\eta}^{\rho_2 \nu_4} \hat{\eta}^{\mu_4 \mu_3} \bigg ) \Bigg \}
	\end{split}
	\end{equation}
\end{exmp}

\subsection{Feynman rules for gravitons and matter}

Having done all preparations in \ssecref{ssec:preparations_gravitons-matter}, we now list the corresponding Feynman rules for the interactions of gravitons with matter from the Standard Model. To this end we state the Feynman rules for the interactions according to the classification of \lemref{lem:matter-model-Lagrange-densities} and refer for the corresponding matter contributions to \cite{Romao_Silva} in order to keep this article at a reasonable length.

\enter

\begin{thm} \label{thm:matter-fr}
	Given the situation of \thmref{thm:grav-fr} and the matter-model Lagrange densities from \lemref{lem:matter-model-Lagrange-densities}, the graviton-matter \(n\)-point vertex Feynman rule for (effective) Quantum General Relativity coupled to the matter-model Lagrange density of type \(k\) reads:
	\begin{equation}
	\matterfrk_n^{\mu_1 \nu_1 \vert \cdots \vert \mu_n \nu_n} \left ( p_1^\sigma, \cdots, p_n^\sigma \right ) = \frac{\imaginary}{2^n} \sum_{\mu_i \leftrightarrow \nu_i} \sum_{s \in S_n} \prematterfrk_n^{\mu_{s(1)} \nu_{s(1)} \vert \cdots \vert \mu_{s(n)} \nu_{s(n)}} \left ( p_{s(1)}^\sigma, \cdots, p_{s(n)}^\sigma \right )
	\end{equation}
	with
	{\allowdisplaybreaks
	\begin{align}
		\prematterfri_n^{\mu_1 \nu_1 \vert \cdots \vert \mu_n \nu_n} \left ( \tensor[_1]{\! \widehat{T}}{} \right ) & = \tensor[_1]{\! \widehat{T}}{} \mathfrak{v}_n^{\mu_1 \nu_1 \vert \cdots \vert \mu_n \nu_n} \, , \\
		\begin{split}
			\prematterfrii_n^{\mu_1 \nu_1 \vert \cdots \vert \mu_n \nu_n} \left ( \tensor[_2]{\! \widehat{T}}{} \right ) & = \tensor[_2]{\! \widehat{T}}{_\mu _\nu} \sum_{m_1 + m_2 = n} \left ( -1 \right )^{m_1} \mathfrak{h}_{m_1}^{\mu \nu \triplevert \mu_1 \nu_1 \vert \cdots \vert \mu_{m_1} \nu_{m_1}} \\ & \hphantom{=} \times \mathfrak{v}_{m_2}^{\mu_{{m_1} + 1} \nu_{{m_1} + 1} \vert \cdots \vert \mu_n \nu_n} \, ,
		\end{split}
		\\
		\begin{split}
			\prematterfriii_n^{\mu_1 \nu_1 \vert \cdots \vert \mu_n \nu_n} \left ( \tensor[_3]{\! \widehat{T}}{} \right ) & = \tensor[_3]{\! \widehat{T}}{_\mu _\nu _\rho _\sigma} \sum_{\subalign{m_1 & + m_2 \\ & + m_3 = n}} \left ( - 1 \right )^{m_1 + m_2} \mathfrak{h}_{m_1}^{\mu \nu \triplevert \mu_1 \nu_1 \vert \cdots \vert \mu_{m_1} \nu_{m_1}} \\ & \hphantom{=} \times \mathfrak{h}_{m_2}^{\rho \sigma \triplevert \mu_{{m_1} + 1} \nu_{{m_1} + 1} \vert \cdots \vert \mu_{{m_1} + {m_2}} \nu_{{m_1} + {m_2}}} \\ & \hphantom{=} \times \mathfrak{v}_{m_3}^{\mu_{{m_1} + {m_2} + 1} \nu_{{m_1} + {m_2} + 1} \vert \cdots \vert \mu_n \nu_n} \, ,
		\end{split}
		\\
		\begin{split}
			\prematterfriv_n^{\mu_1 \nu_1 \vert \cdots \vert \mu_n \nu_n} \left ( \tensor[_4]{\! \widehat{T}}{}; p_1^\sigma \right ) & = \tensor[_4]{\! \widehat{T}}{_\rho} \boldsymbol{\Gamma}^{\mu_1 \nu_1}_{\mu \nu \sigma} \left ( p_1^\sigma \right ) \sum_{\subalign{m_1 & + m_2 \\ & + m_3 = n - 1}} \left ( - 1 \right )^{m_1 + m_2} \\ & \hphantom{=} \times \mathfrak{h}_{m_1}^{\mu \nu \triplevert \mu_2 \nu_2 \vert \cdots \vert \mu_{m_1 + 1} \nu_{m_1 + 1}} \\ & \hphantom{=} \times \mathfrak{h}_{m_2}^{\rho \sigma \triplevert \mu_{{m_1} + 2} \nu_{{m_1} + 2} \vert \cdots \vert \mu_{{m_1} + {m_2} + 1} \nu_{{m_1} + {m_2} + 1}} \\ & \hphantom{=} \times \mathfrak{v}_{m_3}^{\mu_{{m_1} + {m_2} + 2} \nu_{{m_1} + {m_2} + 2} \vert \cdots \vert \mu_n \nu_n} \, ,
		\end{split}
		\\
		\begin{split}
			\prematterfrv_n^{\mu_1 \nu_1 \vert \cdots \vert \mu_n \nu_n} \left ( \tensor[_5]{\! \widehat{T}}{}; p_1^\sigma \right ) & = \tensor[_5]{\! \widehat{T}}{_\rho _\sigma _\kappa} \boldsymbol{\Gamma}^{\mu_1 \nu_1}_{\mu \nu \lambda} \left ( p_1^\sigma \right ) \sum_{\subalign{m_1 & + m_2 + m_3 \\ & + m_4 = n - 1}} \left ( - 1 \right )^{m_1 + m_2 + m_3} \\ & \hphantom{=} \times \mathfrak{h}_{m_1}^{\mu \nu \triplevert \mu_2 \nu_2 \vert \cdots \vert \mu_{m_1 + 1} \nu_{m_1 + 1}} \\ & \hphantom{=} \times \mathfrak{h}_{m_2}^{\rho \sigma \triplevert \mu_{{m_1} + 2} \nu_{{m_1} + 2} \vert \cdots \vert \mu_{{m_1} + {m_2} + 1} \nu_{{m_1} + {m_2} + 1}} \\ & \hphantom{=} \times \mathfrak{h}_{m_3}^{\kappa \lambda \triplevert \mu_{{m_1} + {m_2} + 2} \nu_{{m_1} + {m_2} + 2} \vert \cdots \vert \mu_{{m_1} + {m_2} + {m_3} + 1} \nu_{{m_1} + {m_2} + {m_3} + 1}} \\ & \hphantom{=} \times \mathfrak{v}_{m_4}^{\mu_{{m_1} + {m_2} + {m_3} + 2} \nu_{{m_1} + {m_2} + {m_3} + 2} \vert \cdots \vert \mu_n \nu_n} \, ,
		\end{split}
		\\
		\begin{split}
			\prematterfrvi_n^{\mu_1 \nu_1 \vert \cdots \vert \mu_n \nu_n} \left ( \tensor[_6]{\! \widehat{T}}{}; p_1^\sigma, p_2^\sigma \right ) & = \tensor[_6]{\! \widehat{T}}{_\kappa _\iota} \boldsymbol{\Gamma}^{\mu_1 \nu_1}_{\mu \nu \lambda} \left ( p_1^\sigma \right ) \boldsymbol{\Gamma}^{\mu_2 \nu_2}_{\rho \sigma \tau} \left ( p_2^\sigma \right ) \sum_{\subalign{m_1 & + m_2 + m_3 \\ & + m_4 + m_5 = n - 2}} \\ & \hphantom{=} \times \left ( - 1 \right )^{m_1 + m_2 + m_3 + m_4} \mathfrak{h}_{m_1}^{\mu \nu \triplevert \mu_3 \nu_3 \vert \cdots \vert \mu_{m_1 + 2} \nu_{m_1 + 2}} \\ & \hphantom{=} \times \mathfrak{h}_{m_2}^{\rho \sigma \triplevert \mu_{{m_1} + {m_2} + 3} \nu_{{m_1} + {m_2} + 3} \vert \cdots \vert \mu_{{m_1} + {m_2} + 2} \nu_{{m_1} + {m_2} + 2}} \\ & \hphantom{=} \times \mathfrak{h}_{m_3}^{\kappa \lambda \triplevert \mu_{{m_1} + 1} \nu_{{m_1} + 1} \vert \cdots \vert \mu_{{m_1} + {m_2}} \nu_{{m_1} + {m_2}}} \\ & \hphantom{=} \times \mathfrak{h}_{m_4}^{\iota \tau \triplevert \mu_{{m_1} + 1} \nu_{{m_1} + 1} \vert \cdots \vert \mu_{{m_1} + {m_2}} \nu_{{m_1} + {m_2}}} \mathfrak{v}_{m_5}^{\mu_{m + 1} \nu_{m + 1} \vert \cdots \vert \mu_n \nu_n} \, ,
		\end{split}
		\\
		\begin{split}
			\prematterfrvii_n^{\mu_1 \nu_1 \vert \cdots \vert \mu_n \nu_n} \left ( \tensor[_7]{\! \widehat{T}}{} \right ) & = \tensor[_7]{\! \widehat{T}}{_o} \sum_{m_1 + m_2 = n} \binom{\frac{1}{2}}{m_1} \hat{\eta}_{0 \upsilon} \mathfrak{h}_{m_1}^{\upsilon o \triplevert \mu_1 \nu_1 \vert \cdots \vert \mu_{m_1} \nu_{m_1}} \\ & \hphantom{=} \times \mathfrak{v}_{m_2}^{\mu_{{m_1} + 1} \nu_{{m_1} + 1} \vert \cdots \vert \mu_n \nu_n} \, ,
		\end{split}
		\\
		\begin{split}
			\prematterfrviii_n^{\mu_1 \nu_1 \vert \cdots \vert \mu_n \nu_n} \left ( \tensor[_8]{\! \widehat{T}}{} \right ) & = \tensor[_8]{\! \widehat{T}}{_o _\rho _r} \sum_{\subalign{m_1 & + m_2 \\ & + m_3 = n}} \binom{\frac{1}{2}}{m_1} \binom{- \frac{1}{2}}{m_2} \hat{\eta}_{0 \upsilon} \mathfrak{h}_{m_1}^{\upsilon o \triplevert \mu_1 \nu_1 \vert \cdots \vert \mu_{m_1} \nu_{m_1}} \\ & \hphantom{=} \times \mathfrak{h}_{m_2}^{\rho r \triplevert \mu_{m_1 + 1} \nu_{m_1 + 1} \vert \cdots \vert \mu_{m_1 + m_2} \nu_{m_1 + m_2}} \\ & \hphantom{=} \times \mathfrak{v}_{m_3}^{\mu_{m_1 + m_2 + 1} \nu_{m_1 + m_2 + 1} \vert \cdots \vert \mu_n \nu_n} \, ,
		\end{split}
		\\
		\begin{split}
			\prematterfrix_n^{\mu_1 \nu_1 \vert \cdots \vert \mu_n \nu_n} \left ( \tensor[_9]{\! \widehat{T}}{}; p_1^\sigma, \cdots, p_n^\sigma \right ) & = \tensor[_9]{\! \widehat{T}}{_o _r _s _t} \sum_{\subalign{m_1 & + m_2 + m_3 \\ & + m_4 + m_5 = n}} \binom{\frac{1}{2}}{m_1} \binom{- \frac{1}{2}}{m_2} \binom{- \frac{1}{2}}{m_3} \binom{\frac{1}{2}}{m_4} \\ & \hphantom{=} \mkern-136mu \times \hat{\eta}_{0 \upsilon} \mathfrak{h}_{m_1}^{\upsilon o \triplevert \mu_1 \nu_1 \vert \cdots \vert \mu_{m_1} \nu_{m_1}} \\ & \hphantom{=} \mkern-136mu \times \mathfrak{h}_{m_2}^{\rho r \triplevert \mu_{m_1 + 1} \nu_{m_1 + 1} \vert \cdots \vert \mu_{m_1 + m_2} \nu_{m_1 + m_2}} \\ & \hphantom{=} \mkern-136mu \times \mathfrak{h}_{m_3}^{\sigma s \triplevert \mu_{m_1 + m_2 + 1} \nu_{m_1 + m_2 + 1} \vert \cdots \vert \mu_{m_1 + m_2 + m_3} \nu_{m_1 + m_2 + m_3}} \\ & \hphantom{=} \mkern-136mu \times \hat{\eta}_{\sigma \tau} \left ( \mathfrak{h}_{m_4}^\prime \right )_\rho^{\tau t \triplevert \mu_{m_1 + m_2 + m_3 + 1} \nu_{m_1 + m_2 + m_3 + 1} \vert \cdots \vert \mu_{m_1 + m_2 + m_3 + m_4} \nu_{m_1 + m_2 + m_3 + m_4}} \\ & \hphantom{\times \hat{\eta}_{\sigma \tau} \left ( \mathfrak{h}_{m_4}^\prime \right )} \mkern-136mu \left ( p_{m_1 + m_2 + m_3 + 1}^{\sigma_{m_1 + m_2 + m_3 + 1}}, \cdots, p_{m_1 + m_2 + m_3 + m_4}^{\sigma_{m_1 + m_2 + m_3 + m_4}} \right ) \\ & \hphantom{=} \mkern-136mu \times \mathfrak{v}_{m_5}^{\mu_{m_1 + m_2 + m_3 + m_4 + 1} \nu_{m_1 + m_2 + m_3 + m_4 + 1} \vert \cdots \vert \mu_n \nu_n} \, ,
		\end{split}
		\intertext{and}
		\begin{split}
			\prematterfrx_n^{\mu_1 \nu_1 \vert \cdots \vert \mu_n \nu_n} \left ( \tensor[_{10}]{\! \widehat{T}}{}; p_1^\sigma \right ) & = \tensor[_{10}]{\! \widehat{T}}{_o _r _s _t}  \boldsymbol{\Gamma}^{\mu_1 \nu_1}_{\rho \sigma \tau} \left ( p_1^\sigma \right ) \sum_{\subalign{m_1 & + m_2 + m_3 \\ & + m_4 + m_5 = n}} \binom{\frac{1}{2}}{m_1} \binom{- \frac{1}{2}}{m_2} \binom{- \frac{1}{2}}{m_3} \\ & \hphantom{=} \mkern-100mu \times \binom{- \frac{1}{2}}{m_4} \hat{\eta}_{0 \upsilon} \mathfrak{h}_{m_1}^{\upsilon o \triplevert \mu_2 \nu_2 \vert \cdots \vert \mu_{m_1 + 1} \nu_{m_1 + 1}} \\ & \hphantom{=} \mkern-100mu \times \mathfrak{h}_{m_2}^{\rho r \triplevert \mu_{m_1 + 2} \nu_{m_1 + 2} \vert \cdots \vert \mu_{m_1 + m_2 + 1} \nu_{m_1 + m_2 + 1}} \\ & \hphantom{=} \mkern-100mu \times \mathfrak{h}_{m_3}^{\sigma s \triplevert \mu_{m_1 + m_2 + 2} \nu_{m_1 + m_2 + 2} \vert \cdots \vert \mu_{m_1 + m_2 + m_3} \nu_{m_1 + m_2 + m_3}} \\ & \hphantom{=} \mkern-100mu \times \mathfrak{h}_{m_4}^{\tau t \triplevert \mu_{m_1 + m_2 + m_3 + 2} \nu_{m_1 + m_2 + m_3 + 2} \vert \cdots \vert \mu_{m_1 + m_2 + m_3 + m_4 + 1} \nu_{m_1 + m_2 + m_3 + m_4 + 1}} \\ & \hphantom{=} \mkern-100mu \times \mathfrak{v}_{m_5}^{\mu_{m_1 + m_2 + m_3 + m_4 + 2} \nu_{m_1 + m_2 + m_3 + m_4 + 2} \vert \cdots \vert \mu_n \nu_n} \, .
		\end{split}
	\end{align}
	}
\end{thm}

\begin{proof}
	This follows directly from \colref{col:inverse_metric_vielbeins_FR} with Lemmata~\ref{lem:Christoffel_FR}~and~\ref{lem:Riemannian_volume_form_FR}.
\end{proof}

\section{Conclusion} \label{sec:conclusion}

We have derived and presented the Feynman rules for (effective) Quantum General Relativity and the gravitational couplings to the Standard Model. The main results are \thmref{thm:grav-fr} stating the graviton vertex Feynman rules, \thmref{thm:grav-prop} stating the corresponding graviton propagator Feynman rule, \thmref{thm:ghost-fr} stating the graviton-ghost vertex Feynman rules and \thmref{thm:ghost-prop} stating the corresponding graviton-ghost propagator Feynman rule. Additionally, the graviton-matter vertex Feynman rules are stated in \thmref{thm:matter-fr} on the level of 10 generic matter-model Lagrange densities, as classified by \lemref{lem:matter-model-Lagrange-densities}. The complete graviton-matter Feynman rules can then be obtained by adding the corresponding matter contributions, as listed e.g.\ in \cite{Romao_Silva}. Finally, we display the three- and four-valent graviton and graviton-ghost vertex Feynman rules explicitly in \exref{exmp:FR}. In future work, we want to study the BRST double complex for (effective) Quantum General Relativity coupled to the Standard Model in \cite{Prinz_5} and the corresponding longitudinal and transversal structures in \cite{Prinz_6}. Furthermore, we study the appropriate setup for a generalization of Wigner's elementary particle classification to Linearized General Relativity in \cite{Prinz_7}. Moreover, we have studied the renormalization properties of gauge theories and gravity from a Hopf algebraic perspective in \cite{Prinz_2,Prinz_3}. The gravitational Ward identities will be checked in future work as well, with the aim to construct the corresponding cancellation identities. This than leads to the possibility of deriving the corresponding Corolla polynomial \cite{Kreimer_Yeats,Kreimer_Sars_vSuijlekom,Sars,Prinz_1,Kreimer_Corolla,Berghoff_Knispel}, which would relate gravitational amplitudes to the amplitudes of scalar \(\phi^3_4\)-theory.

\section*{Acknowledgments}
\addcontentsline{toc}{section}{Acknowledgments}

The author thanks Axel Kleinschmidt for a clarifying discussion on terminology, Jan Plefka for pointing out further references and Stavros Mougiakakos for useful comments on the first version of the preprint. This research is supported by the International Max Planck Research School for Mathematical and Physical Aspects of Gravitation, Cosmology and Quantum Field Theory.

\bibliography{References}{}
\bibliographystyle{babunsrt}

\end{document}